\documentclass[11pt,letterpaper]{article}
\usepackage[utf8]{inputenc}
\usepackage[rflt]{floatflt}
\usepackage{fullpage}
\usepackage[T1]{fontenc}
\usepackage{graphicx}
\usepackage{epsfig}
\usepackage{color}
\usepackage{enumerate}
\usepackage{verbatim}
\usepackage{amsthm, amssymb}
\usepackage{amsmath}
\usepackage{thmtools,thm-restate}
\usepackage{wrapfig}
\usepackage{xspace}

\definecolor{darkgreen}{rgb}{0,0.5,0}
\usepackage{hyperref}
\hypersetup{
    unicode=false,          
    colorlinks=true,        
    linkcolor=blue,          
    citecolor=darkgreen,        
    filecolor=magenta,      
    urlcolor=cyan           
}

\usepackage{algorithm}
\usepackage{algorithmicx}
\usepackage[noend]{algpseudocode}
\usepackage[disableredefinitions,roman]{complexity}

\usepackage[capitalize, nameinlink]{cleveref}
\crefname{theorem}{Theorem}{Theorems}
\Crefname{lemma}{Lemma}{Lemmas}
\Crefname{claim}{Claim}{Claims}
\Crefname{observation}{Observation}{Observations}

\algblock{Begin}{End}

\date{}
\MakeRobust{\overrightarrow}

\newcommand{\eqdef}{\stackrel{\text{\tiny\rm def}}{=}}

\newtheorem{theorem}{Theorem}
\newtheorem{lemma}[theorem]{Lemma}

\newtheorem{observation}[theorem]{Observation}
\newtheorem{definition}[theorem]{Definition}
\newtheorem{corollary}[theorem]{Corollary}

\newcommand{\Gepsj}[1]{G_{\eps, \sigma_{#1}}}
\newcommand{\GepsZero}[1]{G_{#1, 0}}

\makeatletter
\newcounter{algorithmicH}
\let\oldalgorithmic\algorithmic
\renewcommand{\algorithmic}{%
  \stepcounter{algorithmicH}
  \oldalgorithmic}
\renewcommand{\theHALG@line}{ALG@line.\thealgorithmicH.\arabic{ALG@line}}
\makeatother

\def\polylog{\operatorname{polylog}}

\newcommand{\machines}{M}

\def\dist{\mathit{dist}}

\newcommand{\out}[1]{}

\newcommand{\rb}[1]{\left( #1 \right)}

\newcommand{\ee}[1]{{\mathbb E} \left[ #1 \right]}

\newcommand{\prob}[1]{\mathbb P \left[ #1 \right]}
\newcommand{\eps}{\epsilon}

\newcommand{\tO}{\tilde{O}}
\newcommand{\allones}{\vec{1}}

\newcommand{\BipartitenessTester}{\textsc{BipartitenessTester}\xspace}
\newcommand{\ExpansionTester}{\textsc{ExpansionTester}\xspace}
\newcommand{\RandomWalks}{\textsc{RandomWalks}\xspace}
\newcommand{\PageRankOfBalancedGraphs}{\textsc{PageRankOfBalancedGraphs}\xspace}
\newcommand{\PageRankOfGeneralGraphs}{\textsc{PageRankOfGeneralGraphs}\xspace}
\newcommand{\PageRankLargerDamping}{\textsc{PageRankLargerDampingFactor}\xspace}

\newcommand{\TranslateWalkBase}{\textsc{TranslateWalk}\xspace}
\newcommand{\TranslateWalk}{\textsc{TranslateWalk-$\sigma$}\xspace}
\newcommand{\TranslateWalkEps}{\textsc{TranslateWalk-$\eps$}\xspace}
\newcommand{\TranslateBalancedToPageRankWalk}{\textsc{TranslateToPageRankWalk}\xspace}
\newcommand{\StationaryDistribution}{\textsc{StationaryDistribution}\xspace}
\newcommand{\NumberingSublists}{\textsc{NumberingSublists}\xspace}
\newcommand{\Predecessor}{\textsc{Predecessor}\xspace}

\newcommand{\accept}{\textsc{Accept}\xspace}
\newcommand{\reject}{\textsc{Reject}\xspace}

\newcommand{\fail}{\textsc{fail}\xspace}
\newcommand{\succeed}{\textsc{succeed}\xspace}

\newcommand{\davg}{d_{\rm avg}}
\newcommand{\dmax}{d_{\rm max}}

\newcommand{\bbR}{\mathbb{R}}
\newcommand{\bbN}{\mathbb{N}}
\newcommand{\tpi}{\tilde{\pi}}
\newcommand{\alphastar}{\alpha^\star}

\newcommand{\ltwo}[1]{\|#1\|}

\newcommand{\rw}[1]{\textsc{RW}(#1)}
\newcommand{\undir}[1]{\bar{#1}}

\newcommand{\cycleConj}{\textsc{1-vs.-2-Cycles}}

\title{
Walking Randomly, Massively, and Efficiently
}
\author{
Jakub Łącki\\
Google Research
\and
Slobodan Mitrović\\
MIT
\and
Krzysztof Onak\\
IBM Research
\and
Piotr Sankowski\\
University of Warsaw
}

\begin{document}

\begin{titlepage}
\def\thepage{}
\thispagestyle{empty}
\maketitle

\begin{abstract}
We introduce a set of techniques that allow for efficiently generating many independent random walks in the Massive Parallel Computation (MPC) model with space per machine strongly sublinear in the number of vertices. In this space-per-machine regime, many natural approaches to graph problems struggle to
overcome the $\Theta(\log n)$ MPC round complexity barrier. Our techniques enable breaking this barrier for PageRank---one of the most important applications of random walks---even in more challenging directed graphs, and for approximate bipartiteness and expansion testing.

In the undirected case, we start our random walks from the stationary distribution, which implies that we approximately know the empirical distribution of their next steps.
This allows for preparing continuations of random walks in advance and applying a doubling approach.
As a result we can generate multiple random walks of
length $l$ in $\Theta(\log l)$ rounds on MPC. Moreover, we show that under the popular \cycleConj{} conjecture, this round complexity is asymptotically tight.

For directed graphs, our approach stems from our treatment of the PageRank Markov chain. We first compute the PageRank for the undirected version of the input graph and then slowly transition towards the directed case, considering convex combinations of the transition matrices in the process.

For PageRank, we achieve the following round complexities for damping factor equal to $1 - \eps$:
\begin{itemize}
\item in $O(\log \log n + \log 1 / \eps)$ rounds for undirected graphs (with $\tilde O(m / \eps^2)$ total space),
\item in $\tilde O(\log^2 \log n + \log^2 1/\eps)$ rounds for directed graphs  (with $\tilde O((m+n^{1+o(1)}) / \poly\, \eps)$ total space).
\end{itemize}

The round complexity of our result for computing PageRank has only logarithmic dependence on $1/\eps$. We use this to show that our PageRank algorithm can be used to construct directed length-$l$ random walks in $O(\log^2 \log{n} + \log^2{l})$ rounds (with $\tilde O((m+n^{1+o(1)}) \poly\, l)$ total space). Namely, by letting $\eps = \Theta(1 / l)$, a length-$l$ PageRank walk with constant probability contains no random jump, and hence is a directed random walk.
\end{abstract}
\end{titlepage}

\section{Introduction}
Computing random walks in graphs is a fundamental algorithmic primitive. Random walks find applications in a plethora of computer science research areas. A non-exhaustive list includes optimal PRAM algorithms for connectivity~\cite{Reif:1985,DBLP:journals/jcss/HalperinZ96,Halperin:1996},
rating web pages~\cite{pagerank,2015:FDP}, partitioning graphs~\cite{AndersenCL06}, minimizing query complexity in property testing~\cite{GoldreichR99,KaufmanKR04,CzumajMOS,czumaj2010testing,kale2011expansion,NachmiasS10,Czumaj:2015,8555132}, finding graph matchings~\cite{GoelKK13}, generating random spanning trees~\cite{KelnerM09}, and counting problems~\cite{Jerrum:1996:MCM}.

Intuitively, computing random walks, especially independent random walks from all vertices, should be highly parallelizable, since random walks are memoryless.
However, even if we start a single random walk from each vertex, after a constant number of steps many of the walks may collide in the same vertex.
This is especially problematic in directed graphs, since it is not known how to precompute the vertices where many collisions would happen, other than by simulating length-$l$ walks in $l$ steps or computing the $l$-th power of the transition matrix, which takes quadratic space.


The focus of this work is on generating a large number of independent random walks in a parallel setting. For undirected graphs, we take advantage of the fact that the stationary distribution is proportional to vertex degrees and can thus be computed in a trivial way. If the starting points of the random walks are distributed according to the stationary distribution, then after any number of steps the distribution of the endpoints is the same. We exploit this to pre-sample continuations of random walks and recursively stitch them together in order to generate length-$l$ walks in $O(\log l)$ parallel rounds.


The situation becomes more complicated for directed graphs (or Markov chains such as the one defining PageRank), since we do not know the stationary distribution a priori and hence cannot apply this approach directly. Instead, to compute PageRank, we start from the undirected closure $\bar{G}$ of the input graph $G$, for which we can generate random walks, using the ideas described above.
We then slowly transition from $\bar{G}$ to $G$, and gradually update our approximation of the stationary distribution.
Roughly speaking, at each step we consider a convex combination of the transition matrices of $\undir{G}$ and $G$.
This technique, together with its analysis, is the most important and complex technical contribution of the paper.

Another challenge in computing random walks in directed graphs is the fact that the probabilities of some vertices in the stationary distribution $\pi$ can be as low as $O(1/2^n)$.
Hence, computing $t(v) \sim \pi(v)$ random walks from each vertex $v$, where $t(v) \geq 1$, would result in exponentially many walks.
This challenge can be addressed by using the Markov chain that is used to define  PageRank.
In this Markov chain, each random walk ending at $v$ is either extended with a random outedge of $v$ (with probability $1-\epsilon$), or jumps to a random vertex of the graph (with probability $\epsilon)$.
This small change influences the stationary distribution $\pi$ significantly, since in the modified graph we have $\pi(v) \geq \epsilon / n$ for each vertex $v$.
At the same time, by setting $\epsilon = O(1/l)$, one can guarantee that each random walk does not make a random jump with constant probability, and is thus a random walk in the original graph.
Altogether, our ideas lead to an algorithm for computing length-$l$ random walks in directed graphs in $O(\log^2 \log n + \log^2 l)$ rounds and an algorithm for PageRank that uses $\tO(\log^2 \log n)$ rounds.
Both these algorithms require total space that is almost linear in the input size (assuming we only compute random walks of length $l = \poly \log n$).

We show that our algorithms can be implemented in the Massively Parallel Computation (MPC) model, which has been extensively studied by the theory community in the recent years~\cite{lattanzi2011filtering,AndoniNOY14,czumaj2018round,assadi2019coresets,brandt2018matching,ghaffari2018improved,ghaffari2019sparsifying,onak2018round,gamlath2018weighted,ghaffari2019improved,harvey2018greedy,ASW,log-diameter,assadi2019sublinear,behnezhad2019exponentially,ghaffari2019conditional,BDELM2019}.
We use the most challenging space-per-machine regime of the MPC model, in which the space available on each machine is strongly sublinear in the number of vertices of the graph, i.e., at most $n^{1-\Omega(1)}$. This allows for handling large graphs that do not fit onto a single machine, even if they are sparse, which is the case for social networks and the webgraph.

\subsection{Our Results}
We give new algorithms for sampling independent random walks and show that they can be efficiently implemented in the MPC model.
For the formal description of the model see~\cref{sec:mpc}.

The first result is an algorithm for sampling random walks in undirected graphs.

\begin{restatable}{theorem}{llogl}
\label{theorem:llogl}
Let $G$ be an undirected graph and $C\ge 1$. Let $l$ be a positive integer such that $l = o(S)$, where $S$ is the available space per machine.
There exists an MPC algorithm that samples $\deg^{+}(v) \lceil C \ln n \rceil$ independent random walks of length $l$ starting in $v$ for each vertex $v$ in $G$.
The algorithm runs in $O(\log l)$ rounds and uses $O(Cml\log l\log n)$ total space and strongly sublinear space per machine.
If the algorithm has to return only the endpoints of each random walk, the total space complexity can be reduced to $O(Cml \log n)$ and $l$ can be arbitrarily large.
The algorithm is an imperfect sampler (see \cref{def:sampler}) that does not fail with probability $1-n^{-\frac{C}{6}+1}$.
\end{restatable}

Our algorithm assumes that the length of each random walk is at most the space available per machine, which is $n^\gamma$ for $\gamma \in (0, 1)$.
We believe this assumption is not limiting, since in applications the most interesting regime, $l = n^{o(1)}$, especially $l = O(\poly \log n)$.

One of the main results of this paper is an algorithm for sampling random walks in directed graphs.
To the best of our knowledge, this is the first $o(l)$-round algorithm for sampling length-$l$ independent random walks in any distributed or parallel model, other than the trivial algorithm based on matrix exponentiation, which requires quadratic space.
\begin{restatable}{theorem}{randomwalksdirected}\label{theorem:directed-walks}
Let $G$ be a directed graph. Let $D$ and $l$ be positive integers such that $l = o(S)/\log^3 n$, where $S$ is the available space per machine.
There exists an MPC algorithm that samples $D$ independent random walks of length $l$ starting in $v$ for each $v$ in $G$.
The algorithm runs in $O(\log^2 \log{n} + \log^2 l)$ rounds and uses $\tO\rb{m + n^{1 + o(1)} l^{3.5} + D n l^{2 + o(1)}}$ total space and strongly sublinear space per machine.
The algorithm is an imperfect sampler (see \cref{def:sampler}) that does not fail with probability $1-O(n^{-1})$.
\end{restatable}

We also show that the algorithm for undirected graphs is almost optimal under the \cycleConj{} conjecture, which is the most popular conjecture for showing conditional hardness in the MPC model.
\begin{restatable}{theorem}{lowerbound}\label{thm:lower-bound}
Let $\gamma \in (0, 1)$ be a constant.
In the MPC model with $O(n^{1-\gamma} \poly \log n)$ machines, each having space $O(n^\gamma)$, the following holds.
Each algorithm that can sample $\Theta(\log n)$ independent random walks of length $\Theta(\log^4 n)$ starting at each vertex of the graph requires $\Omega(\log \log n)$ rounds, unless the \cycleConj{} conjecture does not hold.
\end{restatable}

By using our random walk sampling primitive, we give an algorithm for computing PageRank in undirected graphs.

\begin{restatable}{theorem}{prundirected}\label{theorem:PageRank-undirected}
Let $\alpha \in [1 / n, 1 / 4]$ and $\eps \in [\log n/o(S), 1]$, where $S$ is the available space per machine.
There exists an MPC algorithm that, with probability at least $1-\frac{5}{n^2}$, computes a $(1+\alpha)$-approximate PageRank vector in
undirected graphs with jumping probability  of $\eps$ in $O\rb{\log \log n + \log 1/\eps}$ rounds using $O\rb{\frac{m\log^2 n \log \log n}{\epsilon^2 \alpha^2}}$ space.
\end{restatable}

Our next result is an algorithm for computing PageRank in directed graphs.
This is by far the most technically involved part of the paper.
In fact, our algorithm for sampling random walks in directed graphs is a corollary from the lemmas that we develop to obtain the following result.

\begin{restatable}{theorem}{theoremPageRankdirected}\label{theorem:PageRank-directed}
Let $G$ be a directed graph.
Let $\alpha \in [1 / n, 1 / 4]$ and $\eps \in [\log^3 n/o(S), 1]$, where $S$ is the available space per machine.
There exists an MPC algorithm that, with probability at least $1-O(\frac{1}{n})$,
computes a $(1+\alpha)$-approximate $\eps$-PageRank vector of $G$ in $\tO\rb{\log^2 \log n + \log^2 1/\eps}$ rounds, using $\tO\rb{\frac{m}{\alpha^2}+ \frac{n^{1+o(1)}}{\eps^{3.5}\alpha^2}}$ total space and strongly sublinear space per machine.
\end{restatable}

This gives an exponential improvement in the number of rounds with respect to the previously known results~\cite{2015:FDP}.

Recently, it was shown that computing $\Theta(1)$-approximate maximum matching, $\Theta(1)$-approximate minimum vertex cover, and maximal independent set admit $\Omega(\log \log{n})$ conditional lower bound on the round complexity in the MPC model~\cite{ghaffari2019conditional} in the regime with strongly sublinear space per machine.
Hence, obtaining $O(\poly(\log \log n))$-round algorithms seem to be the new complexity benchmark to reach.

\paragraph{Random Walks in PRAM.}
We show that our algorithm for sampling random walks is generic enough to yield interesting results beyond the MPC model.
In particular, we show that it can also be implemented in PRAM. Note that the following theorem gives a nontrivial result whenever $l = \omega(\poly \log n)$.

\begin{restatable}{theorem}{rwpram}
\label{thm:pram}
Let $G$ be a directed graph and $1 \leq l \leq n$. 
There exists an NC algorithm that uses $O((n+m)^{1+o(1)})$ processors and samples one random walk from each vertex of $G$.
All sampled walks are independent.
The algorithm is an imperfect sampler (see \cref{def:sampler}) that fails with probability $1-O(n^{-1})$.
\end{restatable}

We could also formulate a similar theorem for computing PageRank, but an NC algorithm for PageRank is already known, since it can be obtained by simply using power method.

\paragraph{Applications to Property Testing.}
We show how to use our random walk algorithm to approximately test bipartiteness and expansion.
Instead of solving the exact version of the problems, we consider relaxed versions of these problems known from property testing.
Our algorithms either show that the graph is close to having a given property or far away from satisfying it. It is unlikely that the exact versions
of these problems have $o(\log n)$-round algorithms, due to the \cycleConj{} conjecture (see \cref{sec:lower_bound} for conjecture statement).

We now sketch the reductions.
For testing expansion it suffices to observe that to solve the problem from the \cycleConj{} conjecture, it suffices to determine whether the expansion is positive.
For bipartiteness testing, we observe that if we start with two cycles and independently replace each edge with a path of length 2 with probability $1/2$, then with constant probability, we obtain a graph that has an even number of edges overall, but is not bipartite. This will never occur in this transformation if the input is a single cycle. Hence if we could exactly check if a graph is bipartite, we could distinguish inputs that are a single cycle from those that are a disjoint union of two cycles.

%

\begin{itemize}
 \item {\bf Testing Bipartiteness (\cref{sec:Bipartiteness}).} We show how to use our random walk algorithm for \emph{testing bipartiteness}. In this promise problem, we are given a graph $G$ on $n$ vertices with $m$ edges and a parameter $\eps \in (0,1)$. We want to distinguish the case that $G$ is bipartite from the case that at least $\eps m$ edges have to be removed to achieve this property. Our parallel algorithm first reduces the size of the input graph by sampling, and then combines property testing techniques for bipartiteness~\cite{GoldreichR99,KaufmanKR04} with our simulation of random walks. Similar ideas were used in~\cite{CFSV16} for the CONGEST model.

 \item {\bf Testing Expansion (\cref{sec:testing-expansion}).} We say that a graph $G$ of maximum degree $d$ is $\eps$-far from every $\alphastar$-vertex-expander if it is needed to change (add/remove) more than $\eps d n$ edges of $G$ so that the obtained graph is $\alphastar$-vertex-expander. Our algorithm gets $\alpha$ as its input, and returns $\accept$ if $G$ is an $\alpha$-vertex-expander, or returns $\reject$ if $G$ is far from being such an expander. The idea of the algorithm is as follows. From each vertex, for $O(\poly \log{n})$ many times we run a pair of ``long'' random walks. The collision probability of thus generated random walks can be used for approximately testing expansion in $O(\log\log(n/\eps))$ rounds. Our starting point here is the analysis of random walk collision probability, introduced by Czumaj and Sohler~\cite{czumaj2010testing}.

\end{itemize}

\subsection{Previous Research}

\paragraph{Random walks in the streaming model.}
The problem of generating random walks was considered in a number of streaming and parallel computation papers. The paper of Sarma, Gollapudi, and Panigrahy~\cite{SarmaGP11} introduced multi-pass streaming algorithms for simulating random walks. For a single starting point, they can, for instance, simulate single a length $l$ random walk in $\tilde O(n)$ space and $\tilde O(\sqrt{l})$ passes. The paper by Jin~\cite{Jin19} gives algorithms for generating a single random walk from a prespecified vertex in one pass. For directed graphs the algorithm requires roughly $\Theta(nl)$ space and for undirected graphs $\Theta(n\sqrt{l})$ space.

\paragraph{Parallel distributed computation.}
Bahmani, Chakrabarti, and Xin~\cite{BahmaniCX11} give an MPC algorithm for constructing length-$l$ random walks in directed graphs. Their algorithm runs in $O(\log{l})$ rounds, but the walks it generates are \emph{not independent}. 

A recent result by Assadi, Sun and Weinstein~\cite{ASW} gives an MPC algorithm for detecting well-connected components in small space per machine and with exponential speed-up over the direct exploration. As a subroutine, the paper presents an algorithm for generating random walks in an undirected regular graph.
Note that in regular graphs the stationary distribution is uniform and thus the problem of generating random walks becomes somewhat simpler.
Let us also mention that there exist random walk generating algorithms for the distributed CONGEST model~\cite{SarmaNPT13}. These algorithms require however a number of rounds that is at least linear in the diameter, which can be $\Omega(\log n)$ even for expanders.

Computing random walks has also been used in PRAM model as a subroutine in algorithms for computing connected components using a near-linear number of processors~\cite{DBLP:journals/siamcomp/KargerNP99, DBLP:journals/jcss/HalperinZ96}.
However, in both these algorithms random walks starting at different vertices are not independent.

\paragraph{PageRank.}
Since it was introduced in \cite{brin1998anatomy,pagerank}, the computation of PageRank has been extensively studied in various settings. We refer a reader to~\cite{berkhin2005survey,langville2004deeper,duhan2009page} for the early development and theoretical foundations of PageRank.
\cite{SarmaGP11} consider PageRank approximation in the streaming setting. As their main result, they show how to compute an $l$-length random walk in $O(\sqrt{l})$ passes by using $O(n)$ space. By using this primitive, \cite{SarmaGP11} show how to approximate PageRank for directed graphs in $\tO(M^{3/4})$ passes by using $\tO(n M^{-1/4})$ space, where $M$ is the mixing time of the underlying graph.

By building on \cite{avrachenkov2007monte}, \cite{2015:FDP} studied PageRank in a distributed model of computation. For directed graphs, they design a PageRank approximation algorithm that runs in $O\rb{\tfrac{\log{n}}{\eps}}$, where $\eps$ is the jumping probability, i.e., $1-\eps$ is the damping factor. To achieve this bound, they use the fact that w.h.p.~a random walk from a vertex jumps to a random vertex within $O\rb{\tfrac{\log{n}}{\eps}}$ steps. Moreover, to estimate PageRank it suffices to count the number of random walks ending at a given vertex while ignoring how those random walks reached at the vertex. The authors of \cite{2015:FDP} exploited this observation to show that many PageRank random walks can be simulated in parallel while not over-congesting the network. This approach can be implemented in $O\rb{\tfrac{\log{n}}{\eps}}$ MPC rounds. The authors also show how to extend the ideas of \cite{SarmaGP11} and obtain an algorithm for undirected graphs that approximates PageRank in $O\rb{\tfrac{\sqrt{\log{n}}}{\eps}}$ rounds.

Another line of work considered approximate PageRank in the context of sublinear time algorithms, e.g., \cite{chen2004local,bressan2011local,borgs2012sublinear,BPP18} just to name a few. This research culminated in a result of \cite{BPP18} who show that a PageRank of a given vertex can be approximated by examining only $\tO\rb{\min\{m^{2/3} \Delta^{1/3} d^{-2/3}, m^{4/5} d^{-3/5}\}}$ many vertices/arcs, where $\Delta$ and $d$ are the maximum and the average degree, respectively. It is not clear how to simulate this approach in MPC efficiently for all vertices, in terms of both round and total space complexity.

Bahmani, Chakrabarti, and Xin~\cite{BahmaniCX11} show how to use their result for generating random walks to get a constant \emph{additive} approximation to Personalized PageRank of each vertex, which can be used to obtain a constant additive approximation to PageRank.\footnote{Compared to PageRank, in the Personalized PageRank of vertex $v$ a random jump is always performed to $v$.} 
Note that in the case of PageRank constant additive approximation is a weak guarantee, since an all-zero vector provides a constant additive approximation for all but $O(1)$ vertices.
In particular, the algorithm of~\cite{BahmaniCX11} provides non-zero estimates for only $O(\log{n})$ vertices.
In this paper, we show how to compute a multiplicative approximation of PageRank, and hence provide an estimate for each vertex.

We now comment why using random walks generated by the algorithm of~\cite{BahmaniCX11} would require at least quadratic space to obtain multiplicative approximation to PageRank (assuming a standard approach). To compute a length-$l$ random walk from each vertex, the algorithm of \cite{BahmaniCX11} samples $O(l)$ random edges from each vertex. Then, these random edges are used to construct the desired random walks by doubling, that is two walks of length $2^i$ are concatenated to form a walk of length $2^{i+1}$.
The same random edge sample can be used for multiple walks, which results in the following undesirable behavior.

Consider a star-graph with the center being vertex $s$ and use the algorithm of \cite{BahmaniCX11} to construct $T$ random walks of length $O(l)$ from each vertex. Let $W$ be the collection of these walks.
Each of the walks in $W$ either starts at $s$ or visits $s$ directly after the first step.
Since the algorithm samples $O(l \cdot T)$ edges from $s$, there exists a set of vertices $V'$ of size $O(l \cdot T)$, such that except for staring vertices, all vertices of all random walks belong to $V'$.
Hence, for $l \cdot T = o(n)$ there exists a vertex different than $s$ that among the walks of $W$ appears $\Omega(n T/ (l \cdot T)) = \Omega(n / l)$ many times, and also a vertex different than $s$ that appears only $T$ times (but only as the starting vertex of some of the walks).

A standard approach to estimating the PageRank of vertex $v$ consists in counting the number of occurrences of $v$ in the endpoints of walks in $W$.
To estimate PageRank, the value of $l$ is set to be $O(\log{n})$. So, unless $T = \Omega(n / \log{n})$, there exist two vertices $u$ and $v$, both different from $s$, whose visit counts differ by a super-constant multiplicative factor. However, $u$ and $v$ are symmetric and their PageRanks, both personalized and non-personalized, are equal.

\section{Preliminaries}

In this paper we consider directed multigraphs, that is we allow multiple edges between each pair of vertices.
Let $G = (V, E)$ be a directed multigraph.
We use $\deg^{+}_G(v)$ and $\deg^{-}_G(v)$ to denote the number of outgoing and incoming edges of $v$ respectively.
We also define $\deg_G(v) := \deg^{-}_G(v) + \deg^{+}_G(v)$.
The subscript $G$ is often omitted if it is clear from the context.

A graph $G = (V, E, w)$ is \emph{weighted} if $w : E \rightarrow \mathbb{R}$ is a function assigning real weights to the edges.
Note that different edges between the same pair of vertices may have different weights.
We extend the definitions for unweighted graphs to weighted graphs in a natural way.

For a weighted graph $G = (V, E, w)$ and $s \in \mathbb{R}$ we define $s \cdot G$ (often abbreviated as $sG$) to be a graph $G' = (V, E, w')$, where $w'(e) := s\cdot w(e)$ for each $e \in E$.
For two weighted graphs $G_1 = (V, E_1,w_1)$ and $G_2 = (V, E_2,w_2)$, we define $G_1 + G_2 := (V, E_1 + E_2,w_1.w_2)$. Here, $E_1 + E_2$ denotes the \emph{multiset sum} of $E_1$ and $E_2$,
i.e., an element of cardinality $c_1$ in $E_1$ and $c_2$ in $E_2$ has cardinality $c_1+c_2$ in the sum. The weights $w_1.w_2 : E_1 + E_2 \rightarrow \mathbb{R}$ are defined as
\[
	(w_1.w_2)(e)= \begin{cases}
		w_1(e)  & \quad \text{if } e \in E_1,\\
		w_2(e)  & \quad \text{if } e \in E_2.
  \end{cases}
	\]
If $G_1 = (V, E_1)$ and $G_2 = (V, E_2)$ are unweighted, then $G_1 + G_2 = (V, E_1 + E_2)$.

The transpose of a graph $G = (V, E)$ is a graph $G^T = (V, E^T)$ where $E^T = \{vu \mid uv \in E\}$.
A graph is called \emph{undirected} if it is equal to its transpose.
We define the \emph{undirected closure} of $G$, denoted by $\undir{G}$ to be $\undir{G} = G + G^T$.
Note that $\deg^{+}_{\undir{G}}(v) = \deg_G(v)$.

We call a weighted graph $G = (V, E, w)$ \emph{stochastic} if all edge weights are nonnegative and for each $v \in V$ we have $\sum_{vx \in E} w(vx) = 1$.\footnote{Note that we slightly abuse notation and each outgoing edge of $v$ corresponds to a separate summand, even in the presence of parallel edges.}
Observe that if $G_1 = (V, E_1, w_1)$ and $G_2 = (V, E_2, w_2)$ are stochastic $x, y \geq 0$ and $x + y = 1$, then $xG_1 + yG_2$ is stochastic as well.

For a stochastic graph $G = (V, E, w)$, we define $T(G) : V \times V \rightarrow \mathbb{R}$ to be the \emph{transition matrix} of $G$, where $T(G)_{u, v}$ is the total weight of edges $uv$.
For a transition matrix $T(G)$, a \emph{stationary distribution} is a probability distribution $\pi : V \rightarrow \mathbb{R}$, such that $\pi\cdot T(G) = \pi$.
Note that the stationary distribution may not be unique.
A stationary distribution of a stochastic graph is a stationary distribution of its transition matrix.

We say that a vertex $v \in V$ is \emph{dangling} if $\deg^{+}_G(v) = 0$.
For an unweighted graph $G = (V, E)$ with no dangling vertices we denote by $\rw{G}$ the weighted graph obtained by assigning a weight of $1 / \deg^{+}_G(v)$ to each outgoing edge of $v$.
Note that $\rw{G}$ is stochastic.

A \emph{walk} in a graph $G$ is a sequence of edges $u_1v_1, u_2v_2, \ldots, u_kv_k$, such that for $1 \leq i < k$, $v_i = u_{i+1}$.
The \emph{length} of a walk is the number of edges it consists of.

For a stochastic graph $G = (V, E, w)$, a \emph{random walk} in $G$ of length $k$ starting in $s \in V$ is a walk $W = u_1v_1, u_2v_2, \ldots, u_kv_k$ in $G$, where $s = u_1$, which is constructed with the following algorithm.
For each $1 \leq i \leq k$, the edge $u_iv_i$ is chosen independently at random among all outgoing edges of $u_i$.
The probability of choosing a particular outgoing edge is equal to the edge weight.

\subsection{The MPC Model}\label{sec:mpc}
In this paper we will work with the computational model introduced by Karloff, Suri, and Vassilvitskii~\cite{KarloffSV10} and refined in later works~\cite{goodrich2011sorting,BeameKS13,AndoniNOY14}. We call it \emph{massively parallel computation} (MPC), which is similar to the name in Beame et al.~\cite{BeameKS13}.

This model captures main aspects of modern parallel systems, where we have $\machines$ machines and each of them has $S$ words of space.
The total space available on all the machines should not be much higher than the size of the input. In the case of graph algorithms studied
here the input is a collection $E$ of edges and each machine receives approximately $|E|/\machines$ of them at the very beginning.

In MPC model the computation proceeds in \emph{rounds}. During the round, each machine first processes its local data without communicating with other machines.
Then machines exchange messages. When sending a message the machine specifies unique recipient of this message.
Moreover, we require that all messages sent and received by each machine in each round fit into local memory of this machine. Hence, their total length of these messages is bounded by $S$.\footnote{This for instance allows a machine to send a single word to $S/100$ machines or $S/100$ words to one machine, but not $S/100$ words to $S/100$ machines if $S = \omega(1)$, even if the messages are identical.}
This implies that the total communication of the MPC model is bounded by $\machines \cdot S$ in each round. The messages put into recipients local space and can be processed by them in the next round.

At the end of the computation, machines collectively output the solution, i.e., the solution is formed by the union of the outputs of all the machines. The data output by each machine has to fit in its local memory. Hence again, each machine can output at most $S$ words.

\paragraph{Possible Values for $S$ and $\machines$.} Typically in MPC models, one assumes that $S$ is sublinear in the input size $N$. In such case $\machines \ge N/S$.
Formally, one usually considers $S = \Theta\left(N^{1 - \epsilon}\right)$, for some $\epsilon > 0$.

In this paper, the focus is on graph algorithms. By $n$ we denote the number of vertices in the graph, and $m$ is the number of edges. The input
size is $O(n+m)$, where $m$ can be as large as $\Theta\left(n^2\right)$.

The algorithms presented in this paper use total space which is almost linear in the input size, that is $M\cdot S = (m+n)^{1+o(1)}$.
The value of $S$, which is the amount of space available on each machine, is \emph{strongly sublinear} in the number of vertices in the graph,
that is they use $O(n^{\gamma})$ space per machine for a constant $\gamma \in (0, 1)$. This range of parameters is the most interesting
one due to the popular \cycleConj{} conjecture that implies that many graph problems cannot be solved in $o(\log n)$ rounds (see \cref{sec:lower_bound}).
Our algorithms work for any $0<\gamma<1$ so for simplicity we omit this parameter in our theorems. However, we specify only the total space that needs to be available on all the $\machines$.


\paragraph{Communication vs.\ Computation Complexity.} The main complexity measure in this work will be the number of rounds required to finish computation, i.e., communication rounds
are typically the most costly component in the MPC computation. Also, even though we do not make an effort to explicitly bound it, it is apparent from the design of our algorithms that every machine performs $O(S \polylog{S})$ computation steps locally. This in particular implies that the overall work across all the machines is $O(r (n+m) \polylog{n})$, where $r$ is the number of rounds.
The total communication during the computation is $O(r (n+m))$ words.

\subsection{PageRank}

Let $G$ be a directed graph and let $\eps \in (0,1)$. PageRank~\cite{brin1998anatomy} is the stationary distribution of the following Markov chain on the vertices of $G$. At a given vertex $v$, with probability $\eps$, the next vertex is selected uniformly at random from the set of all vertices of $G$, and with probability $1-\eps$, it is selected uniformly at random from the heads of outedges of $v$. Note that PageRank is unique, because the Markov chain described above is ergodic. $1-\eps$ is called the \emph{damping factor}.

\subsection{Relevant Concentration Bounds}
Throughout the paper, we will use the following well-known variants of Chernoff bound.
\begin{theorem}[Chernoff bound]\label{lemma:chernoff}
	Let $X_1, \ldots, X_k$ be independent random variables taking values in $[0, 1]$. Let $X \eqdef \sum_{i = 1}^k X_i$ and $\mu \eqdef \ee{X}$. Then,
	\begin{enumerate}[(A)]
		\item\label{item:delta-at-most-1} For any $\delta \in [0, 1]$ it holds $\prob{|X - \mu| \ge \delta \mu} \le 2 \exp\rb{- \delta^2 \mu / 3}$.
		\item\label{item:delta-at-most-1-le} For any $\delta \in [0, 1]$ it holds $\prob{X \le (1 - \delta) \mu} \le \exp\rb{- \delta^2 \mu / 2}$.
		\item\label{item:delta-at-most-1-ge} For any $\delta \in [0, 1]$ it holds $\prob{X \ge (1 + \delta) \mu} \le \exp\rb{- \delta^2 \mu / 3}$.
		\item\label{item:delta-at-least-1} For any $\delta \ge 1$ it holds $\prob{X \ge (1 + \delta) \mu} \le \exp\rb{- \delta \mu / 3}$.
	\end{enumerate}
\end{theorem}

\section{Overview of Our Techniques}
In order to speed up the generation of a random walk, one may be tempted to generate in parallel its different sections and then stitch them together.
This approach becomes challenging with limited space if, for instance, one wants to generate a large number of random walks that all start from the same vertex.
Unfortunately, until we know where the walks are after $k$ steps, it is difficult to limit the number of continuations corresponding to the consecutive steps---$k+1$, $k+2$, and so on---that we have to consider.
This seems to be an issue that many, if not all, attempts at generating random walks with limited space encounter~\cite{SarmaGP11}.

However, if starting points of random walks are sampled from the stationary distribution, the distribution after any number of steps is the same.
(This observation was previously used by Censor{-}Hillel et al.~\cite{CFSV16} to construct bipartiteness testers in constant-degree graphs for the CONGEST model.) Hence we can sample only slightly more mid-points from the same stationary distribution and recursively generate random walks from them. The key observation is that
the stationary distribution in undirected graphs is known in advance, i.e., it is proportional to vertex degrees.
This approach enables generating many fully independent random walks of length $l$ from each vertex
in $\Theta(\log l)$ rounds of MPC using space near-linear in the size of the input. We present this idea in detail in \cref{sec:random-walks}. We also argue
that this round complexity is tight under the \cycleConj{} conjecture (see \cref{sec:lower_bound}).

\subsection{Random Walks and PageRank in Directed Graphs}

The problem of generating random walks is much more challenging in directed graphs.
We do not know a priori a stationary distribution, so it is not possible to directly apply the previous approach.
Namely, our sampler for undirected graphs crucially uses the fact that we know the stationary distribution in advance and the probabilities in this distribution are not too small.
This is because the number of random walks sampled from each vertex has to be both proportional to the stationary distribution and large enough to obtain concentration guarantees.

In general directed graphs these assumptions do not hold.
Some vertices in the stationary distribution may have small or even zero probability.
One can also show that even an exponential number of samples (i.e., $2^{\Theta(n)}$) may not be enough to estimate the stationary distribution for some vertices.
Hence, even if we knew the stationary distribution $\pi$ and wanted to compute $t(v)$ random walks from each vertex $v$, such that $t(v) \sim \pi$ and $t(v) \geq 1$, we may need to compute exponentially many walks altogether.
In fact, in some computational models, there is a known separation between the difficulty of generating random walks for directed and undirected graphs~\cite{Jin19}.

If the outdegrees in the graph are bounded by $\poly \log n$ we can use the following simple approach (see \cref{sec:rw-directed} for a more detailed description).
For each vertex $v$, we find all vertices whose distance from $v$ is at most $\epsilon \log n / \log \log n$.
Once we know this set, we are able to simulate $\epsilon \log n / \log \log n$ steps of a random walk in a single round.
At the same time the assumption on the outdegrees allows us to bound the total space usage.
This approach does not generalize to the case when the outdegrees are arbitrary or we want to generate random walks of length $\omega(\log n)$.

Instead of dealing with the problem of sampling random walks in directed graphs, we first consider the specific case of computing PageRank.
The starting point of our approach is the algorithm of \cite{Breyer02markovianpage} that we review in \cref{sec:PageRank}.
At a high level, the algorithm boils down to computing $O(\log n / \epsilon)$ random walks of length $O(\log n / \epsilon)$ in a graph $G_\epsilon$, which is defined by
\[ G_\epsilon = (1-\epsilon)\rw{G} + \frac{\epsilon}{n} J, \] 
where $J$ is complete directed graphs with self loops. 
In other words, we consider walks that with probability $1-\epsilon$ perform a single step of a random walk on $G$, and with probability $\epsilon$ jump to a random vertex in $G$.
Note that PageRank is exactly the stationary distribution of $G_\epsilon$.

Note that in order to make $G_\epsilon$ well defined, we need to assume that $G$ does not contain {\em dangling vertices}, i.e., vertices without outedges.
The usual approaches for handling graphs with dangling vertices are discussed in \cref{app:handling-dangling-nodes}. In particular, we show 
that the two most typical approaches: adding self-loops and restarting the walks, are equivalent up to a simple transformation. 
To our surprise this relation was not previously observed, e.g., \cite{delcorso2005} argues why one method is better than the other. 

Using $G_\epsilon$ instead of $G$ for sampling random walks resolves one of the challenges.
Namely, we show that in the stationary distribution of $G_\epsilon$ the probability of each vertex is at least $\epsilon / n$.
Hence, by sampling $\Theta(n \log n)$ random walks we are actually able to approximate the stationary distribution.
This observation was already used by Das Sarma et al.~\cite{2015:FDP} to obtain an $O(\log n)$-round algorithm for directed
and $O(\sqrt{\log n})$ round algorithm for undirected graphs.

However, we are still left with the other difficulty: we do not know the stationary distribution.
We overcome it by using a novel sampling technique, which is the main technical contribution of the paper and may be of independent interest.
In the remaining part of this section we give an overview of the technique.

The core part of our algorithm is a procedure for adjusting sampling probabilities.
Let $D_1$ and $D_2$ be two discrete probability distributions that are, roughly speaking, similar, i.e., the elementary events are assigned similar probabilities in both distributions.
The procedure, given a random sample from $D_1$ either produces a sample from $D_2$ or fails.
As long as $D_1$ does not differ much from $D_2$, the failure probability is small.
We give two implementations of the procedure, each of which is well suited for some specific distributions $D_1$ and $D_2$.
The easier implementation is based on \emph{rejection sampling}.
This means that the procedure has the property that it either returns the sample it gets or fails.
As a result, the procedure never actually needs to sample from any of the distributions $D_1$ or $D_2$.
It actually only needs to make a single Bernoulli trial in order to decide whether to fail.
This is a very useful property, as in our case sampling a random walk is much more expensive than doing a Bernoulli trial.

Let us now describe our sampling technique using the above idea.
Recall that $\bar{G}$ denotes the undirected closure of $G$.
We note that by using~\cref{alg:random-walks} we can efficiently sample PageRank walks in $\bar{G}$.
Hence we are going to first sample PageRank walks in $\bar{G}$ and then gradually shift towards PageRank walks in $G$ by adjusting transition probabilities.

The main technical challenge is to overcome the fact that even small changes to the transition probabilities can cause large changes to stationary distribution.
Small changes can be amplified $\Theta(n)$ times by the network structure, so that we cannot use stationary distribution computed at one step for sampling walks in the following step even if the difference is very similar.
Instead we reinterpret walks directly using the procedure for adjusting sampling probabilities, i.e., in each step we sample slightly more walks and then reinterpret them as walks for modified graph. In this process we lose some
fraction of walks. Finally, we use the resulting walks to estimate the stationary distribution for the next step.

Let us now describe this idea more formally.
Consider a sequence $\bar{G} = G_1, G_2, \ldots, G_k = G$ of graphs, where, informally speaking, each $G_i$ is a mixture of $\bar{G}$ and $G$.
(In the algorithm we are going to use $k = \Theta(1 / \log \log n)$.)
Formally, $G_i = (k-i)/(k-1)\bar{G} + (i-1)/(k-1)G$.

Our algorithm computes PageRank walks in $G_i$ for $i=1, \ldots, k$.
The first step is easy: as we noted above, computing PageRank walks in $G_1$ can be done using~\cref{alg:random-walks}.
In each of the remaining steps the algorithm starts with PageRank walks in $G_i$.
Then, it uses the procedure for adjusting sampling probabilities to obtain PageRank walks in $G_{i+1}$.
However, the number of random walks in $G_{i+1}$ obtained this way is significantly smaller than the number of walks in $G_i$ that we had at the beginning of the step.
Still, the number of walks is large enough to estimate the stationary distribution of PageRank walks in $G_{i+1}$.
The algorithm then uses this estimated stationary distribution to compute more PageRank walks in $G_{i+1}$, which allows the process to continue.

Intuitively, the first steps of this process are the most challenging due to potential existence of vertices, whose degrees in $G$ and $\bar{G}$ differ significantly.
As an example, consider a vertex $v$ in $G$ with a single outgoing edge $vw$ and $n-1$ incoming edges.
A random walk in $\bar{G}$ of length 1 starting at $v$ (that is, a random edge incident to $v$) is the edge $vw$ only with probability $O(1 / n)$.
Hence, the rejection sampling would fail with a very large probability.

To alleviate that, we first transform our input graph into a directed graph such that for each vertex $v$ it holds $c \cdot \deg^+(v) \ge \deg(v)$, where $c$ is a constant that we set later.
We call such graphs \emph{$c$-balanced}.
This transformation is explained in \cref{sec:transforming-to-c-balanced}.
We show how to compute PageRank of $c$-balanced graphs in \cref{sec:our-algorithm-directed-PageRank}.

In the process of transforming an input to a $c$-balanced graph, each edge is replaced by a path of length $\log{n}$.
This means that a PageRank walk of length $k$ in the input graph corresponds to a PageRank walk of length $k \log{n}$ in the transformed $c$-balanced one.
As a result, informally speaking, PageRank walks in a $c$-balanced graphs for the jump factor $\eps$ jump to random vertices more often then the corresponding walks in the original graph.
Hence, if we applied the algorithm unchanged to the transformed algorithm, we would effectively compute $(\eps\log{n})$-PageRank walks.

A natural idea is to run our algorithm for $\eps' = \eps / \log{n}$, but this would increase the round complexity or space usage of our algorithm significantly.
Instead, we reuse the idea of gradually adjusting the transition matrix that we use for sampling random walks.
We start with a jump factor of $1/2$ and then move towards $\eps / \log{n}$ in steps, in that way ensuring that the space requirement and round complexity is as desired.
This approach is presented in \cref{sec:lower-eps}.

Once we obtain PageRank walks for a $c$-balanced graph with respect to jump factor $\eps / \log{n}$, in \cref{sec:mapping-back} we show how to map those walks back to the input graph.

Finally, we observe that the algorithm for computing random walks in $G_\epsilon$ can be used to compute random walks in $G$.
Assume we are interested in random walks of length $l$.
Then, each walk in $G_{\epsilon}$ for $\epsilon = 1/l$ does not contain any random jumps with constant probability and is thus a random walk in $G$.
Hence, by sampling sufficiently many random walks in $G_{1/l}$ we are able to compute random walks in $G$.
The key property here is that the round complexity of sampling random walks in $G_\epsilon$ is $\tilde{O}(\log^2 (1/\epsilon) + \log^2 \log n)$, so as long as $l = \poly \log n$, the overall round complexity is $O(\poly \log \log n)$.

\section{Sampling Random Walks}
\label{sec:random-walks}

In this section we show algorithms for sampling random walks from a given stochastic graph $G$.
The algorithms require that we know (at least approximately) the stationary distribution of $G$.
This is easy in the case when we deal with undirected graphs, that is $G = \rw{G_U}$, where $G_U$ is an undirected graph.
In this case the stationary distribution of $G$ is given by
\begin{equation}
\label{eq:stationary}
	\pi(v)=\frac{\deg^{+}_{G}(v)}{2m}.
\end{equation}
Hence, if we sample the starting point of a random walk from from $\pi$, then after any number of steps the endpoint will follow distribution $\pi$ as well.
This allows us to use doubling.
Since the number of random walks ending in each vertex is (in expectation) the same as the number of random walks starting in each vertex, we can pair up these random walks and stitch together each pair of walks of length $k$ into one walk of length $2k$.
The pseudocode of our algorithm is given as~\cref{alg:random-walks}.

\begin{algorithm}[H]
\begin{algorithmic}[1]
	\Function{RandomWalks}{$G, l, t$}
\ForAll{$v \in V(G)$ in parallel\label{line:RandomWalks-generating-W0}}
	\State{Generate $t_0(v)$ length $1$ random walks in $G$ staring in $v$. Let $W_0(v)$ be the set of these walks.}
\EndFor

\For{$i \gets 1 \ldots \lceil \log{l} \rceil$ \label{line:RandomWalks-doubling-loop}}
	\ForAll{$v \in V$ in parallel}
		\State{Select $t_{i}(v)$ random walks from $W_{i - 1}(v)$. Let that set be $U_{i}(v)$.}
		\State{For each walk $w \in U_{i}(v)$, consider its endpoint $u$. Ask $u$ to extend $w$ by a yet unused walk from $W_{i - 1}(u) \setminus U_{i}(u)$. Let $W_i(v)$ denote the set of all these extended walks originating at $v$. If $u$ does not have unused walks anymore, the algorithm fails. \label{line:extend-random-walks}}
	\EndFor
\EndFor
\State{For each $v\in V$ truncate walks in $W_{\lceil \log{l} \rceil}(v)$ to length $l$.}
\State{Return $W_{\lceil \log{l}\rceil }(v)$ for each $v \in V$}
\EndFunction
\end{algorithmic}
	\caption{Given $G, l, t_0, \ldots, t_{\lceil \log l \rceil}$, sample $t_{\lceil \log l \rceil} (v)$ random walks of length $l$ according to $T$ from each $v \in V(G)$.}
\label{alg:random-walks}
\end{algorithm}

The algorithm takes the following parameters.
	$G$ is the input graph, which has to be stochastic, $l$ is the desired walk length and $t_i(v)$ controls the number of walks starting in vertex $v$ that we would like to sample in the $i$-th iteration of the algorithm. Note that the algorithm requires that $t_i$ has certain properties, e.g., for undirected graphs we show that the algorithm works for $t_i(v)$ being proportional to $\deg^{+}_v(G)$.
Also, for a fixed $v$ the sequence $t_0(v), t_1(v), \ldots$ has to fulfill certain properties.

In the ideal scenario for each vertex the number of random walks that start and finish in each vertex are equal to the expected value. In such case,
in each step we could match all walks into pairs and obtain two times fewer walks of twice the length.
However, the numbers may diverge from the expected values and thus we need to sample a bit more random walks to ensure that there are enough of them with high probability.
	We set $t_i(v) = \deg^{+}_G(v) \lceil C \log n \cdot k_i \rceil$, whereas the sequence $k_i$ controls how many more walks we sample in each step.
A simple solution is to set $k_i =2^{2\lceil \log l\rceil-2 i}$, which implies that $k_{\lceil \log l \rceil} = 1$ and $k_i = k_{i-1}/4$.
By doing so, in each step, in expectation we have twice as many walks as we really need, and it is easy to show that the number of walks is sufficient with high probability.
For the proof we are going to use fewer walks and thus slightly reduce the space complexity.

Observe that if \cref{alg:random-walks} never failed, it would have generated independent random walks.
However, when many walks collide, i.e., end in the same vertex, the algorithm is forced to fail. This means
that we get a random sample from a modified distribution in which the probability of some elements is decreased.
This fraction on which the algorithm fails will be very small.
We formalize this notion using the following definition.

\begin{definition}\label{def:sampler}
Let $(X,p)$ be a discrete probability space.
An {\em imperfect sampler} for $(X, p)$ is an algorithm that returns samples from a probability space $(X \cup \{\mathrm{fail}\}, p')$, such
that $p'(x)\le p(x)$ for all $x\in X$.
The \emph{failure probability} of the sampler is $p'(\mathrm{fail})$.
\end{definition}

We are going to construct samplers where $p'(\mathrm{fail})$ is arbitrarily small.
Note that an algorithm $\mathcal{A}$ using a (perfect) sampler for a probability space $(X, p)$ can be naturally translated to an algorithm $\mathcal{A'}$ using an imperfect sampler.
Unless the sampler fails, $\mathcal{A'}$ produces the same result as $\mathcal{A}$.

We now define the sequence $k_i$ that we use
\begin{eqnarray}
\label{eq:ki-l-log-l-1}
  k_0 &=& 2^{\lceil \log{l}\rceil + 6} \cdot (\lceil \log{l} \rceil + 6) \\
\label{eq:ki-l-log-l-2}
  k_{i - 1} &=&  2 k_i + \sqrt{k_i}.
\end{eqnarray}

The following bounds can be proven for $k_i$.
\begin{lemma}
\label{lemma:bound-ki}
For $0\le i \le \lceil \log l\rceil$ we have
\begin{enumerate}[(i)]
\item \label{item:lower-bound-on-k-i} $k_{i} \ge 2^6$,
\item \label{item:upper-bound-on-k-i} $k_{i} \le 2^{\lceil \log{l}\rceil - i + 6} \cdot (\lceil \log{l} \rceil + 6)$.
\end{enumerate}
\end{lemma}
\begin{proof}
In order to show (i) we will prove that
\begin{equation}\label{eq:ki-bound}
  k_i\ge 2^{\lceil \log{l}\rceil - i + 6} \cdot (\lceil \log{l}\rceil - i + 6).
\end{equation}
We will show \eqref{eq:ki-bound} by induction on $i$.
	
	For $i = 0$, \eqref{eq:ki-bound} follows by definition of $k_0$.
	
	Assume now that $k_{j} \ge 2^{(\lceil \log{l}\rceil) - j + 6} \cdot (\lceil \log{l}\rceil - j + 6)$ for each $j \le i - 1$, and we want to prove that the inequality holds for $j = i$ as well. Recall that $k_{i - 1} = 2 k_i + \sqrt{k_i}$. Towards a contradiction, assume that $k_{i} < 2^{\lceil \log{l}\rceil - i + 6} \cdot (\lceil \log{l}\rceil - i + 6)$. For the sake of brevity, define $t \eqdef \lceil \log{l}\rceil - i + 6$. Then, we have
	\[
		k_{i - 1} = 2 k_i + \sqrt{k_i} < 2^{t + 1} \cdot t + \sqrt{2^t \cdot t} < 2^{t + 1} \cdot t + 2^t < 2^{t + 1} \cdot (t + 1) \le k_{i - 1},
	\]
	which is a contradiction. Hence, \eqref{eq:ki-bound} holds, and (i) follows.

   As for (ii) we have $k_{i-1} = 2k_i+ \sqrt{k_i} \ge 2k_i$. Hence, $k_i \le 2^{-1} k_{i-1} \le 2^{-i} k_0 = 2^{\lceil \log{l}\rceil - i + 6} \cdot (\lceil \log{l} \rceil + 6)$.
\end{proof}

We now show that~\cref{alg:random-walks} can be used to sample random walks in undirected graphs.
The proof is a relatively simple application of Chernoff bound.

\begin{lemma}
\label{lem:rw-fail}
	Let $G$ be a stochastic graph, such that $G = \rw{G_U}$ for some undirected graph $G_U$, $l,C \geq 1$, $l = o(S)$, $t_i(v) = \deg^{+}_G(v) \lceil C \log n \cdot k_i \rceil$, where $k_i$ is given by~\eqref{eq:ki-l-log-l-1} and~\eqref{eq:ki-l-log-l-2}.
Then $\textsc{RandomWalks}(G, l, t)$ (\cref{alg:random-walks}) does not fail with probability at least $1 - n^{-\frac{C}{6}+1}$.
\end{lemma}

\begin{proof}
Let $X^u_i$ be the number of random walks that end at vertex $u$ in iteration $i$.  As long as $X^u_i + t_i(u) \le t_{i - 1}(u)$ holds for each vertex $u$, the algorithm does not fail,
    whereas failure happens when $X^u_i + t_i(u) > t_{i - 1}(u)$.  Let $\delta_i$ be such that $1 + \delta_i = (t_{i - 1}(u) - t_i(u)) / t_i(u)$.
\begin{eqnarray*}
	1 + \delta_i &=& \frac{t_{i - 1}(u) - t_i(u)}{t_i(u)} = \frac{\deg^{+}(v)\cdot \lceil C \ln n (2k_i+\sqrt{k_i})\rceil - \deg^{+}(v)\cdot \lceil C \ln n \cdot k_i \rceil}{\deg^{+}(v) \lceil C \ln n \cdot k_i \rceil}\\
	&=& \frac{\lceil C \ln n (2k_i+\sqrt{k_i})\rceil - \lceil C \ln n k_i \rceil}{\lceil C \ln n k_i \rceil}\\
	&=& \frac{\lceil C \ln n (2k_i+\sqrt{k_i})\rceil}{\lceil C \ln n k_i \rceil }-1\\
	&\geq& \frac{C \ln n (2k_i+\sqrt{k_i})}{\lceil C \ln n k_i \rceil }-1\\
	&=& 1 + \frac{C \ln n \cdot \sqrt{k_i} - 2}{\lceil C \ln n\cdot k_i \rceil}
\end{eqnarray*}

 Now, by Chernoff bound, we have
	\begin{eqnarray*}
	  \prob{X^v_i > t_{i - 1}(v) - t_i(v)} &=& \prob{X^u_i > (1 + \delta_i) \ee{X^u_i}} \le \exp{\rb{-\frac{\delta_i^2 \cdot t_i(v)}{3}}} \\
		&\le & \exp{\rb{-\frac{(C \ln n \cdot \sqrt{k_i} - 2)^2 k_i C \ln n}{3\lceil C \ln n \cdot k_i \rceil^2} \cdot \deg^{+}(v) }} \\
		&\le & \exp{\rb{-\frac{C \ln n}{3 \cdot 2} }}  = n ^{-C / 6}.\\
	\end{eqnarray*}
Where the last inequality follows from the fact that $(\sqrt{x} - 2)^2 x / (x+1)^2 \geq 1/2$ for $x \geq 64$.
    By taking union bound over all vertices and all iterations of the algorithm the probability of failure is less than $n^{-\frac{C}{6}+1}$.

	
\end{proof}

\begin{lemma}
\label{lemma:ref-space-1}
	Assume that $G$, $l$, $t$, are defined as in \cref{lem:rw-fail}.
Then $\textsc{RandomWalks}(G, l, t)$ (\cref{alg:random-walks}) can be implemented to run in $O(\log l)$ MPC rounds using $O(Cml\log l\log n)$ total space.
\end{lemma}
\begin{proof}
The space is bounded by $\max_{1\le i \le \lceil \log l \rceil} O\rb{\sum_{v\in V} t_{i}(v) \cdot 2^i \cdot k_i}$. By \cref{lemma:bound-ki}~\eqref{item:upper-bound-on-k-i} we have
\begin{eqnarray*}
	\sum_{v\in V} t_{i}(v)\cdot 2^i \cdot k_i  &=& \sum_{v \in V} \deg^{+}(v) \cdot \lceil C \ln{n} \rceil\cdot 2^i \cdot 2^{\lceil \log{l}\rceil - i + 6} =  O(Cml \log l \ln n).
\end{eqnarray*}

The MPC implementation is described in \cref{section:MPC-Implementation-RandomWalks}.

\end{proof}

\subsection{MPC Implementation of \RandomWalks}
\label{section:MPC-Implementation-RandomWalks}
We now describe how to implement \RandomWalks. First, we show how to implement every iteration of the loop on \cref{line:RandomWalks-doubling-loop} in $O(1)$ rounds. Then, we show how to implement \cref{line:RandomWalks-generating-W0} also in $O(1)$ rounds.
We begin by defining primitives \NumberingSublists and $\Predecessor$.
\begin{itemize}
	\item \NumberingSublists: Given a list $L$ of tuples, let $L(x)$ be the sublist of $L$ containing all the tuples whose first coordinate is $x$. \NumberingSublists assigns distinct numbers $1$ through $|L(x)|$ to the tuples of $L(x)$, for all values $x$. The numbers are assigned arbitrarily.
	\item \Predecessor: Considered an ordered list of tuples such that each tuple is labeled by $0$ or by $1$. Then, for each tuple $t$ labeled by $0$ $\Predecessor$ associates the closest tuple $t'$ labeled by $1$ such that $t'$ comes before $t$ in the ordering. In \cite{BDELM2019} was proved that $\Predecessor$ can be implemented in $O(1)$ MPC rounds with $n^{\delta}$ space per machine, for any constant $\delta > 0$.
\end{itemize}

\paragraph{\NumberingSublists in $O(1)$ rounds.}
Given a list $L$ of tuples, sort all the tuples with respect to their first coordinate. Number by $j$ the $j$-th tuple in that sorted list; let $j(t)$ be the number of tuple $t$. Values $j(t)$ can be computed by prefix sums with each tuple having value $1$. In \cite{goodrich2011sorting} was shown how to find prefix sums in $O(1)$ rounds. Observe that this number would assign consecutive numbers to the tuples of $L(x)$, but not necessarily starting from $1$. We now show how to shift these numbers so that the tuples of $L(x)$ are numbered $1$ through $|L(x)|$.

Define $m(x) = \arg \min_{t \in L(x)} j(t)$. Label $m(x)$ by $1$, and all other tuples of $L(x)$ label by $0$. A tuple $t \in L(x)$ \emph{is not} $m(x)$ if on its machine there is another tuple in $L(x)$ with number smaller than $t$, or if $t$ has the smallest number on its machine but on the previous machine there is also a tuple from $L(x)$. So, to perform this labeling, it suffices that each machine sends to the next machine its tuple with the highest number.  Use $\Predecessor$ to assign $m(x)$ to each tuple of $L(x)$. Finally, assign to each tuple $t \in L(x)$ value $1 + j(t) - j(m(x))$. These values are the assignments that \NumberingSublists is supposed to output.

\paragraph{Implementation of \cref{line:RandomWalks-doubling-loop} of \RandomWalks.}
Consider the $i$-th iteration of the loop on \cref{line:RandomWalks-doubling-loop}. Assume that so far we have $W_{i - 1}$ and our goal is to obtain $W_i$. First, invoke $\NumberingSublists$ to number all the walks from each $W_{i - 1}(v)$ with the numbers $1$ through $|W_{i - 1}(v)|$. For each walk $w \in W_{i - 1}(v)$, let $j(w)$ be the number assigned to $w$ by $\NumberingSublists$. Create a list of tuples as follows:
\begin{itemize}
	\item If $j(w) \le t_i(v)$, create a tuple $(w_{last}, j(w), w)$, where $w_{last}$ is the last vertex of $w$. This effectively creates $U_i(v)$.
	\item If $j(w) > t_i(v)$, create a tuple $(v, j(w) - t_i(v), w)$. This effectively creates $W_{i - 1}(v) \setminus U_i(v)$.
\end{itemize}
Sort all these created tuples. Now, in this sorted list, next to each other will be tuples $(u, z, w_1)$ and $(u, z, w_2)$ such that $u$ is the last vertex of $w_1$ and the first vertex of $w_2$. Let $v$ be the first vertex of $w_1$. We append $w_2$ to $w_1$ to obtain a random walk in $W_i(v)$. Since each of these operations requires $O(1)$ rounds, a single iteration of \cref{line:RandomWalks-doubling-loop} can be implemented in $O(1)$ rounds.

\paragraph{Implementation of \cref{line:RandomWalks-generating-W0} of \RandomWalks.}
The loop on \cref{line:RandomWalks-generating-W0} generates $t_0(v)$ random edges incident to $v$. We implement this step in MPC as follows.
First, we create a sorted array of length $\sum_{v \in V} t_0(v)$ such that this array contains $t_0(v)$ copies of $v$. Then, each copy of $v$ in this array will be ``in charge'' of choosing one random edge incident to $v$. To store this array across machines, we think of the array being partitioned into subarrays of length $S' = \Theta(S)$, where the $i$-th machine holds subarray $i$.

Next, we explain how to replicate vertex $v$ $t_0(v)$ many times. Consider $x = \left \lceil \sum_{v \in V} t_0(v) / S' \right \rceil$ machines and assume that their IDs are $1$ through $x$. For each vertex $v$ compute $X(v) = \sum_{u < v} t_0(u)$. For each $v$, if $1 + X(v)$ is not divisible by $S'$, $v$ sends to machine $\left \lceil (1 + X(v)) / S' \right \rceil$ pair $((1 + X(v)) \bmod S', v)$ and otherwise $v$ sends $(S', v)$. In addition, to machine $x$ we send $\sum_{v \in V} t_0(v)$. Intuitively, in our big array this effectively marks the first position where the copies of $v$ begin. After this, each machine labels the received pair $(y, v)$ by $1$ and also creates $(z, \bot)$ for each $z \in \{1, \ldots, S'\}$ such that the machine did not receive pair $(z, u)$; the exception is for machine $x$ in which $z$ ranges until the last value representing position $\sum_{v \in V} t_0(v)$. Pairs $(z, \bot)$ are labeled by $0$. Within each machine, pairs $(z, \bot)$ and $(y, v)$ are sorted (they are \emph{not} sorted between the machines). Next, we use $\Predecessor$ to associate the closest $(y, v)$ to $(z, \bot)$. After this step, we replace $(z, \bot)$ by $(z, v)$. Then, each pair $(z, v)$ samples a random integers $j$ between $1$ and $d(v)$ and creates pair $(v, j)$. Notice that there might be pairs $(v, j_1)$ and $(v, j_2)$ such that $j_1 \neq j_2$ or there might be multiples pairs $(v, j_1)$. Let $J$ denote the multiset of these pairs across all machines.

Now, for each edge $\{u, v\}$ create two pairs $(u, v)$ and $(v, u)$. By using $\NumberingSublists$, assign numbers $1$ through $d(v)$ to pairs $\{(v, w)\ |\ w \in N(v)\}$. Let $j$ be the number assigned to $(v, w)$. Create triples $(v, j, w)$. Finally, sort all these triples together with $J$, where in the ordering triple $(v, j, w)$ comes before a pair $(v, j) \in I$. Label triples by $1$ and the pairs of $J$ by $0$. Use $\Predecessor$ to assign a triple to each pair from $J$. If triple $(v, j, w)$ is associated to a pair $(v, j)$, then add $(v, w)$ to $W_0(v)$. This concludes the implementation of \cref{line:RandomWalks-generating-W0}.

\subsection{Sampling Endpoints}
In this section we show that if we only need to find the endpoint of each random walk, then the space usage of~\cref{alg:random-walks} can be improved.
Actually, this will be the case for some of our applications.  In order to reduce the space requirement we observe
that \cref{alg:random-walks} only looks on endpoints of the sampled walks. Hence, we do not need to store internal vertices of the walks in the algorithm.
This alone does not suffice to improve the space usage, as the first iteration, which samples random walks of length $1$, still achieves the peak space usage (and clearly does not benefit from not storing the internal vertices of the random walks).
However, we can first sample random walks of length roughly $\log \log l$ using a simple simulation taking $O(\log \log l)$ steps and after that continue similarly to~\cref{alg:random-walks}.
Another benefit of the algorithm is that since it does not store the entire walks, the walks that it produces can be arbitrarily long, i.e., their length is not bounded by the space available on each machine.

For simplicity of presentation in this section we assume that $l$ is a power of $2$.
The new algorithm is given as~\cref{alg:random-walks-2}.

\begin{algorithm}[H]
\begin{algorithmic}[1]
\Function{RandomWalkEndpoints}{$G, l, t$}
\ForAll{$v \in V$ in parallel}
	\State{Perform $t_{\lceil \log \log l \rceil}(v)$ length $2^{\lceil \log \log l \rceil}$ random walks in $G$ starting at $v$.}
	\State{Let $W_{\lceil \log \log l\rceil}(v)$ be the set of endpoints of these walks.}
\EndFor

\For{$i \gets \lceil \log \log l \rceil+1 \ldots \log{l} $}\label{l:a2a}
	\ForAll{$v \in V$ in parallel}
		\State{Randomly select $t_i(v)$ elements from $W_{i - 1}(v)$. Let that set be $U_{i}(v)$.}
		\State{For vertex $w \in U_{i}(v)$ take yet unused element from $W_{i - 1}(u) \setminus U_{i}(u)$.
               Let $W_i(v)$ denote the set of all these elements. If $u$ does not have unused elements, the algorithm fails.}
	\EndFor
	\EndFor\label{l:a2b}
\State{Return $W_{\log{l}}(v)$ for each $v \in V$}
\EndFunction
\end{algorithmic}
\caption{Given $G, l, t_1, \ldots, t_{\lceil \log l \rceil}$, sample $t_{\lceil \log l \rceil} (v)$ endpoints of random walks of length $l$ in $G$ each $v \in V(G)$.}
\label{alg:random-walks-2}
\end{algorithm}

\begin{lemma}
\label{lem:rw-fail-2}
	Let $G$ be a stochastic graph, such that $G = \rw{G_U}$ for some undirected graph $G_U$, $l, C \geq 1$, $t(v) = \deg^{+}_G(v) \lceil C \log n\rceil$ and $k_i$ be given by~\eqref{eq:ki-l-log-l-1} and~\eqref{eq:ki-l-log-l-2}.
Then $\textsc{RandomWalkEndpoints}(G, l, t)$ (\cref{alg:random-walks-2}) does not fail with probability at least $1 - n^{-\frac{C}{6}+1}$.
\end{lemma}

\begin{proof}
The proof is almost identical to the proof of \cref{lem:rw-fail}.
It suffices to observe that the first loop of \textsc{RandomWalkEndpoints} effectively replaces the first $\lceil \log \log l\rceil$ iterations of the loop in \cref{alg:random-walks} by a procedure that does not fail.
\end{proof}

\begin{lemma}
\label{lemma:ref-space-2}
Assume that $G$, $l$, $t$ are defined as in \cref{lem:rw-fail-2}.
Then \textsc{RandomWalkEndpoints} (\cref{alg:random-walks-2}) requires $O(\log l)$ rounds and $O(Cml\log n)$ total space.
\end{lemma}
\begin{proof}
The space is bounded by $O\rb{\sum_{v\in V} t_{\lceil \log \log l \rceil}(v)}$ and we have
\begin{eqnarray*}
	\sum_{v\in V} t_{\lceil \log \log l \rceil}(v) &=& \sum_{v \in V} \deg^{+}(v)\cdot \lceil C \ln n\rceil \lceil k_{\lceil \log \log l \rceil} \rceil
 \\
   &\stackrel{\text{by \cref{lemma:bound-ki}~\eqref{item:upper-bound-on-k-i}}}{\le} & (C+1) m\ln n \cdot  2^{\lceil \log{l}\rceil - \lceil \log \log l\rceil + 6} \cdot (\lceil \log{l} \rceil + 6) \\
	& \le & (C+1)m\ln n \frac{2^{\lceil \log l\rceil}}{\log l} \cdot (\lceil \log{l} \rceil + 6)\\
   &=& O(Cml \ln n).
\end{eqnarray*}

\end{proof}

By combining \cref{lem:rw-fail,lemma:ref-space-1,lem:rw-fail-2,lemma:ref-space-2} we obtain the first result of the paper.

\llogl*

\subsection{Sampling Random Walks Given Approximate $\pi$}
\label{sec:random-walks-with-tpi}
While computing the stationary distribution $\pi$ is easy for undirected graphs, the problem is significantly more involved if we consider directed graphs.
In this section we show how to use \cref{alg:random-walks} to sample random walks in a directed graph given an \emph{approximation} $\tpi$ of $\pi$.
Our approach would also require that $\pi = \Omega(1/n)$.
In \cref{sec:PageRank} we show how to use this sampling procedure for computing PageRank and for sampling random walks in directed graphs.

As the first step, in the next lemma we show that if $\pi$ and $\tpi$ are close, then taking a $1$-step random walks starting from $\tpi$ results in a distribution that also is also close to $\pi$.

\begin{lemma}\label{lemma:approximate-tpi-implies-approximate-Ttpi}
Let $G$ be a stochastic graph and $T = T(G)$.
Let $\pi$ be the stationary distribution for matrix $T$ and let $\tpi \in \bbR^V$ be an arbitrary vector. Then, the following holds:
\begin{enumerate}[(A)]
	\item \label{item:mul-approximation} If $|\tilde{\pi}(v)-\pi(v)|\le \alpha \pi(v)$ for all $v\in V$, then $|(T\tpi)(v)-\pi(v)|\le \alpha \pi(v)$ for all $v\in V$.
	\item \label{item:add-approximation} If $|\tilde{\pi}(v)-\pi(v)|\le \alpha$ for all $v\in V$, then $|(T\tpi)(v)-\pi(v)|\le \alpha$ for all $v\in V$.
\end{enumerate}
\end{lemma}
\begin{proof}
We prove each of the claims separately.
\paragraph{Claim \eqref{item:mul-approximation}.}
Let $\Delta = \tpi - \pi$. The statement implies that $|\Delta(v)| \le \alpha \pi(v)$. Now we have
\[
	T \tpi = T(\pi + \Delta) = \pi + T \Delta.
\]
Therefore,
\[
	|(T\tpi)(v)-\pi(v)| = |(T \Delta)(v)|.
\]
	Let $T_v$ be the $v$-th row-vector of $T$. Hence, $|(T \Delta)(v)| = |T_v \Delta|$. We next obtain
\[
	|T_v \Delta| \le \sum_{u \in V} |T_{v, u} \Delta(u)| \le \alpha \sum_{u \in V} |T_{v, u} \pi(u)|.
\]
Using the fact that the entries of $T$ and of $\pi$ are non-negative, from the last chain of inequalities we derive
\[
	|T_v \Delta| \le \alpha \sum_{u \in V} T_{v, u} \pi(u) = \alpha T_v \pi = \alpha \pi(v),
\]
as desired.

\paragraph{Claim \eqref{item:add-approximation}.}
Similarly as above define $\Delta = \tpi - \pi$, so $T \tpi = \pi + T \Delta$. Moreover,
$|\Delta(v)| \le \alpha$, so
\[
	|(T\tpi)(v)-\pi(v)| = |(T \Delta)(v)|= |T_v\Delta| \le \sum_{u\in V} |T_{v,u} \Delta(u)| \le \alpha |T_v| \le \alpha.
\]
\end{proof}

Note that \cref{alg:random-walks} is parameterized by $G,l,t$ and in particular it does not take stationary distribution as an argument.
Instead, the stationary distribution which is encoded in $t$. To define $t$ in our earlier analysis of $\RandomWalks$ we used $k_i$ defined by~\eqref{eq:ki-l-log-l-1} and~\eqref{eq:ki-l-log-l-2}. Intuitively, $k_i$ defines the overhead of the number of random walks that we have to sample from each vertex so that the algorithm has sufficiently many random walks for the next iteration of doubling. When we do not have access to the exact value of $\pi$, but instead to its $(1 + \alpha)$-approximation, we have to use larger values of $k_i$ to accommodate that approximation. So, for $1\le i \le \lceil \log l \rceil$, we define
\begin{equation}
\label{eq:ki-delta}
k_i=(2+4\alpha)^{\lceil \log l\rceil- i}.
\end{equation}
\paragraph{Comparison between \eqref{eq:ki-delta} and \eqref{eq:ki-l-log-l-2}.} Note that for $\alpha=0$ the value of $k_i$ given by \eqref{eq:ki-delta} does not equal to the one given by \eqref{eq:ki-l-log-l-2}. We made such choice so to simplify the expression of $k_i$ when we have access to only an approximation of $\pi$. As a consequence of definition \eqref{eq:ki-delta}, the space requirement in our analysis depends on $1 / \alpha$, which can be very large for small values of $\alpha$. Nevertheless, if $\tpi$ is a $(1 + \alpha)$-approximation of $\pi$, then $\tpi$ is also a $(1 + \alpha')$-approximation for any $\alpha' \ge \alpha$. Hence, and without loss of generality, we can assume that $\alpha \ge 1/\log{n}$.

We are now ready to analyze properties of $\RandomWalks$ when it has access to only an approximation of $\pi$.
\begin{lemma}\label{theorem:alg-failure-inner-loop}
Let $G$ be a stochastic graph. Let $\pi$ be the stationary distribution of $G$, and $\eps$ be such that $\pi(v) \ge \eps / n$ for all $v$. Let $l, C \geq 1$ and $l = o(S)$. Finally, let $\alpha \in (0, 1/4]$ be a parameter, and let $\tpi$ be a probability distribution on $V$ such that $|\tpi(v) - \pi(v)| \le \alpha \pi(v)$ for all $v$.
Define $k_i$ as in~\eqref{eq:ki-delta} and $t_i(v) = \lceil C \tpi(v) n \ln n \cdot k_i\rceil$.
Then, \RandomWalks$(G, l, t)$ (\cref{alg:random-walks}) has the following properties:
\begin{enumerate}[(i)]
	\item\label{item:RandomWalks-approx-success-prob} The algorithm is an imperfect sampler (see~\cref{def:sampler}) for generating $\lceil C\tilde{\pi}(v) n\ln{n}\rceil$ length $l$ random walks starting from each vertex $v \in V$. The failure probability is at most $n^{-\frac{C\alpha \epsilon}{3}+2}e^2$.
	\item\label{item:RandomWalks-approx-rounds} The algorithm can be executed in $O(\log l)$ MPC rounds.
	\item\label{item:RandomWalks-approx-space} The space requirement of the algorithm is $O\rb{m+C l^{1+ 2 \alpha}n \ln{n}}$.
\end{enumerate}
\end{lemma}
\begin{proof}
	Property~\eqref{item:RandomWalks-approx-rounds} follows by \cref{lemma:ref-space-1}. Property~\eqref{item:RandomWalks-approx-space} follows from the fact that the space requirement is dominated by the first iteration of \RandomWalks. Hence, as $k_0 = (2+4\alpha)^{\lceil \log l \rceil}$, the space requirement is $O\rb{m+C (2+4\alpha)^{\log l} n \ln{n}} = O\rb{m+C l^{1+ 2 \alpha}n \ln{n}}$.	
	\paragraph{Property~\eqref{item:RandomWalks-approx-success-prob}.}
	At the end of $i$-th step of \cref{alg:random-walks}, $t_i(v)$ denotes the number of random walks of length $2^{i - 1}$ that originate at vertex $v$. At step $i$, each vertex $v$ doubles $t_i(v)$ walks arbitrarily chosen from $W_{i - 1}(v)$.
    Let $X^u_i$ be the number of these random walks ending at vertex $u$. Note that $X^u_i$ is a sum of $0/1$
   random variables $Y^v_{i, j}$, where $Y^v_{i, j}$ equals $1$ iff the $j$-th selected random walk of $W_{i - 1}(v)$ ends at $u$.
Let $T = T(G)$ be the transition matrix of $G$.
	From \cref{lemma:approximate-tpi-implies-approximate-Ttpi}, and as $T \pi = \pi$, we have
	\begin{equation}\label{eq:upper-bound-T^j t_i}
		(T^j t_i)(u) = (T^{j - 1} (T t_i))(u) \le (T^{j - 1} (T (C\tpi k_i n \log n + \allones)))(u) \le (1 + \alpha) C \pi(u) k_i n \log n + 1.
	\end{equation}
	Using this upper-bound, we further derive
	\begin{eqnarray}
	  |\ee{X^u_i} -t_i(u) | &\le & |\ee{X^u_i} -C \tpi(u)k_i n\ln n |+1  \nonumber \\
		& = & |(T^{(2^{i-1})} t_i)(u) -C \tpi(u)k_i n\log n | +1 \nonumber \\
    &\le& |(T^{(2^{i-1})} t_i)(u) -C (1 - \alpha) \pi(u)k_i n \ln n |+1 \nonumber  \\
	  &\stackrel{\text{from  } \eqref{eq:upper-bound-T^j t_i}}{\le} & |(1 + \alpha) C\pi(u)k_i n \ln n - C (1 - \alpha) \pi(u)k_i n \ln n |+2 \nonumber \\
		& \le &  2 \alpha C \pi(u) k_i n \ln n+2. \label{eq:the-final-bound-on-eeX^u_i-t_i(u)}
	\end{eqnarray}


	As long as $X^u_i \le t_{i - 1}(u)- t_i(u)$ holds for each vertex $u$, a vertex $u$ is able to (1)
    extend all the $X^u_i$ random walks that ended at $u$, and (2) double the length of $t_i(u)$ random walks from $U_{i}(u)$ staring in $u$.

    We have
    \begin{eqnarray*}
         t_{i - 1}(u)- t_i(u) &=& \lceil C \tpi(u) k_{i-1} \cdot n \ln{n} \rceil -\lceil C \tpi(u) k_{i} \cdot n\ln{n} \rceil \\
      &=&  \lceil C \tpi(u) \cdot (2+2\alpha)  k_{i} \cdot n \ln{n} \rceil - \lceil C \tpi(u) k_{i} \cdot n \ln{n} \rceil\\
      &=& \lceil 2 \alpha C \tpi(u) k_i  \cdot n \ln{n} -2\rceil +2 \lceil C \tpi(u) k_{i} \cdot n \ln{n} \rceil  -\lceil C  \tpi(u) k_{i} \cdot n \ln{n} \rceil\\
       &\ge &   4 \alpha C \tpi(u) k_{i} \cdot n \ln{n} - 2 + t_i(u)\\
			 &\ge &   4 \alpha (1 - \alpha) C  \pi(u) k_{i} \cdot n \ln{n} -2 + t_i(u)\\
       &\stackrel{\text{from \eqref{eq:the-final-bound-on-eeX^u_i-t_i(u)}}}\ge&  (2 \alpha - 4 \alpha^2) C \pi(u) k_{i} \cdot n  \ln{n} - 4 + \ee{X^u_i}\\
       &\stackrel{\text{from $\alpha \le 1/4$}}{\ge} &  \alpha C \pi(u) k_{i} \cdot n  \ln{n}  - 4 + \ee{X^u_i}\\
       &\ge& \alpha C \epsilon k_{i} \cdot \ln{n}  - 4 + \ee{X^u_i},
    \end{eqnarray*}
		where the last inequality follows from our assumption that $\pi(u) \ge \eps / n$.
    Now, by Chernoff bound, it holds
	\begin{eqnarray*}
\prob{X^u_i > t_{i - 1}(u)- t_i(u)} &\le& \prob{X^u_i > \ee{X^u_i}+  \alpha C \epsilon k_{i} \cdot \ln{n}-4 } \\
& \le& \exp{\rb{-\frac{ \alpha C\epsilon k_i \cdot\ln{n}-4}{3}}}\\
&\le& \exp{\rb{-\frac{\alpha C\epsilon}{3}\ln{n}}}e^2 = n^{-\frac{\alpha C\epsilon}{3}}e^2.
	\end{eqnarray*}
    By taking union bound over all vertices and all rounds of the algorithm the probability of failure is less than $n^{-\frac{C\alpha \epsilon}{3}+2}e^2$.
\end{proof}

\subsection{Lower Bound}
\label{sec:lower_bound}
In this section we show that our algorithm for sampling random walks in undirected graphs is conditionally optimal under the popular \cycleConj{} conjecture~\cite{BeameKS13,roughgarden2018shuffles}.
The conjecture states that any algorithm in the MPC model which distinguishes between a graph being a cycle of length $n$ from a graph consisting of two cycles of length $n/2$, and uses $O(n^\gamma)$ space per machine and $O(n^{1-\gamma})$ machines requires $\Omega(\log n)$ rounds.

\lowerbound*

The proof is based on the fact that by running $\Theta(\log n)$ random walks of length $\Theta(\log^4 n)$ from a vertex $v$ one discovers $\Theta(\log^2 n)$ nearest vertices to $v$ with high probability. This property follows from the following well known lemma.

\begin{lemma}\label{lem:rw}
Let $t > 1$ be an even integer and $X_1, \ldots, X_t$ be a sequence of i.i.d random variables, such that $P(X_i = 1) = P(X_i = -1) = 1/2$.
Let $X = \sum_{i=1}^t X_i$.
Then, $P(|X| \geq \sqrt{t}/2) = \Omega(1)$.
\end{lemma}

\begin{proof}
$(X-t)/2$ follows a binomial distribution with $p=1/2$.
Hence, the mode of the distribution of $X$ is $0$, that is for every integer $i$, $P(X = 0) \geq P(X = i)$.
Moreover, $P(X = 0) = \binom{t}{t/2}(1/2)^t \leq c/\sqrt{t}$, for some constant $c$ where the inequality follows from Stirling's approximation.
	Hence, $P(|X| \geq \sqrt{t}/(2c)) \geq 1-\sum_{i=-\lceil\sqrt{t}/(2c)\rceil}^{\lceil\sqrt{t}/(2c)\rceil} P(X = i) \geq 1 - \sum_{i=-\lceil\sqrt{t}/(2c)\rceil}^{\lceil\sqrt{t}/(2c)\rceil} P(X = 0) \geq 1 - (\sqrt{t}/(2c)+3) c / \sqrt{t} = 1 - 1/2 + 3c/\sqrt{t} = \Omega(1).$
\end{proof}

\begin{proof}[Proof of \cref{thm:lower-bound}]	Let us assume that there exists an algorithm which samples $\Theta(\log n)$ independent random walks of length $\Theta(\log^4 n)$ starting at each vertex and takes $f(n) = o(\log \log n)$ rounds.
Note that we allow the algorithm to use total space that is $\poly \log n$ factor more than the size of the graph (otherwise storing the output would not be possible).
We will use it to solve the problem from the problem from the \cycleConj{} conjecture more efficiently.

We begin by dealing with the following technical difficulty.
The \cycleConj{} conjecture considers a setting, in which the total space is linear in the graph size.
Thus, we first show an algorithm that in $O(\poly \log \log n)$ rounds shrinks the length of each cycle by a constant factor with probability $1-n^{-d}$, where the constant $d$ can be made arbitrarily large.
The algorithm samples a random bit $X_v$ for each vertex $v$.
Then, for each tree vertices $a, b, c$ which are adjacent on the cycle (or cycles), if $X_a = X_c = 0$ and $X_b = 1$, we connect $a$ and $c$ with an edge and remove $b$ from the graph.
It is easy to see that this indeed shrinks each cycle by a constant factor with probability at least $1-n^{-d}$, as long as the length of each cycle is $\omega(d \log n)$.
By running this algorithm $O(\log \log n)$ times, we can shrink each cycle by a factor of $\Omega(\poly \log n)$.
We obtain a graph having $n' = O(n / \poly \log n)$ vertices and edges, so from now on we can afford to use an algorithm whose total space usage is $O(n' \poly \log n')$, since $n' \poly \log n' = O(n)$.

	We now show how to use random walks to shrink the cycles further, namely shrink them by a factor of $\Omega(\log n)$ in $f(n)$ rounds.
Let us sample each vertex independently with probability $1 / \log n$.
With probability at least $1-n^{-d}$ among every $\Omega(\log^2 n)$ consecutive vertices along each cycle at least one vertex is sampled.
Our goal is to use random walks to contract the graph to the set of sampled vertices.
To that end, we run $\Theta(\log n)$ random walks of length $\Theta(\log^4 n)$ from each vertex and show that with high probability for each pair of consecutive sampled vertices on the cycle, there exists at least one random walks which visits both of them.
Consider a sampled vertex.
We show that with high probability, random walks starting at this vertex visits the two neighboring sampled vertices along the cycle.
Indeed, since we have sample $\Theta(\log n)$ random walks per vertex, it suffices to show that each random walk visits one of the neighboring vertices with constant probability.
	This in turn follows from the fact that with constant probability, a random walk on a line of length $\Omega(t)$ ends in some vertex at distance $\sqrt{t}$ with constant probability (see \cref{lem:rw}).

As a result, in $f(n)$ rounds we are able to contract each cycle to the set of sampled vertices, whose size is clearly $O(n / \log n)$ with high probability.
By repeating this step $\Theta(\log n / \log \log n)$ times, we can reduce the total number of vertices to $o(n^\gamma)$, at which point the remaining graph can be sent to one machine to check whether it is a cycle or two cycles.
The total number of rounds is $\Theta(\log n / \log \log n) f(n) = o(\log n)$.

\end{proof}

\section{PageRank}\label{sec:PageRank}


In this section we show MPC algorithms for computing PageRank both in undirected and directed graphs, that is we prove \cref{theorem:PageRank-undirected} and \cref{theorem:PageRank-directed}.
As an easy corollary we obtain an algorithm for sampling random walks in directed graphs (\cref{theorem:directed-walks}).

We are going to use the following most basic algorithm for estimating PageRank using random walks.
Let $0\le \epsilon \le 1$ be a parameter. We are going to sample random walks form a stochastic graph
\begin{equation}\label{eq:definition-T-eps}
	G_\eps = (1-\eps)G + \frac{\epsilon}{n} J,
\end{equation}
where $J$ is a complete directed graph on the vertex set $V(G)$ (containing a self loop in every vertex).
In other words, with probability $\epsilon$ we jump to a random vertex, and
with probability $1-\epsilon$ we walk according to edges of $G$.
\begin{definition}[Jump transition]\label{def:transitions}
Let $G_\eps$ be the stochastic graph defined by \eqref{eq:definition-T-eps}. Then, \emph{jump transition} refers to the transition performed within the graph $J$.
\end{definition}

Consider the transition matrix $T_\eps = T(G_\eps)$.
The stationary distribution $\pi_\eps$ of $T_\eps$ satisfies $T_\eps \pi_\eps = \pi_\eps$, which implies the following equation
\[
\rb{I-(1-\epsilon)T(G)}\pi_\epsilon=\frac{\epsilon}{n}\vec{1}.
\]

The crucial property of $\pi_\epsilon$ is that the probabilities of ending
in a given vertex do not decrease much relatively to decrease of $\epsilon$, as shown
by the following lemma.

\begin{lemma}
\label{lemma:pi-smooth}
For any $0<\delta\le 1$  we have
\[
\pi_{\epsilon\cdot \delta} \ge\delta\cdot \pi_\epsilon,
\]
where inequality is taken over all coordinates.
\end{lemma}
This result follows from the Taylor expansion of $\pi_\epsilon$
\begin{equation}\label{eq:taylor-pagerank}
	\pi_\epsilon = \frac{\epsilon}{n} \sum_{i=0}^{\infty} ((1-\epsilon) T(G))^i\vec{1}.
\end{equation}
\begin{proof}
Using \cref{eq:taylor-pagerank} we get
\begin{eqnarray}
	\pi_{\epsilon\cdot \delta} &=& \frac{\epsilon \cdot \delta}{n} \sum_{i=0}^{\infty} ((1-\epsilon\delta) T(G))^i  \vec{1}=\delta\rb{\frac{\epsilon}{n} \sum_{i=0}^{\infty} (1-\epsilon\delta)^i T(G)^i\vec{1}} \\
	&\ge& \delta \rb{\frac{\epsilon}{n} \sum_{i=0}^{\infty} (1-\epsilon)^i T(G)^i\vec{1} } =\delta \rb{\frac{\epsilon}{n} \sum_{i=0}^{\infty} ((1-\epsilon) T(G))^i\vec{1}} =\delta\cdot \pi_\epsilon.
\end{eqnarray}
\end{proof}
The next follows from the observations that $\pi_1(v)=\frac{1}{n}$.
\begin{corollary}
\label{cor:lower-bound-pagerank}
For any $0< \epsilon \le 1$ and any $v\in V$ we have
\[
\pi_{\epsilon}(v)\ge \frac{\epsilon}{n}.
\]
\end{corollary}

The Taylor expansion \cref{eq:taylor-pagerank} suggests the following algorithm for estimating PageRank (see e.g., \cite{Breyer02markovianpage}).

\begin{algorithm}[H]
\begin{algorithmic}[1]
\Function{\StationaryDistribution}{$W, \eps$}
\ForAll{$v\in V$ in parallel\label{line:StationaryDistribution-keep-first-K}}
	\State{Remove from $W$ all but $K = \left \lceil \frac{9\ln n }{\eps \alpha^2} \right \rceil$ walks starting in $v$. If $W$ does not contain enough walks, then ``fail''.}
\EndFor
\State{Truncate each walk in $W$ just before the first jump transition (see \cref{def:transitions}). \label{line:StationaryDistribution-truncate}}
\ForAll{$v\in V$ in parallel\label{line:StationaryDistribution-compute-tpi}}
\State{Let $n_v$ be the number of the walks from $W$ ending in $v$.}
\State{$\tilde{\pi}(v) \gets \frac{n_v}{Kn}$.}
\EndFor
\State{Return $\tilde{\pi}$}
\EndFunction
\end{algorithmic}
\caption{An algorithm for approximating the PageRank with damping factor $1 - \eps$ by using a set of random walks $W$.}
\label{alg:stationary-distribution}
\end{algorithm}

\begin{algorithm}[H]
\begin{algorithmic}[1]
	\State{Sample a set $W$ of $K = \left \lceil \frac{9\ln n }{\eps \alpha^2} \right \rceil$ random walks starting from each vertex of $G_\epsilon$ of length $l= \left \lceil \frac{9 \ln n}{\epsilon}\right \rceil$. \label{line:sample-Kn-walks}}
	\State Return \StationaryDistribution($W, \eps$)
\end{algorithmic}
\caption{An algorithm for approximating PageRank using random walks.}
\label{alg:random-walks-directed-main-2}
\end{algorithm}


\begin{lemma}
\label{lemma:no-jump}
Let $\eps \leq 1$ and assume we mark each edge of a walk of length  $l=\lceil \frac{9 \ln n}{\epsilon}\rceil$ independently
with probability $\eps$.  The probability that no edge is marked is
less than $\frac{1}{n^9}$.
\end{lemma}
\begin{proof}
$(1-\epsilon)^l \le (1-\epsilon)^{9 \frac{\ln n}{\epsilon}}\le e^{-9 \ln n} = \frac{1}{n^9}$.
\end{proof}

Let us prove the approximation ratio obtained by \cref{alg:random-walks-directed-main-2}.
Note that the lower bounds for $\alpha$ and $\epsilon$ are not limiting, since the interesting values for both parameters are constant.

\begin{lemma}
\label{lemma:sampling-pagerank}
Let $\alpha \in [1 / n, 1 / 4]$, $\eps \in [1/n, 1]$, and $0<\alpha<1$.
Denote by $\tilde{\pi}$ the output of \cref{alg:random-walks-directed-main-2}.
Then, $|\tilde{\pi}_{\epsilon}(v)-\pi_{\epsilon}(v)|\le \alpha \pi_{\epsilon}(v)$ for all $v\in V$
(i.e., $\tpi$ is $(1+\alpha)$-approximation of $\pi_\epsilon$)
with probability at least $1- \frac{4}{n^2}$.
\end{lemma}
\begin{proof}
The probability that some walk is not truncated by \cref{lemma:no-jump} is $\frac{1}{n^9}$, as jump transition happens with probability $\eps$. Hence, by union bound over $nK \le 10n^5$ walks some walk is not truncated with probability at most $\frac{10}{n^5}\le \frac{2}{n^2}$.

Observe that $n_{v}$ is a sum of $nK$ $0,1$-random variables. Moreover, its expectation by \cref{cor:lower-bound-pagerank} is $\ee{n_v}\ge\frac{\epsilon}{n}\cdot nK =\epsilon K$, so by Chernoff bound (see \cref{lemma:chernoff}\eqref{item:delta-at-most-1})
\begin{eqnarray*}
  \Pr\left[\left|n_v-\ee{n_v}\right| \le \alpha \ee{n_v}\right] &\le&  2 \exp\rb{-\frac{\alpha^2 \ee{n_v}}{3}} \le 2 \exp\rb{-\frac{\alpha^2\epsilon K}{3}} \\
   &\le& 2 \exp\rb{-3\ln(n)}=\frac{2}{n^3}.
\end{eqnarray*}
Again, by union bound over all $n$ vertices some estimate is incorrect with probability at most $\frac{2}{n^2}$. Hence, the total probability of failure is $\frac{4}{n^2}$.
\end{proof}
\StationaryDistribution also has an efficient MPC implementation.
\begin{lemma}\label{lemma:StationaryDistribution-rounds}
	\StationaryDistribution (\cref{alg:stationary-distribution}) can be implemented in $O(1)$ MPC rounds.
\end{lemma}
\begin{proof}
	To implement \cref{line:StationaryDistribution-keep-first-K}, we invoke \NumberingSublists on the walks from $W$; the numbering of walks is performed with respect to their starting vertices. (See \cref{section:MPC-Implementation-RandomWalks} for details about \NumberingSublists.) Those walks that get number larger than $K$ are removed. This can be implemented in $O(1)$ rounds.
	
	\cref{line:StationaryDistribution-truncate} is performed without additional communication.
	
	To simulate the loop on \cref{line:StationaryDistribution-compute-tpi}, on the remaining walks we again invoke \NumberingSublists, but this time we perform numbering of the walks with respect to their ending vertices. For each vertex $v$, we calculate the maximum among those numbers. That maximum value equals $n_v$. These operations can be performed in $O(1)$ rounds.
\end{proof}

\prundirected*

\begin{proof}
We set $K=\left \lceil \frac{9 \ln n}{\epsilon \delta^2} \right \rceil$, $l= \left \lceil \frac{9 \ln n}{\eps} \right \rceil $ and  $C=\frac{6}{\epsilon\alpha^2}$.
We first execute \cref{alg:random-walks} with $C$ and $l$. Next, we give sampled walks to \cref{alg:random-walks-directed-main-2} with $K$ and $l$.

\paragraph{Space requirement.}
By \cref{lemma:ref-space-1} we require $O(C m l \log l\log n) = O\rb{\frac{m\log^2n \log \log n}{\epsilon^2\alpha^2}}$ total space.

\paragraph{Round complexity.}
By \cref{lemma:ref-space-1,lemma:StationaryDistribution-rounds} we require $O(\log l)= O(\log \log n + \log 1/\eps)$ rounds.

\paragraph{Success probability.}
On one hand, by \cref{theorem:llogl} the algorithm
fails with probability at most $n^{-\frac{C}{3}+3} = n^{-\frac{2}{\epsilon \alpha^2}+3}\le n^{-32+3} =n^{-29}$.
On the other hand, by \cref{lemma:sampling-pagerank} we obtain $(1+\alpha)$-approximation of $\pi$ with probability at least $1-\frac{4}{n^2}$.
Hence, the final success probability is at least $1-\frac{5}{n^2}$.
\end{proof}

\subsection{Directed Balanced Graphs}
\label{sec:our-algorithm-directed-PageRank}
The first step towards our algorithm for directed graphs is considering balanced directed graphs.
A directed graph $G$ is called {\em $c$-balanced} when for all $v\in V$ we have $c \deg^+(v)\ge \deg(v)$.
In particular in $c$-balanced graphs there are no dangling vertices, i.e., vertices that do not have any out edge.
The idea is to first consider $\undir{G}$ -- the undirected closure of $G$. In $\undir{G}$ we can compute
long walks fast and then gradually move towards directed graph $G$. The $c$-balanced
property allows as to argue that each edge from a random walk in $\undir{G}$ is with probability
at least $1/c$ directed according to $G$, enabling us to prove the following result.
\begin{restatable}{theorem}{PageRankcbalanced}\label{theorem:PageRank-directed-c-balanced}
let $G$ be a $c$-balanced graph.
Let $\alpha \in [1 / n, 1 / 4]$,  $\eps \in [\log n/o(S), 1]$, and $\delta \in (0,1]$ such that $\delta^{-1} \in \bbN$.
There exists an MPC algorithm that with probability at least $1-\frac{12}{\delta n^2}$ computes $(1+\alpha)$-approximate $\eps$-PageRank vector of $G$ in $O\rb{\delta^{-1} (\log \log n + \log 1/\eps)}$ rounds using $O\rb{\frac{m\ln^2 n \ln \ln n}{\epsilon^2 \alpha^2}+n^{18 \frac{c\delta}{\eps}} \frac{n \ln^{2.5} n}{\epsilon^{3.5} \cdot \alpha^2}}$ total space and strongly sublinear space per machine.
\end{restatable}

For this we need few more definitions.
We are going to sample random walks from the stochastic graph defined as follows

\begin{equation}
  \label{eq:definition-T-eps-sigma}
	G_{\epsilon,\sigma} = (1-\epsilon)\sigma \rw{\undir{G}} + (1-\epsilon)(1-\sigma) \rw{G} + \frac{\epsilon}{n} J.
\end{equation}
We denote by $\pi_{\epsilon,\sigma}$ the stationary distribution of $G_{\epsilon,\sigma}$. In our algorithms we will be using $G_{\eps,\sigma}$, but the graph corresponding to $J$ will be constructed only implicitly. That is, instead of constructing graph $J$ explicitly, which contains $n^2$ edges, each vertex $v$ of $G_{\epsilon,\sigma}$ will hold a value $\eps/n$ implying that $v$ has an edge of weight $\eps/n$ to each vertex.

We also note that $G_{\epsilon,\sigma}$ can be constructed in $O(1)$ MPC rounds. Namely, we first broadcast $\eps$ and $\sigma$ to each machine, which can be done in $O(1)$ rounds as described in \cite{goodrich2011sorting}. Then, each edge is copied and annotated so to construct $(1-\epsilon)\sigma \rw{\undir{G}} + (1-\epsilon)(1-\sigma) \rw{G}$. Finally, each vertex is annotated by $\eps$ which, as described, suffices to implicitly construct $\frac{\epsilon}{n} J$.
\begin{observation}\label{observation:constructing-stochastic-graph}
	A stochastic graph as defined by \eqref{eq:definition-T-eps-sigma} can be implicitly constructed in $O(1)$ MPC rounds.
\end{observation}
We now define three transition types that capture different components of stochastic graphs.
\begin{definition}[Transition types]\label{def:transitions2}
Let $G_{\eps,\sigma}$ be the stochastic graph as defined by \eqref{eq:definition-T-eps-sigma}.
	Each edge of $G_{\eps,\sigma}$ originates from one of the graphs $G$, $\undir{G}$ and $J$.
	We call an edge of $G_{\eps,\sigma}$ a
	\begin{itemize}
		\item \emph{directed transition}, if it originates from $G$, and
		\item \emph{undirected transition}, if it originates from $\undir{G}$, and
		\item \emph{jump transition}, if it originates from $J$.
	\end{itemize}
\end{definition}

In the following we assume that each edge of each walk in $G_{\eps,\sigma}$ has its edges labeled with the transition types defined above.

In the main algorithm of this section (see \cref{alg:random-walks-directed-main}), we start from $\sigma=1$ and then
gradually decrease the value to obtain $\sigma=0$. This
sequence is defined as
\[
\sigma_{j+1} = \sigma_j - \delta,
\]
for $1\le j \le \delta^{-1}$. (We will set $\delta$ so that $\delta^{-1}$ is an integer.)
We now state the main algorithm of this section.
It uses \RandomWalks (\cref{alg:random-walks}) as a subroutine.

\subsubsection{Algorithms}
We now describe our algorithms used to compute approximate PageRank for $c$-balanced graphs. The main algorithm is \PageRankOfBalancedGraphs (see \cref{alg:random-walks-directed-main}) that essentially repeats \cref{alg:random-walks-directed-main-2} (see the steps within loop \cref{line:mixing-loop}) $\delta^{-1}$ many times. This loop implements the gradual change from an undirected to the corresponding directed graph.

\begin{algorithm}[H]
\begin{algorithmic}[1]
\Function{\PageRankOfBalancedGraphs}{$G$, $\eps$, $\alpha$, $\delta$}
	\State{Compute approximate stationary distribution $\tilde{\pi}_{\epsilon,1}$ using \cref{theorem:PageRank-undirected}. \label{line:directed-PageRank-initial-tpi}}
    \State{$l\gets \left \lceil \frac{9 \ln n}{\epsilon} \right \rceil $}
    \For{$j \gets 1\ldots \delta^{-1}$ \label{line:mixing-loop}}
	\State Implicitly compute $G_{\eps, \sigma_j}$ using \cref{eq:definition-T-eps-sigma}. \label{line:PageRank-balanced-compute-stochastic}
        \State Run \textsc{RandomWalks}$(G_{\epsilon, \sigma_j}, l, t)$ where $t_i(v) = \lceil C \tilde{\pi}_{\epsilon, \sigma_j}(v) n \ln n \cdot k_i\rceil$ and $k_i$ is defined by~\eqref{eq:ki-delta}. Let $W$ be the set of resulting walks.\label{line:PageRank-run-RandomWalks}\Comment{$C$ is defined in \cref{theorem:PageRank-directed-c-balanced}}
        \State $W_T\gets \emptyset$
        \ForAll{$w\in W$ in parallel}
            \If{$w_T=\TranslateWalk(G, \eps, j, w)$ did not ``fail'' \label{line:invoke-SampleRandomWalk}}
                \State{Add $w_T$ to $W_T$. \label{line:define-WT}}
            \EndIf
        \EndFor
				\State{$\tilde{\pi}_{\epsilon,\sigma_{j+1}} = \StationaryDistribution(W_T, \eps)$ \label{alg:random-walks-directed-main-fails}}
	\EndFor
\EndFunction
\end{algorithmic}
\caption{An algorithm for computing $(1 + \alpha)$-approximate $\eps$-PageRank of a $c$-balanced graph $G$.}
\label{alg:random-walks-directed-main}
\end{algorithm}
	
As the main primitive, \PageRankOfBalancedGraphs invokes $\TranslateWalk$ (\cref{alg:translate-walk}) on \cref{line:invoke-SampleRandomWalk}, which takes a PageRank walk in $\Gepsj{j}$ and either returns a PageRank walk in $\Gepsj{j + 1}$ or fails.
In the following we say that it \emph{translates} the given walk.
Each translation can fail with relatively large probability, but we will take enough walks so that the whole process has a small failure probability.

\begin{algorithm}[H]
\begin{algorithmic}[1]
\Function{$\TranslateWalk$}{$G$, $\eps$, $j$, $w$}
	\If{$\sigma_j < 1/2$}
        \State{Let $d$ be the number of directed transitions in $w$.}
        \State{Let $u$ be the number of undirected transitions in $w$.}
        \State{$p \gets r^l \frac{\sigma_{j+1}^{u}(1-\sigma_{j+1})^d}{\sigma_j^{u} (1-\sigma_j)^d}$}
        \Comment{$r$ is set in \cref{lemma:prob-bound-on-repeating-directed-pi} to guarantee that $p\le 1$.}
        \State{With probability $p$ return $w$; otherwise, return ``fail''.}

	\Else
		\State{Let $w_T$ equal $w$. We annotate the transitions of $w_T$ as follows.}
		\For{each edge $e = uv$ of $w_T$}
			\If{$e$ is an undirected transition in $w$ \label{line:TranslateWalk-e-in-G_U}}
				\State{-- With probability $1 - \rho(v)$ we keep $e$ in $w_T$ as an undirected transition. \Comment{We define $\rho(v)$ later in \cref{lemma:prob-bound-on-repeating-directed-pi}.}}
				\State{-- With probability $\rho(v)$ do the following: If $uv$ is an arc in $G$, then we replace $e$ in $w_T$ with this directed edge; otherwise, return ``fail''.}
			\ElsIf{$e$ is a directed transition in $w$ \label{line:TranslateWalk-e-in-G}}
				\State{Keep $e$ in $w_T$ as a directed transition.}
			\ElsIf{$e$ is a random jump in $w$ \label{line:TranslateWalk-e-random}}
				\State{With probability $\beta(v)$ return ``fail'', and otherwise keep $e$ in $w_T$ as a random jump. \Comment{We define $\beta(v)$ later in \cref{lemma:prob-bound-on-repeating-directed-pi}, \eqref{eq:define-beta(v)}.}}
			\EndIf
		\EndFor
		\State{\Return $w_T$ together with its transition types}
	\EndIf
\EndFunction
\end{algorithmic}
\caption{Given a random walk $w$ in $\Gepsj{j}$, return a random walk in $\Gepsj{j + 1}$ or ``fail''.}
\label{alg:translate-walk}
\end{algorithm}

In the rest of this section we analyze \cref{alg:random-walks-directed-main} and show that it satisfies the claim given in \cref{theorem:PageRank-directed}.


\subsubsection{The Success Probability of \cref{alg:random-walks-directed-main}}
%


We now prove the correctness of $\TranslateWalk$, which transforms a random walk $w$ in $\Gepsj{j}$ to a random walk in $\Gepsj{j + 1}$, or returns ``fail''.
We show that the output walk is a random walk in $\Gepsj{j + 1}$ and that the probability of ``failing'' is relatively small.

\begin{lemma}\label{lemma:prob-bound-on-repeating-directed-pi}
Let $w$ be a random walk of length $l$ in $\Gepsj{j}$.
Assume that $\delta \leq 1/(2c)$.
Then,
$\TranslateWalk(w)$ (\cref{alg:translate-walk}) does not fail and outputs a random walk $w_T$ of length $l$ in $\Gepsj{j + 1}$ with probability at least $(1 - 2 c\delta)^{l}$.
\end{lemma}

The proof analyzes two separate cases.
In one of them we use the following procedure, which is often called \emph{rejection sampling}
\begin{lemma}\label{lem:rejection-sampling}
Consider two discrete probability spaces $P$ and $Q$ over the same space $\{e_1, \ldots, e_n\}$.
For $1 \leq i \leq n$, let $p_i$ and $q_i$ give the probabilities of $e_i$ in $P$ and $Q$ respectively.
Finally, let $0 \leq r \leq \min_{1 \leq i \leq n, q_i \neq 0} p_i / q_i$.

Consider an algorithm which given a random sample $X = e_i$ from $P$ returns it with probability $q_i / p_i \cdot r$, and does nothing with the remaining probability.
Then, the algorithm returns a result with probability $r$ and each time it does so, it is an element sampled according to $Q$.
\end{lemma}

\begin{proof}
We first note that the probability $q_i / p_i \cdot r$ is well-defined.
Indeed, it is only evaluated when $X_i = e_i$, which implies $p_i > 0$.
Moreover, thanks to the choice of $r$, it does not exceed $1$.

The probability of the algorithm not failing is
\[ \sum_{i=1}^n p_i \cdot q_i / p_i \cdot r = \sum_{i=1}^n q_i \cdot r = r. \]
Moreover, the probability of returning $e_j$ is $p_j \cdot q_j / p_j \cdot r = q_j \cdot r$.
Hence, if we condition on the algorithm not failing, the probability of returning $e_j$ is $q_j$, as desired.
\end{proof}

We can now prove~\cref{lemma:prob-bound-on-repeating-directed-pi}.

\begin{proof}[Proof of~\cref{lemma:prob-bound-on-repeating-directed-pi}]


For the proof we consider a procedure that translates the input walk $w$ in $l$ steps, one edge at a time.
We are going to show that after $i$ steps either the procedure has failed or the first $i$ edges of the walk are now sampled according from $\Gepsj{j+1}$.
We split the proof into two cases: $\sigma_j < 1/2$ and $\sigma_j \ge 1/2$.

\paragraph{Case $\sigma_j < 1/2$.}
Observe that in this case $\TranslateWalk(w)$ either returns the walk $w$ or fails.
Recall that we consider a procedure that translates the consecutive edges of the walk one by one.
At each step we are going to use \cref{lem:rejection-sampling}.
Each edge of the walk is either a jump transition (with probability $\epsilon$), directed transition (with probability $(1-\eps)\sigma$) or an undirected transition (with probability $(1-\eps)(1-\sigma)$.

Fix a vertex $v$ and let $e$ be some transition that can be chosen by a random walk that has reached $v$.
Denote by $f(\sigma, e)$ the probability of choosing $e$.
In order to use \cref{lem:rejection-sampling}, we need to consider the ratios $f(\sigma_j, e) / f(\sigma_{j+1}, e)$.
For any jump transition $e_j$ we have $f(\sigma_j, e_j) / f(\sigma_{j+1}, e_j) = 1$.
	For a directed transition $e_d$, $f(\sigma_j, e_d) / f(\sigma_{j+1}, e_d) = ((1-\eps)(1-\sigma_j)) / ((1-\eps)(1-\sigma_{j+1})) = (1-\sigma_j)/(1-\sigma_{j+1})$.
Finally, for an undirected transition $e_u$ we have $f(\sigma_j, e_u) / f(\sigma_{j+1}, e_u) = \sigma_j / \sigma_{j+1}$.
Note that out of the three ratios, only the second one is smaller than $1$ and hence is the smallest one.
Hence, for any transition $e$ we have
\begin{eqnarray*}
	\frac{f(\sigma_j, e)}{f(\sigma_{j+1}, e)} \ge \frac{1-\sigma_j}{1-\sigma_{j+1}} = \frac{1-\sigma_j}{1-\sigma_j + \delta} = 1-\frac{\delta}{1-\sigma_j+\delta} \geq 1-\frac{\delta}{\frac{1}{2}+\delta} \geq \frac{1}{1+2\delta}.
\end{eqnarray*}

In order to use~\cref{lem:rejection-sampling} we are going to set $r = 1/(1+2\delta)$.
The definition of $r$ together with the sampling ratios define the probability that the algorithm of~\cref{lem:rejection-sampling} does not fail.
Altogether, if the input walk consists of $l$ edges, has $d$ directed transitions and $u$ undirected transitions, then the probability that none of the $l$ applications of~\cref{lem:rejection-sampling} fails is
\[
	r^l \frac{(1-\sigma_{j+1})^d\sigma_{j+1}^u}{(1-\sigma_j)^d\sigma_j^u}.
\]

Observe that this is exactly what $\TranslateWalk$ does, which implies that the procedure is correct in the case when $\sigma_j < 1/2$.
	Finally, to bound the success probability note that $1/(1+2\delta)^l \geq (1-2c\delta)^l$.

\paragraph{Case $\sigma_j \ge 1/2$.}

As in the previous case, we prove the claim separately for each edge of the walk.
Differently from the previous case, though, in this case the translation procedure may actually change the transition types along the walk.

Let $w$ be a random walk in $\Gepsj{j}$.
Consider an edge $e$ belonging to the walk $w$, which the walk traverses after reaching vertex $v$.
Observe that \cref{line:TranslateWalk-e-random} translates $e$ to an undirected transition only if $e$ is an undirected transition and a coin toss with probability $1-\rho(v)$ succeeds.
	Thus, this happens with probability $p_U(v) = (1-\eps) \sigma_j (1 -\rho(v))$.

On the other hand, $e$ is translated to a directed transition either when $e$ is a directed transition, or when it is an undirected transition, and coin toss with probability $1 - \rho(v)$ fails, and there exists a corresponding directed transition.
The latter happens with probability $\deg^+_G(v) / \deg_G(v)$, since in this case $e$ is a randomly chosen undirected transition incident to $v$.
	Hence, the overall probability of translating $e$ to a directed transition is $p_D(v) = (1-\eps)\rb{1 - \sigma_j + \tfrac{\deg^+_G(v)}{\deg_G(v)} \rho(v) \sigma_j}$.



Since the output of \cref{line:TranslateWalk-e-random} should be a random walk in $\Gepsj{j + 1}$, we must have
	\begin{equation}\label{eq:p_U-and-p_D-ratio}
		\frac{p_U(v)}{p_D(v)} = \frac{(1-\eps)(\sigma_j - \delta)}{(1-\eps)(1 - \sigma_j + \delta)}.
	\end{equation}
This allows us to derive $\rho(v)$, let $r(v) = \tfrac{\deg_G(v)}{\deg^+_G(v)}$ and let $y(v) = \rho(v) / r(v) = \tfrac{\deg^+_G(v)}{\deg_G(v)} \rho(v)$.
Then, from~\eqref{eq:p_U-and-p_D-ratio} we have
	\begin{eqnarray}
		& & \sigma_j (1 -\rho(v)) \rb{1 - \sigma_j + \delta} = \rb{1 - \sigma_j + y(v) \sigma_j} \rb{\sigma_j - \delta} \nonumber \\
		& \implies & y(v) \sigma_j \rb{\sigma_j - \delta + r(v) - r(v) \sigma_j + r(v) \delta} = \delta \nonumber \\
		& \implies & y(v) \sigma_j = \frac{\delta}{\sigma_j - \delta + r(v) - r(v) \sigma_j + r(v) \delta}, \label{eq:value-of-y(v)}
	\end{eqnarray}
which in particular gives us a formula for $\rho(v)$.
In the following we verify $\rho(v) \in [0, 1]$.
To that end, we upper-bound $y(v)$. Given that $\sigma_j \ge 1/2$, \eqref{eq:value-of-y(v)} implies
	\begin{equation}\label{eq:value-of-y(v)-upper-bound}
		y(v) \le \frac{2 \delta}{\sigma_j - \delta + r(v) - r(v) \sigma_j + r(v) \delta}.
	\end{equation}
The values $r(v)$ and $\delta$ are fixed for a given vertex $v$. We upper-bound $y(v)$ from \eqref{eq:value-of-y(v)-upper-bound} by minimizing the denominator of \eqref{eq:value-of-y(v)-upper-bound}. Since $r(v)$ appears in the form $r(v) (1 - \sigma_j) + r(v) \delta > 0$ in the denominator, and $r(v) \in [1, c]$, the denominator is minimized for $r(v) = 1$. For $r(v) = 1$, the denominator of~\eqref{eq:value-of-y(v)-upper-bound} becomes $1$.
This implies that $y(v)$ and $\rho(v)$ are nonnegative.
Moreover,
	\begin{equation}\label{eq:upper-bound-on-alpha}
		y(v) \le 2 \delta \quad \text{and} \quad \rho(v) \le 2 r(v) \delta \le 2 c \delta \leq 1.
	\end{equation}

	Now we can upper bound the probability that translation of an edge sampled from $\Gepsj{j}$ to an edge sampled from $\Gepsj{j + 1}$ ``fails'', conditioned on $e$ not being a random jump. Denote that probability by $\beta(v)$. Then, we have
	\begin{eqnarray}
		\beta(v) & = & 1 - p_U(v) - p_D(v) = 1- \sigma_j (1 -\rho(v)) - \rb{1 - \sigma_j + \frac{\deg^+_G(v)}{\deg_G(v)} \rho(v) \sigma_j} \nonumber \\
    &=& \sigma_j \rho(v) \rb{1-\frac{\deg^+_G(v)}{\deg_G(v)}} \le \rho(v) \stackrel{\eqref{eq:upper-bound-on-alpha}}{\le} 2c\delta. \label{eq:define-beta(v)}
	\end{eqnarray}
		
		We now comment on random jumps. From our analysis, a transition which is not a random jump is rejected (i.e., ``fails'') with probability $\beta(v)$. To account for that, if a transition of $w$ is a random jump, $\TranslateWalk$ will also ``fail'' with probability $\beta(v)$, and with the remaining probability keep this random jump. In this way, conditioned that $\TranslateWalk$ \emph{does not} ``fail'', an edge $e_T$ of the output of $\TranslateWalk(w)$ is a random jump with probability $\eps$, a directed edge with probability $(1 - \eps) (1 - \sigma_{j + 1})$, and an undirected edge with probability $(1 - \eps) \sigma_{j + 1}$.
		
	Hence, $\TranslateWalk$ outputs ``fail'' per edge with probability $\eps \beta(v) + (1 - \eps) \beta(v) = \beta(v) \le 2 c \delta$, and hence an invocation of this method does not output ``fail'' with probability at least $(1 - 2 c \delta)^l$.

\end{proof}


\subsubsection{Proof of \cref{theorem:PageRank-directed-c-balanced}}
\label{section:proof-of-PageRank-directed-c-balanced}
We now prove \cref{theorem:PageRank-directed-c-balanced} by showing that \cref{alg:random-walks-directed-main} has the properties given in the theorem statement. To show correctness, observe that an iteration of the loop of \cref{alg:random-walks-directed-main} simulates algorithm \cref{alg:random-walks-directed-main-2} for computing $\tilde{\pi}_{\epsilon, \sigma_{j+1}}$.

We split the rest of the proof into three parts: the space requirements; the round complexity; and, the probability of success. Throughout the proof, we set parameters as $l= \left \lceil  \frac{9 \ln n}{\epsilon}\right \rceil$ and $K= \left \lceil \frac{9 \ln n}{\epsilon\cdot \alpha^2} \right \rceil\le \frac{10 \ln n}{\epsilon\cdot \alpha^2}$. The proof of the success probability fixes the value of $C$.

\paragraph{Round complexity.}
\cref{alg:random-walks-directed-main} executes $O(\delta^{-1})$ iterations. Each iteration implicitly constructs a stochastic graph on \cref{line:PageRank-balanced-compute-stochastic}, which by \cref{observation:constructing-stochastic-graph} can be done in $O(1)$ rounds. Also in each iteration is invoked \cref{alg:random-walks}, which takes $O(\log{l}) = O(\log \log n + \log 1/\eps)$ rounds.
Since we assume that each walk is stored entirely on a machine, \TranslateWalk on \cref{line:invoke-SampleRandomWalk} can be implemented without extra communication. To obtain $W_T$ defined by \cref{line:define-WT}, each walk that ``fails'' is marked by flag \fail, and otherwise marked by flag \succeed. Those walks marked by \succeed define $W_T$.

By \cref{lemma:StationaryDistribution-rounds}, \cref{alg:random-walks-directed-main-fails} can be implemented in $O(1)$ rounds.
Hence, the total round complexity is $O\rb{\delta^{-1} (\log \log n + \log 1/\eps)}$.

\paragraph{Success probability.}
By \cref{lemma:sampling-pagerank}, the probability that any $\tilde{\pi}_{\epsilon, \sigma_j}$ is not $(1+\alpha)$-approximation of $\pi_{\epsilon, \sigma_j}$ is at most $\frac{3}{n^2}$.

The next place when \cref{alg:random-walks-directed-main} can fail is \cref{alg:random-walks-directed-main-fails}, i.e., we need to have
at least $K$ walks staring in each vertex $v$. Note that algorithm generates
\begin{eqnarray*}
\lceil C \tilde{\pi}_{\epsilon, \sigma_j}(v) n \ln n k_{\lceil \log l\rceil}\rceil &\ge& C \tilde{\pi}_{\epsilon, \sigma_j}(v) n \ln n \\
&\ge& C\pi_{\epsilon, \sigma_j}(v)(1-\alpha) n \ln n \ge C\epsilon (1-\alpha)  \ln n.
\end{eqnarray*}
walks from each vertex. These walks are subsampled with probability at least $(1-2c\delta)^l$ by~\cref{lemma:prob-bound-on-repeating-directed-pi}.
We will aim to have $2K$ walks in expectations, so that by Chernoff bound (\cref{lemma:chernoff} \eqref{item:delta-at-most-1-ge})
the probability of this number being smaller than $K$ is
\[
\exp\rb{- \delta^2 \mu / 3} = \exp\rb{-K/3} \le \exp\rb{-3\ln n} = \frac{1}{n^3}.
\]
Using union bound over all vertices we will have not enough walks with probability at most $\frac{1}{n^2}$. Hence, we need to set $C$ so that
\[
2K \le (1 - 2 c\delta)^{l} \cdot C\epsilon (1-\alpha)  \ln n.
\]
This gives
\[
C \ge \frac{20}{\epsilon^2\alpha^2(1-\alpha) \cdot (1 - 2 c\delta)^{l}} \ge \frac{80}{3\epsilon^2\alpha^2(1 - 2 c\delta)^{l}}.
\]
Hence, this inequality is satisfied by setting $C= \frac{28}{\alpha^2 \epsilon^2\cdot (1 - 2 c\delta)^{l}}$. Now, by \cref{theorem:alg-failure-inner-loop}~\eqref{item:RandomWalks-approx-success-prob} the
sampling algorithm fails with probability at most
\[
n^{-\frac{C\alpha \epsilon}{3}+2}e^2\le  n^{-\frac{28}{3\alpha\epsilon}+2}e^2\le n^{-4+2}e^2 .
\]
The probability of any of these failures happening in each round is at most $\frac{3}{n^2}+ \frac{1}{n^2}+\frac{e^2}{n^2} < \frac{12}{n^2}$. Hence,
over all round the failure probability is $O(\frac{1}{n^2\delta})$.
%

\paragraph{Space requirement.}
By \cref{theorem:alg-failure-inner-loop}~\eqref{item:RandomWalks-approx-space} and from $1 / n \le \alpha \le 1/4$, the space required is
\begin{eqnarray*}
O\rb{m+C l^{1+2 \alpha}n \ln{n}} &=& O\rb{m+\frac{1}{\alpha^2 \epsilon^2 \cdot (1 - 2 c\delta)^{l}}  l^{1+2\alpha} n \ln{n}} \\
  &=& O\rb{m+\frac{1}{\alpha^2 \epsilon^2}  \cdot (1 + 2 c\delta)^{l}  l^{1+2\alpha} n \ln{n}} \\
  &=& O\rb{m+\frac{1}{\alpha^2 \epsilon^2}  \cdot (1 + 2 c\delta)^{l}  l^{1.5} n \ln{n}} \\
  &=& O\rb{m+\frac{1}{\alpha^2 \epsilon^2}  \cdot (1 + 2 c\delta)^{9\frac{\ln n}{\epsilon}}  \rb{9\frac{\ln n}{\epsilon}}^{1.5} n \ln{n}} \\
 &=& O\rb{m+\frac{1}{\alpha^2 \epsilon^{3.5}}  \cdot (1 + 2 c\delta)^{9\frac{\ln n}{\epsilon}}  n \ln^{2.5}{n}} \\
 &=& O\rb{m+\frac{1}{\alpha^2 \epsilon^{3.5}}  \cdot n^{18\frac{c\delta}{\epsilon}}  n \ln^{2.5}{n}}.
\end{eqnarray*}
Moreover, we need $O\rb{\frac{m\log^2 n \log \log n}{\epsilon^2 \alpha^2}}$ space in \cref{line:directed-PageRank-initial-tpi} of \cref{alg:random-walks-directed-main}.
This completes the proof of \cref{theorem:PageRank-directed-c-balanced}.

\subsection{Transformation to a $c$-balanced Graph}
\label{sec:transforming-to-c-balanced}
In this section we will describe how to reduce a general graph $G=(V,E)$ without dangling vertices
to a $3$-balanced multigraph $G_c=(V_c,E_C)$.
The ways to handle dangling vertices are discussed in \cref{app:handling-dangling-nodes}.
The idea
is to replace each vertex by a path of length $\lambda = \lceil \log n \rceil$.

The graph $G_c=(V_c,E_c)$ is
defined as follows
\[
V_c = \{v_{i} : v\in V, i \in [1,\ldots, \lambda]\},
\]

\[
E_c = \{u_{\lambda}v_1 : uv \in V\} \cup \{(u_{i}u_{i+1})^j :   i \in [1,\ldots, \lambda-1], j\in [0,\ldots, \lceil \deg^-(u)/2^{i} \rceil ]\}.
\]

\begin{lemma}
\label{lemma:c-balanced-transformation}
If $G$ does not contain dangling vertices then $G_c$ is a $3$-balanced graph. Moreover, $G_c$
contains $n \lceil \log n \rceil$ vertices and at most $2m + n\lceil \log n\rceil$ edges.
\end{lemma}
\begin{proof}
In order to prove that $G_c$ is $c$-balanced, consider three cases for $v_i\in V_c$:
\begin{description}
  \item[$i=1$] then $v_i$ has $\deg^-(v)$ inedges, and $\lceil \deg^-(v)/2\rceil$ outedges, so
\[
3\deg^+(v_i)=3\lceil \deg^-(v)/2\rceil \ge \lceil \deg^-(v)/2\rceil+\deg^-(v) = \deg(v_i).
\]
  \item[$1<i <\lambda$] then $v_i$ has $\lceil \deg^-(v)/2^{i-1}\rceil$ inedges, and $\lceil \deg^-(v)/2^i\rceil$ outedges, so
\[
3\deg^+(v_i)=3 \lceil \deg^-(v)/2^i\rceil \ge \lceil \deg^-(v)/2^i\rceil + \lceil \deg^-(v)/2^{i-1}\rceil = \deg(v_i).
\]
  \item[$i=\lambda$] then $v_i$ has $\lceil \deg^-(v)/2^{\lambda-1}\rceil$ inedges and at least $1$ outedge, so
\[
3\deg^+(v_i) \ge \deg^+(v_i) + 2 \ge \deg^+(v_i) + \lceil 2 \deg^-(v)/n \rceil \ge \deg^+(v_i) + \lceil \deg^-(v)/2^{\lambda-1}\rceil  = \deg(v_i).
\]
\end{description}
By construction the number of vertices is $n \lceil \log n\rceil$, whereas the number of edges added to $G_c$ can be accounted
to inedges in $G$, i.e., for a vertex $v$ with indegree $\deg^-(v)$ we are adding
\[
\sum_{i=1}^{\lambda-1} \lceil \deg^-(v)/2^{i} \rceil\le  \sum_{i=1}^{\lambda-1}\deg^-(v)/2^{i}+1 \le \deg^-(v) + \lambda-1,
\]
edges. This gives $m+n\lambda = m+ n\lceil \log n\rceil$ additional edges.
\end{proof}

For each $uv \in E$ we call the edge $u_{\lambda}v_{1} \in E_C$ \emph{core}.
From the construction of $G_c$ we easily get the following.

\begin{observation}\label{obs:core}
Let $W = e_1, e_2, \ldots, e_k$ be a walk in $G_C$.
Then, there exists $1 \leq i \leq \lambda$ that has the following property.
Let $W_R$ be a subsequence of $W$ consisting of edges $e_{i + j\lambda}$ for $0 \leq j \leq (k - i) / \lambda$.
Then, $W_R$ is a walk in $G$, which contains all core edges of $W$.
\end{observation}

\subsection{Increasing Damping Factor}
\label{sec:lower-eps}
In \cref{sec:transforming-to-c-balanced} we described how to transform the input graph $G$ to a $c$-balanced graph $G_c$. In this process, each edge of $G$ is replaced by a path of length $\lceil \log{n}\rceil$. That means that a random walk of length $l$ in $G$ corresponds to a random walk of length $l \cdot \lceil \log{n}\rceil$ in $G_c$. In order make a correspondence between PageRank walks in $G$ to those in $G_c$, we need PageRank walks in $G_c$ to make jump transitions roughly $\log{n}$ times less frequently than in $G$. (We make this statement precise in \cref{sec:mapping-back}.) In light of this, we design method $\PageRankLargerDamping$ (\cref{alg:PageRank-lower-eps}) that given an approximate $\eps$-PageRank of $G$ outputs an approximate $\eps'$-PageRank of $G$ for $\eps' < \eps$. Moreover, for a given parameter $\tau$, it does so in $O\rb{\log_{1 + \tau} \tfrac{\eps}{\eps'}}$ iterations each of which is implemented by invoking \RandomWalks for length $O(\log{n} / \eps')$. The parameter $\tau$ also affects space complexity (for details see \cref{theorem:PageRank-lower-eps}), and in the final setup we let $\tau = o(1)$.

\PageRankLargerDamping uses \TranslateWalkEps (\cref{alg:translate-walk-eps}) as a subroutine. Given a random walk $w$ in $\GepsZero{\eps_j}$ and $\eps_{j + 1} \le \eps_j$, $\TranslateWalkEps$ either returns a walk $w_T$ which is a random walk in $\GepsZero{\eps_{j + 1}}$ or ``fails''. This method is very similar to \TranslateWalk for the case $\sigma_j < 1/2$.

\begin{algorithm}[H]
\begin{algorithmic}[1]
\Function{$\TranslateWalkEps$}{$G$, $\eps_j$, $\eps_{j + 1}$, $w$}
	\State{Let $g$ be the number of directed transitions in $w$.}
	\State{Let $t$ be the number of jump transitions in $w$.}
	\State{$p \gets r \frac{\eps_{j+1}^{t}(1-\eps_{j+1})^g}{\eps_j^{t} (1-\eps_j)^g}$}
	\Comment{$r$ is set in \cref{lemma:TranslateWalkEps-fail-prob} to guarantee that $p\le 1$.}
	\State{With probability $p$ return $w$; otherwise, return ``fail''.}
\EndFunction
\end{algorithmic}
\caption{
}
\label{alg:translate-walk-eps}
\end{algorithm}

\begin{algorithm}[H]
\begin{algorithmic}[1]
\Function{\PageRankLargerDamping}{$G$, $\tpi_{\eps_1}$, $\eps'$, $\tau$}
    \For{$j \gets 1 \ldots \lceil \log_{1 / (1 - \tau)} \eps_1 / \eps' \rceil$}
    \State{$l\gets \lceil \frac{9 \ln n}{\eps_j}\rceil $}
	\State Implicitly compute $\GepsZero{\eps_j}$ using \cref{eq:definition-T-eps-sigma}.
        \State Run $\RandomWalks\rb{\GepsZero{\eps_j}, l, t}$ where $t_i(v) = \lceil R \tpi_{\epsilon_j}(v) n \ln n \cdot k_i\rceil$, $k_i$ is defined by~\eqref{eq:ki-delta}. Let $W$ be the set of resulting walks. \label{line:lower-damping-invoke-RandomWalks}
        \Comment{$R$ is set in~\cref{theorem:PageRank-lower-eps}}
        \State $W_T\gets \emptyset$
				\State{$\eps_{j + 1} \gets \max\left\{\eps', \eps_j (1 - \tau)\right\}$}
        \ForAll{$w\in W$ in parallel}
            \If{$w_T=\TranslateWalkEps(G, \eps_j, \eps_{j + 1}, w)$ did not ``fail'' \label{line:translate-walk-eps}}
                \State{Add $w_T$ to $W_T$.}
            \EndIf				
        \EndFor
        \State{$\tpi_{\epsilon_{j + 1}} = \StationaryDistribution(W_T, \eps_j)$ \label{line:transform-walk-eps-compute-tpi}}
	\EndFor
\EndFunction
\end{algorithmic}
\caption{An algorithm that given a $(1 + \alpha)$-approximate $\eps_1$-PageRank $\tpi_{\eps_1}$ of $G$ for $\eps_1 \le 1/2$, outputs a $(1 + \alpha)$-approximate $\eps'$-PageRank $\tpi_{\eps'}$ of $G$ for $\eps' < \eps_1$. Parameter $\tau \in (0, 1/2)$ affects success probability, round complexity and space complexity.}
\label{alg:PageRank-lower-eps}
\end{algorithm}

\begin{lemma}\label{lemma:TranslateWalkEps-fail-prob}
	Given a random walk $w$ of length $l$ in $\GepsZero{\eps_j}$, with probability at least $\rb{\tfrac{1 - \eps_j}{1 - \eps_{j + 1}}}^{l}$ the call $\TranslateWalkEps(w)$ outputs a random walk of length $l$ in $\GepsZero{\eps_{j + 1}}$; otherwise, $\TranslateWalkEps(w)$ reports ``fail''.
\end{lemma}
\begin{proof}
	The proof of this lemma is similar to one of \cref{lemma:prob-bound-on-repeating-directed-pi} for $\sigma_j < 1/2$. In this proof as well, we would like to use \cref{lem:rejection-sampling} for each edge of walk $w$. Each edge of $w$ is either a jump transition (with probability $\eps_j$) or a (directed) graph transition (with probability $1 - \eps_j$).
	
	As in \cref{lemma:prob-bound-on-repeating-directed-pi}, fix a vertex $v$ and let $e$ be some transition that can be chosen by a random walk that has reached $v$. Denote by $f(\eps, e)$ the probability of choosing $e$. To use \cref{lem:rejection-sampling}, consider the ratios $f(\eps_j, e) / f(\eps_{j+1}, e)$. For any jump transition $e_j$ we have $f(\eps_j, e_j) / f(\eps_{j+1}, e_j) = \tfrac{\eps_j}{\eps_{j + 1}} \ge 1$. For a graph transition $e_g$, we have $f(\eps_j, e_g) / f(\eps_{j+1}, e_g) = \tfrac{1-\eps_j}{1-\eps_{j+1}} < 1$. Hence, for any transition $e$ we have
	\[
		\frac{f(\eps_j, e)}{f(\eps_{j + 1}, e)} \ge \frac{1-\eps_j}{1-\eps_{j+1}}.
	\]
	
	So, we set $r = (1 - \eps_j) / (1 - \eps_{j + 1})$ and the proof follows by application of \cref{lem:rejection-sampling} in the same way as in the proof of \cref{lemma:prob-bound-on-repeating-directed-pi} for $\sigma_j < 1/2$.
\end{proof}

We now want to use \cref{lemma:TranslateWalkEps-fail-prob} to lower-bound the success probability of $\TranslateWalkEps$ invoked within \cref{alg:PageRank-lower-eps}. For that, we first establish the following inequality.
\begin{lemma}\label{lemma:1-y/2-ge-exp(-2y)}
	For any $y \in [0, 1/2]$ we have $1 - y \ge \exp(-2 y)$.
\end{lemma}
\begin{proof}
	From Taylor expansion we have $1 - x + x^2 / 2 \ge \exp(-x)$. Since for $x \in [0, 1]$ it holds that $x / 2 \ge x^2 / 2$, we have
	$1 - x/2 \ge 1 - x + x^2 / 2 \ge \exp(-x)$. Now the lemma follows by letting $y = x/2$.
\end{proof}

\begin{lemma}\label{lemma:succeed-translate-walk-eps}
Let $\tau \le 1/2$. If $\eps_{j + 1} = (1 - \tau) \eps_j$, then $\TranslateWalkEps$ invoked by \cref{alg:PageRank-lower-eps} succeeds with probability at least $n^{-36 \tau}$.
\end{lemma}
\begin{proof}
	By \cref{lemma:TranslateWalkEps-fail-prob}, $\TranslateWalkEps$ ``fails'' with probability at most $\rb{(1 - \eps_j) / (1 - \eps_{j + 1})}^{l}$ for $l = \tfrac{c \ln{n}}{\eps_j}$. We now upper-bound this probability. We have
	\[
		\frac{1 - \eps_j}{1 - \eps_{j + 1}} = 1 - \frac{\tau \eps_j}{1 - (1 - \tau) \eps_j}.
	\]
	Therefore,
	\[
		\rb{\frac{1 - \eps_j}{1 - \eps_{j + 1}}}^l = \rb{1 - \frac{\tau \eps_j}{1 - (1 - \tau) \eps_j}}^{\frac{9 \ln{n}}{\eps_j}}
		\stackrel{\text{\cref{lemma:1-y/2-ge-exp(-2y)}}}{\ge} n^{-\frac{18 \tau}{1 - (1 - \tau) \eps_j}}
		\ge n^{-36 \tau},
	\]
	where we used the fact that $\eps_j, \tau \le 1/2$.
\end{proof}

\begin{theorem}\label{theorem:PageRank-lower-eps}
Let $G$ be a directed graph.
Let $\eps_1 > \eps' \geq \log n / o(S)$, and let $\tpi_{\eps_1}$ be a $(1 + \alpha)$-approximate $\eps_1$-PageRank of $G$. Given $\tau \in (0, 1/2]$, \PageRankLargerDamping (\cref{alg:PageRank-lower-eps}) outputs a $(1 + \alpha)$-approximate $\eps'$-PageRank $\tpi_{\eps'}$ of $G$. Moreover, \PageRankLargerDamping can be implemented in $O(\tau^{-1} \cdot \log{1/\eps'} \cdot (\log \log{n} + \log 1 / \eps'))$ MPC rounds and the total space of $O\rb{m+ \frac{1}{\eps'^{3.5} \alpha^{2}} n^{36\tau} n \ln^{2.5}{n}}$ with strongly sublinear space per machine. This algorithm is randomized and outputs a correct result with probability at least $1 - O(\frac{1}{\tau n^2}\cdot \log 1/ \eps')$.
\end{theorem}
\begin{proof}
	We split the proof into three parts: upper-bounding the success probability, deriving the space requirement, and analyzing the round complexity.

\paragraph{Round complexity.}
		The main loop of \cref{alg:PageRank-lower-eps} is executed at most $\left \lceil \log_{1 / (1 - \tau)} 1 / \eps' \right \rceil = O\rb{\tau^{-1} \cdot \log 1/ \eps'}$ times. Each loop invokes $\RandomWalks$ for length $O\rb{\tfrac{\ln{n}}{\eps'}}$, which can be executed in $O(\log \log{n} + \log 1 / \eps')$ MPC rounds. \cref{line:translate-walk-eps} to \cref{line:transform-walk-eps-compute-tpi} can be implemented in $O(1)$ rounds in the same way as described in the proof of \cref{theorem:PageRank-directed-c-balanced} (see \cref{section:proof-of-PageRank-directed-c-balanced}).
	
	\paragraph{Success probability.}
		By \cref{lemma:sampling-pagerank} each $\tpi_{\eps_j}$ computed on \cref{line:transform-walk-eps-compute-tpi} is a $(1 + \alpha)$-approximation of $\pi_{\eps'}$ with probability at least $1 - 3 / n^2$.

Note that algorithm generates at least
\begin{eqnarray*}
\lceil R \tilde{\pi}_{\epsilon_j}(v) n \ln n k_{\lceil \log l\rceil}\rceil &\ge& R \tilde{\pi}_{\epsilon_j}(v) n \ln n\\
&\ge& R\pi_{\epsilon_j}(v)(1-\alpha) n \ln n \ge R\epsilon_j (1-\alpha)  \ln n,
\end{eqnarray*}
random walks from each vertex. Next, these walks are downsampled using \cref{lemma:succeed-translate-walk-eps} with probability $n^{-36 \tau}$.
Similarly as in \cref{theorem:PageRank-directed-c-balanced}, we want the expected number of walks to be $2K_j$, so that the failure probability is less than $\frac{1}{n^2}$. Hence,
the following requirement on $R$
\[
2K_j \le n^{-36 \tau} \cdot R\epsilon_j (1-\alpha)  \ln n
\]
what gives
\[
R \ge \frac{20}{n^{-36\tau} \epsilon_j^2 (1-\alpha) \alpha^2} \le \frac{28}{\epsilon_j^2 \alpha^2} n^{36\tau}
\]
We set $R=\frac{28}{\epsilon_j^2 \alpha^2} n^{36\tau}$ and, similarly, as in \cref{theorem:PageRank-directed-c-balanced} the
sampling algorithm fails with probability at most $n^{-2}e^2$. This gives the failure probability $\frac{12}{n^2}$ of each round, and over all rounds we get $\frac{12}{\tau n^2}\cdot \log 1/ \eps'$.

	
	\paragraph{Space requirement.}
The space usage is dominated by calls to~\cref{alg:random-walks}, which from \cref{theorem:alg-failure-inner-loop}
are bounded by $O\rb{m+R l^{1+ 2 \alpha}n \ln{n}}$. Observe that $l$ is the highest in the last round and equals $l= \lceil \frac{9 \ln n}{\eps'}\rceil= O(\ln n/\eps')$.
This gives the following upper-bound on space
\begin{eqnarray*}
O\rb{m+R l^{1+2 \alpha}n \ln{n}} &=& O\rb{m+ \frac{1}{\eps'^2 \alpha^2} n^{36\tau} l^{1+2\alpha} n \ln{n}} \\
  &=&O\rb{m+ \frac{1}{\eps'^2 \alpha^2} n^{36\tau} l^{1.5} n \ln{n}} \\
  &=&O\rb{m+ \frac{1}{\eps'^2 \alpha^2} n^{36\tau} \rb{\frac{9\ln n}{\eps'}}^{1.5} n \ln{n}}\\
  &=&O\rb{m+ \frac{1}{\eps'^2 \alpha^2} n^{36\tau} \rb{\frac{9\ln n}{\eps'}}^{1.5} n \ln{n}}\\
  &=&O\rb{m+ \frac{1}{\eps'^{3.5} \alpha^{2}} n^{36\tau} n \ln^{2.5}{n}}.\\
\end{eqnarray*}
\end{proof}

\subsection{From PageRank in $c$-balanced to PageRank in General Graphs}
\label{sec:mapping-back}
In the previous sections we developed a way for efficiently approximating PageRank of $c$-balanced graphs. Here, we describe how to use that result to approximate the PageRank of a graph $G$, which is not necessarily $c$-balanced. Let $G_c$ be a $c$-balanced graph obtained from $G$ by applying the transformation described in \cref{sec:transforming-to-c-balanced}. Recall that to approximate PageRank it suffices to sample random walks up to the point of their \emph{first} random jump (see \cref{alg:random-walks-directed-main}). For the input graph $G$ and damping factor $1-\eps$ such a random walk with high probability has length $O\rb{\tfrac{\log{n}}{\eps}}$. Instead of generating PageRank walks in $G$, we will generate them in $G_c$, but for a damping factor $1-\eps / \poly \log{n}$ (see \cref{lemma:random-jumps-early}). This will be done using \cref{alg:PageRank-lower-eps}. Let $W_c$ be an $(\tfrac{\eps}{\poly \log{n}})$-PageRank walk.
We observe that with constant probability the first jump in $W_c$ appears after $\Omega\rb{\tfrac{\log^2{n}}{\eps}}$ steps.  Assuming this from $W_c$ we can obtain
a random walk $W$ in $G$ of length $\Omega\rb{\tfrac{\log{n}}{\eps}}$ such that no transition of $W$ is a random jump.
In order to obtain $\eps$-PageRank we reintroduce random jumps with probability $\eps$ and truncate walks after first such jump.
This is done by sequentially iterating over the edges of $W$ and with probability $\eps$ truncating $W$ at any given step. This process is given as algorithm $\TranslateBalancedToPageRankWalk$ (\cref{alg:translate-balanced-to-PageRank-walk}).

\begin{algorithm}[H]
\begin{algorithmic}[1]
\Function{$\TranslateBalancedToPageRankWalk$}{$G$, $\eps$, $w$}
	\If{$w$ contains random jump \label{line:1st-fail-TranslateBalancedToPageRankWalk}}
		``fail''.
	\EndIf
	\State{Let $w_R$ be the walk in $G$ consisting of all core edges of $w$. \label{line:undo-c-balanced}} \Comment{See \cref{obs:core}}
	\State{Mark each edge of $w_R$ independently and with probability $\eps$. \label{line:mark-edges}}
	\If{at least one edge of $w_R$ is marked}
		\State{Truncate $w_R$ before the first marked edge.\label{line:truncate-TranslateBalancedToPageRankWalk}}
        \State{Return $w_R$.}
	\Else
		\State{``fail'' \label{line:2nd-fail-TranslateBalancedToPageRankWalk}}
	\EndIf
\EndFunction
\end{algorithmic}
\caption{Let $G$ be a graph and $G_c$ its $c$-balanced version. Given a PageRank walk $w$ in $G_c$ that has no random jumps, the algorithm returns an $\eps$-PageRank walk in $G$ that has exactly one random jump transition or the algorithm ``fails''. This random jump transition is the last one in the walk.
}
\label{alg:translate-balanced-to-PageRank-walk}
\end{algorithm}

Now we present the main algorithm of this entire section. The algorithm outputs an approximate PageRank for the input graph.
\begin{algorithm}[H]
\begin{algorithmic}[1]
\Function{\PageRankOfGeneralGraphs}{$G$, $\eps$, $\alpha$}
	\State{Let $G_c$ be the balanced graph of $G$ obtained as described in \cref{sec:transforming-to-c-balanced}.}
	\State{Let $\tpi^c_{1/2} \gets \PageRankOfBalancedGraphs(G_c, 1/2, \alpha)$.}
	\State{$\ell \gets   \lceil \log n \rceil \cdot \left\lceil \tfrac{9 \log{n}}{\eps} \right \rceil$, and $\eps' \gets 1 / (4 \ell)$. \label{line:set-ell-and-eps'}}
	\State{Let $\tpi^c_{\eps'} \gets \PageRankLargerDamping(G_c, \tpi^c_{1/2}, \eps')$. \label{line:execute-PageRankLargerDamping}}
	\State Implicitly compute $(G_c)_{\eps', 0}$ using \cref{eq:definition-T-eps-sigma}.
	\State Run \textsc{RandomWalks}$((G_c)_{\eps', 0}, \ell, t)$ where $t_i(v) = \lceil L \tpi^c_{\eps'}(v) n \ln n \cdot k_i\rceil$, $k_i$ is defined by~\eqref{eq:ki-delta} and $d'$ is a sufficiently large constant. Let $W$ be the set of resulting walks. \label{line:RandomWalks-in-PageRankGeneral}
    \Comment{$L$ is set in \cref{theorem:PageRank-directed}}

	\State $W_T\gets \emptyset$
	\ForAll{$w\in W$  in parallel}
			\If{$w_T=\TranslateBalancedToPageRankWalk(G, \eps, w)$ did not ``fail'' \label{line:invoke-TranslateBalancedToPageRankWalk}}
				\State{Add $w_T$ to $W_T$.}
			\EndIf
	\EndFor
	\State{$\tpi_\eps = \StationaryDistribution(W_T, \eps)$ \label{line:PageRank-general-compute-tpi}}
\EndFunction
\end{algorithmic}
\caption{An algorithm for computing a $(1 + \alpha)$-approximate PageRank $\tilde{\pi}_{\epsilon}$ of a graph $G$.}
\label{alg:PageRank-in-input-graph}
\end{algorithm}

Note that \cref{alg:PageRank-in-input-graph} does not directly map PageRank walks from $G_c$ to PageRank walks in $G$.
Instead, it takes advantage of the fact that in order to approximate PageRank one only needs to know the random walks until the first jump transition and proceeds as follows.
It first computes PageRank walks in $G_c$, then discards all walks that have at least one jump transition and finally truncates the resulting walks by simple coin tossing.
Note that this truncation step is equivalent to truncating the walks just before the first jump transition.

Now we analyze the correctness of \cref{alg:PageRank-in-input-graph}. The following result will be useful in establishing failure probability of \TranslateBalancedToPageRankWalk.
\begin{lemma}\label{lemma:random-jumps-early}
	Let $\ell \ge 1$ be a parameter. Define $\eps' = 1 / (4 \ell)$. An $\eps'$-PageRank walk
    does not make a random jump within the first $\ell$ steps with probability at least $0.6$.
\end{lemma}
\begin{proof}
	An $\eps'$-PageRank walk does not make a random jump within the first $\ell$ steps with probability
	\[
		(1 - \eps')^l \stackrel{\text{\cref{lemma:1-y/2-ge-exp(-2y)}}}{\ge} \exp{\rb{-2 \eps' l}} = \exp{(-1/2)} \ge 0.6.
	\]
\end{proof}

\cref{lemma:random-jumps-early} essentially states the following. If we are given an $\eps'$-PageRank walks in $G_c$ then
with probability at $1/2$ it can be turned into a $\eps$-PageRank walk in $G$.

\begin{lemma}
\label{lemma:downsampling-eps}
	The invocation $\TranslateBalancedToPageRankWalk(G, \eps, w)$ on \cref{line:invoke-TranslateBalancedToPageRankWalk} of \cref{alg:PageRank-in-input-graph} ``fails''
with probability at most $1/2$. If the algorithm succeeds, then it returns an $\eps$-PageRank walk of $G$ that has exactly one jump transition and that jump transition is the last one in the walk.
\end{lemma}
\begin{proof}
	We first analyze the success probability of $\TranslateBalancedToPageRankWalk$ and then show the claim for its output.
	\paragraph{Success probability.}
	There are two lines where \TranslateBalancedToPageRankWalk can ``fail'' -- \cref{line:1st-fail-TranslateBalancedToPageRankWalk,line:2nd-fail-TranslateBalancedToPageRankWalk}. By \cref{lemma:random-jumps-early}, it ``fails'' in \cref{line:1st-fail-TranslateBalancedToPageRankWalk} with probability at most $1/2$.
	\\
	\cref{line:2nd-fail-TranslateBalancedToPageRankWalk} is executed only if no edge of $w_R$ is marked on \cref{line:mark-edges}. We have
\[
|w_R| = \left \lfloor \ell / \left \lceil \log{n} \right \rceil \right\rfloor = \left\lfloor \left \lceil \log n \right \rceil \cdot \left \lceil \tfrac{9 \log{n}}{\eps} \right \rceil / \left \lceil \log{n} \right \rceil \right \rfloor= \left \lceil \tfrac{9 \log{n}}{\eps} \right \rceil.
\]
Hence, by~\cref{lemma:no-jump} no edge of $w_R$ is marked with probability at most $\frac{1}{n^9} \le 1/512$.
This implies that the invocation of \TranslateBalancedToPageRankWalk succeeds with probability at least $0.6 - 1/512 \ge 1/2$, as desired.
	
    \paragraph{Output of $\TranslateBalancedToPageRankWalk$.}
	As input, $\TranslateBalancedToPageRankWalk$ gets a PageRank walk $w$ in $G_c$; this walk is generated on \cref{line:RandomWalks-in-PageRankGeneral} of \cref{alg:PageRank-in-input-graph}. On \cref{line:1st-fail-TranslateBalancedToPageRankWalk} of $\TranslateBalancedToPageRankWalk$, $w$ is discarded if it contains any jump transition. So, if the algorithm does not ``fail'', $w$ is a random walk in $G_c$.
By the transformation on \cref{line:undo-c-balanced} we obtain a walk $w_R$ which is a random walk in $G$. After that, each transition of $w_R$ is marked with probability $\eps$ (that is implemented by marking on \cref{line:mark-edges}), i.e.,
we toss a random coin to see potential steps for random jumps. Then we truncate the walk at the moment of the first jump in $w_R$ occurs, and return this truncated walk. By construction, this walk is an $\eps$-PageRank walk in $G$.
\end{proof}

Now we are ready to prove the correctness of \cref{alg:PageRank-in-input-graph}, which establishes one of the main results of the paper.

\theoremPageRankdirected*

\begin{proof}
We will show that \cref{alg:PageRank-in-input-graph} satisfies properties of this claim.
As in previous proofs we split the proof into three parts: the space requirements; the round complexity; and, the probability of success. This is
the moment when we set all the parameters of the algorithms. We set $\delta$ to be the largest value not greater than $(\log \log n)^{-1}$ such that $\delta^{-1} \in \bbN$. Observe that $\delta \ge 1 / (1 + \log \log{n})$. Also, we set $\tau = \log \log \log^{-1} n$.
We will use the following bounds
\begin{equation*}
  O(1/\eps') = O(\ell) =  O(\log^2n /\eps).
\end{equation*}
\begin{equation*}
O(\log 1/\eps') = O(\log \ell) = O(\log \log n + \log 1/\eps),
\end{equation*}
and when cruder bound is enough we bound
\begin{equation*}
O(\log 1/\eps') = O (\log \log n + \log n) = O(\log n).
\end{equation*}

\paragraph{Round complexity.}
By \cref{theorem:PageRank-directed-c-balanced} the execution of
$\PageRankOfBalancedGraphs(G_c, 1/2, \alpha, \delta)$ takes $O(\delta^{-1} \log \log n) = O(\log^2 \log n)$ rounds.

Recall that $\eps' = \Theta\rb{\tfrac{\eps}{\log^2 n}}$ in \cref{alg:PageRank-in-input-graph} (see \cref{line:set-ell-and-eps'}).
By \cref{theorem:PageRank-lower-eps} the execution of $\PageRankLargerDamping$ on \cref{line:execute-PageRankLargerDamping} of \cref{alg:PageRank-in-input-graph} takes
\begin{eqnarray*}
O(\tau^{-1} \cdot \log{1/\eps'} \cdot (\log \log{n} + \log 1 / \eps')) &=& O(\log \log \log n \cdot (\log \log n+\log 1/\eps)\cdot(\log \log n+\log 1/\eps))\\
 &=& \tO(\log^2 \log n+\log^2 1/\eps)
\end{eqnarray*}
rounds.
Note that since $\epsilon \geq \log^3 n / o(S)$, we have $\epsilon' \geq \log n / o(S)$, which is required by \cref{theorem:PageRank-lower-eps}.

\cref{line:invoke-TranslateBalancedToPageRankWalk} to \cref{line:PageRank-general-compute-tpi} can be implemented in $O(1)$ rounds in the same way as described in the proof of \cref{theorem:PageRank-directed-c-balanced} (see \cref{section:proof-of-PageRank-directed-c-balanced}).

\paragraph{Success probability.}
As in proof of \cref{theorem:PageRank-directed-c-balanced} we require that there are $2K$ walks starting in each vertex before
the call to $\StationaryDistribution$, so that we have $K$ walks in each vertex with probability at least $1-\frac{1}{n^2}$.
Using the downsampling from~\cref{lemma:downsampling-eps} this gives
\[
2K \le \frac{1}{2} \lceil L \tilde{\pi}(v) n\ln{n}\rceil \le L \eps' (1-\alpha)  \ln n,
\]
what bounds $L$ as
\[
L \ge \frac{40}{\epsilon^2\alpha^2(1-\alpha)} \ge \frac{160}{3\epsilon^2\alpha^2}.
\]
We set $L= \frac{54}{\alpha^2 \epsilon^2}$. Now, by \cref{theorem:alg-failure-inner-loop}~\eqref{item:RandomWalks-approx-success-prob} the
sampling algorithm fails with probability at most
\[
n^{-\frac{C\alpha \epsilon}{3}+2}e^2\le  n^{-\frac{54}{3\alpha\epsilon}+2}e^2\le n^{-2}e^2 = O\rb{\frac{1}{n}}.
\]

By \cref{theorem:PageRank-directed-c-balanced} the execution of $\PageRankOfBalancedGraphs$ fails with probability $O\rb{\frac{1}{n^2\delta}} = O\rb{\frac{1}{n}}$.

By~\cref{theorem:PageRank-lower-eps} the execution of $\PageRankLargerDamping$ fails with probability
\begin{eqnarray}
  O\rb{\frac{1}{\tau n^2}\cdot \log 1/ \eps'} &=& O\rb{\frac{\log \log \log n}{ n^2} (\log \log n + \log 1/\eps)} \\
  &=& O\rb{\frac{\log \log \log n}{ n^2} (\log \log n + \log n)} \\
  &=& O\rb{\frac{1}{n}}.
\end{eqnarray}
Hence, in total the failure probability of the algorithm is $O(1/n)$.

\paragraph{Space complexity.}
By \cref{theorem:PageRank-directed-c-balanced} the execution of $\PageRankOfBalancedGraphs(G_c, 1/2, \alpha, \delta)$ requires space
\[
O\rb{\frac{m\ln^2 n \ln \ln n}{(1/2)^2 \alpha^2}+n^{18 \frac{3\delta}{1/2}} \frac{n \ln^{2.5} n}{(1/2)^{3.5} \cdot \alpha^2}}
= \tO\rb{ \frac{m}{\alpha^2} + \frac{n^{1+o(1)}}{\alpha^2}}.
\]
By \cref{theorem:PageRank-lower-eps} the execution of $\PageRankLargerDamping$ needs space
\[
O\rb{m+ \frac{1}{\eps'^{3.5} \alpha^{2}} n^{36\tau} n \ln^{2.5}{n}}= O\rb{m+ \frac{\log^7 n}{\eps^{3.5} \alpha^{2}} n^{36\tau} n \ln^{2.5}{n}} = \tO\rb{m + \frac{n^{1+o(1)}}{\eps^{3.5}\alpha^2}}.
\]
By \cref{theorem:alg-failure-inner-loop} the sampling algorithm $\RandomWalks$ requires space bounded by
\[
O\rb{m+L \ell^{1+ 2 \alpha}n \ln{n}} = \tO\rb{m + \frac{1}{\eps^2 \alpha^2} \cdot \ell^{1.5} n } = \tO\rb{m + \frac{n}{\eps^{3.5} \alpha^2}}.
\]
Hence, the final space complexity of the algorithm is $\tO\rb{\frac{m}{\alpha^2}+ \frac{n^{1+o(1)}}{\eps^{3.5}\alpha^2}}$.
\end{proof}

We can now use \cref{theorem:PageRank-directed} to sample directed random walks and obtain the following result.
\randomwalksdirected*
\begin{proof}
	Let $\alpha = 1 / \log{n}$ and $\eps = 1 / (4 l)$. Invoke \cref{theorem:PageRank-directed} to obtain a $(1 + \alpha)$-approximate $\eps$-PageRank. This invocation can be implemented in $\tO\rb{\log^2 \log{n} + \log^2 l}$ rounds and the total space of $\tO\rb{m + n^{1 + o(1)} l^{3.5}}$.
	
	Let $t_i$ be as defined in \cref{theorem:alg-failure-inner-loop}. Invoke $\RandomWalks(G, l, t)$ with $C = 20 D / (\alpha \eps)$. By \cref{theorem:alg-failure-inner-loop}, and given that $\tpi(v) \ge (1 - \alpha) \pi(v) \ge (1 - \alpha) \eps / n$, with probability at least $1 - \Theta(1 / n)$ this invocation outputs at least $C \tpi(v) n \ln{n} \ge 20 D \ln{n}$ random walks from each vertex $v$. Let $W$ be the collection of those walks. Also by \cref{theorem:alg-failure-inner-loop}, $W$ can be obtained in $O(\log {l})$ rounds by using the total space of $\tO(m + C l^{1 + 2 \alpha} n) \in \tO\rb{m + D n l^{2 + o(1)}}$.
	
	The walks in $W$ are PageRank walks. Nevertheless, by \cref{lemma:random-jumps-early} and Chernoff bound, with probability at least $1 - \Theta(1 / n)$ for each $v$ there exist $D$ PageRank walks in $W$ that contain no random jump. Those walks are the walks that satisfy the claim of this theorem.
\end{proof}

\section{PRAM Implementation}
In this section we discuss PRAM implementation of our algorithms for computing (directed) random walks and hence prove the following.
\rwpram*

\subsection{Storing and Sorting Walks}
In our MPC implementation of the algorithms described in other sections, we assume that an entire random walk can be stored on one machine.
In fact, this is the only reason why our algorithm assume that the length of each walk $l$ satisfies $l = o(S)$.
In this section we show a different approach, which would also allow us to increase the upper bound on $l$.

To store each walk $w$, we allocate $l$ consecutive memory cells, even if $|w| < l$. 
The edges of $w$ occupy the first $|w$| cells and the remaining cells are filled with $\bot$.

\paragraph{Walk identifiers.} To each created walk $w$ we assign an integer identifier $w_{id}$ chosen uniformly at random from the interval $[1, n^{10}]$. With high probability walks have distinct identifiers. 

\paragraph{Sorting walks.} We now explain how to sort the walks with respect to their first vertex. Consider a walk $w$ with the starting vertex $v$. Then, to $i$-th cell allocated for $w$ we assign triple $(v, w_{id}, i)$. This labeling can be done in $O(\log{n})$ time by building a binary tree of processors over the cells allocated for $w$.
A nice feature of this labeling, and also assigning identifiers to random walks, is that after sorting the triples lexicographically, all the cells containing information of $w$ appear consecutively. Sorting walks with respect to their last vertex is done in an analogous way.

\paragraph{Stitching two walks.}
When the algorithm has to stitch two walks $w_1$ and $w_2$, it copies the edges of $w_2$ over the first $|w_2|$ cells allocated for $w_1$ that have value $\bot$. The first cell containing value $\bot$ can be found in $O(\log{n})$ time, e.g., by performing a binary search. After such cell $x$ is found, we copy $w_2$ (that occupies consecutive memory cells) over memory cells $x$ through $x + |w_2| - 1$. This again can be done in $O(\log{n})$ time by building a binary tree over $w_2$.

\subsection{Implementation of \RandomWalks}
\label{section:PRAM-RandomWalks}
From our discussion about MPC implementation of \RandomWalks, it suffices to show that sorting, computing prefix sums, \NumberingSublists and \Predecessor can be implemented in $O(\log{n})$ time. Sorting and computation of prefix sums can be done in $O(\log{n})$ time in \cite{cole1988parallel} and \cite{ladner1980parallel}, respectively.

We now describe how to implement \Predecessor. Recall that in this primitive we are given an ordered list of tuples $L$, each element labeled by $0$ or $1$. This primitive can be implemented as follows.
\begin{itemize}
	\item Assign $i$ to the $i$-th element of $L$. This can be done by computing the prefix sum on $L$ with each element having value $1$.
	\item Consider the $i$-th element $e_i$ of $L$. If the label of $e_i$ equals $0$, assign value/pair $(0, e_i)$ to $e_i$. Otherwise, assign value $(i, e_i)$ to $e_i$.
	\item The prefix sum approach described in \cite{ladner1980parallel} can be performed for any associative operation, including $\max$. Compute the prefix sum with operation $\max$ over the values/pairs assigned to the elements of $L$. These prefix sums correspond to the output of \Predecessor.
\end{itemize}

As described in \cref{section:MPC-Implementation-RandomWalks}, \NumberingSublists can be implemented by using prefix sum computation, sorting and \Predecessor, each of which can be executed in $O(\log{n})$ time.

\subsection{Implementation of \StationaryDistribution}
In \cref{lemma:StationaryDistribution-rounds} we described how to implement \StationaryDistribution in MPC. That proof relied on primitives \NumberingSublists and \Predecessor, that we explained how to implement in $O(\log{n})$ time. The proof also computes maximum over certain subarrays, which can be done in $O(\log{n})$ time.

To conclude the discussion, it remains to comment about implementation of \cref{line:StationaryDistribution-truncate}. This line truncates walks after their first random jump. This again can be done by building a binary tree over each walk $w$ and finding the first edge corresponding to a random jump of $w$. All the cells allocated to $w$ after that jump are set to $\bot$.

\subsection{Construction of Stochastic Graphs}
To construct a stochastic graph with parameters $\eps$ and $\sigma$ in MPC, we had to broadcast $\eps$ and $\sigma$ to all the machines containing the input graph, copy edge and vertices and properly annotate them. See the description above \cref{observation:constructing-stochastic-graph} for more details. In PRAM we perform similar steps. Namely, we build a binary tree over the memory cells that our algorithm uses, copy the corresponding edges and vertices and annotate them by $\eps$ and $\sigma$.

\subsection{Implementation of \TranslateWalkBase}
We use several algorithms that translate random walks between different graphs (\cref{alg:translate-walk,alg:translate-walk-eps,alg:translate-balanced-to-PageRank-walk}). Their MPC implementation is straightforward as we assume that each walk is stored on one machine. Nevertheless, their PRAM implementation is also almost direct. Each of those methods counts the number of distinct transition a walk has, or marks edges independently of each other, or finds the first edge of a walk that has a specific property (e.g., \cref{line:truncate-TranslateBalancedToPageRankWalk} of \cref{alg:translate-balanced-to-PageRank-walk}). Each of those operations can be performed in $O(\log{n})$ rounds by building a binary tree over the corresponding walk.


\section{Testing Bipartitness}
\label{sec:Bipartiteness}
We now show how to use our random walk algorithm for \emph{testing bipartiteness}. In this promise problem, we are given an undirected graph $G$ on $n$ vertices with $m$ edges and a parameter $\eps \in (0,1)$. We want to distinguish the case that $G$ is bipartite from the case that at least $\eps m$ edges have to be removed to achieve this property. Our parallel algorithm combines techniques developed for previous bipartiteness algorithms \cite{GoldreichR99,KaufmanKR04} with our simulation of random walks. For simplicity, we assume that vertices of $G$ are not isolated. Our algorithm can be seen as the following procedure consisting of three steps, in which the first two are preprocessing steps:
\begin{enumerate}
 \item If the graph is dense, we reduce its number of edges to $O(n/\eps)$ by independently keeping each edge with an appropriate probability. The resulting graph is very likely to be still $\eps/2$-far from bipartiteness.
    \item If the graph has high degree vertices, we apply the idea of Kaufman, Krivelevich, and Ron~\cite{KaufmanKR04} to replace all high degree vertices with low degree bipartite expanders. This again preserves the distance from bipartiteness up to a constant factor and allows to assume that the resulting graph has only vertices of small degree.
 \item In the resulting graph, we run a small number of short random walks from every vertex. We show that if the graph is far from bipartiteness then random walks from one of the vertices are very likely to discover an odd-length cycle. 
\end{enumerate}

We present a more formal description of the algorithm as \cref{alg:bipartiteness}.

\begin{algorithm}[H]
\begin{algorithmic}[1]
        \State{Independently keep each edge with probability $\min\left\{1,O\left(\frac{n}{\eps m}\right)\right\}$. \label{line:bipartiteness-sparsify}}
        \State{Replace high degree vertices with bipartite expanders (more details in \cref{section:biparite_replace_expanders}).}\label{alg:bipartiteness:reduce_degree}
	\State Using \cref{alg:random-walks}, generate $\poly(\eps^{-1}\log n)$ random walks of length $\poly(\eps^{-1}\log n)$ from each vertex.	
	\For{$v \in V$}
	        \State{$V_0 \leftarrow {}$vertices reached by the random walks from $v$ in an even number of steps}
	        \State{$V_1 \leftarrow {}$vertices reached by the random walks from $v$ in an odd number of steps}
	        \If{$V_0 \cap V_1 \ne \emptyset$}
		    \State \Return \reject	
		\EndIf
	\EndFor
	\State \Return \accept
\end{algorithmic}
\caption{\BipartitenessTester$(G, \eps)$:
An algorithm for testing bipartiteness of an undirected graph $G=(V,E)$ on $n$ vertices for a closeness parameter $\eps \in (0,1)$.
\label{alg:bipartiteness}}
\end{algorithm}

\subsection{Sampling a Sparse Graph}

We now prove that the first step of our algorithm (with an appropriate constant selection) preserves the distance to bipartiteness. 

\begin{lemma}\label{lem:bipartite_sampling}
Let $G$ be an undirected graph on $n$ vertices with $m$ edges. 
Let a graph $G'$ be on the same set of vertices as $G$ created by selecting each edge independently with probability $\min\left\{1,\frac{Cn}{\eps m}\right\}$, where $C$ is a sufficiently large positive constant.
If $G$ is $\eps$-far from bipartiteness, then with probability $1-2^{-\Omega(n)}$, $G'$ is $\eps/2$-far from bipartite and has at most $O(n/\eps)$ edges.
\end{lemma}

\begin{proof}
Suppose that $G=(V,E)$ is $\eps$-far.
The lemma holds trivially if $\frac{Cn}{\eps m} \ge 1$, because $G'=G$ is $\eps/2$-far and has at most $Cn/\eps$ edges. We can therefore focus on the case that 
$\frac{Cn}{\eps m} < 1$. The expected number of edges in this case is $Cn/\eps$. By the Chernoff bound, the number of edges in $G'$ is, with probability $1-2^{-\Omega(n)}$, at most $\frac{11}{10}Cn/\eps$. Now consider any partition of the set $V$ of vertices of $G$ into two sets $V_1$ and $V_2$. Since $G$ is $\eps$-far from bipartiteness, the sum of the number of edges in $G[V_1]$ and $G[V_2]$ is at least $\eps m$. Otherwise, we could delete them to turn the graph into bipartite. The expected sum of the number of edges in $G'[V_1]$ and $G'[V_2]$ has then to be at least $\eps m \cdot \frac{Cn}{\eps m} = Cn$. Again, by the Chernoff bound, the number of them is at least $\frac{9}{10}Cn$ with probability $1-2^{-\Omega(n)}$, where the constant hidden by the $\Omega$-notation can be made arbitrarily large by making $C$ sufficiently large. By the union bound, the probability that the total number of edges in $G'$ is more than $\frac{11}{10}Cn/\eps$ and that in one of the partitions, fewer than $\frac{9}{10}C n$ edges can be removed to make the graph bipartite is at most $2^{-\Omega(n)} + 2^{n-1} \cdot 2^{-\Omega(n)} = 2^{-\Omega(n)}$. This holds because all constants hidden by the $\Omega$ notation can be made arbitrarily small by setting $C$ to be sufficiently large. The distance of $G'$ from bipartiteness is then at least $\left(\frac{9}{10}Cn\right) / \left(\frac{11}{10}Cn/\eps\right) \ge \eps/2$.
\end{proof}

\subsection{Replacing High-Degree Vertices with Expanders}
\label{section:biparite_replace_expanders}

We now give more details of Step~\ref{alg:bipartiteness:reduce_degree} of \cref{alg:bipartiteness}, which reuses the degree reduction method of Kaufman et al.~\cite{KaufmanKR04}. More specifically, Section~4.1 of their paper shows how to take a graph $G=(V,E)$ and turn into a graph $G'=(V',E')$ in which the maximum degree equals roughly the average degree of $G$. Additionally, $G'$ preserves $G$'s distance to bipartiteness.

We start by describing the transformation of the vertex set.
Let $d = \davg^{+}(G)$. We copy every vertex $v$ of $G$ such that $\deg^{+}(v) \le d$ into $G'$. Vertices $v$ of higher degree are replaced by bipartite expanders as follows. For each such $v$, we introduce two sets of vertices of cardinality $\lceil\deg^{+}(v)/d\rceil$. We refer to one of them as \emph{internal}, and the other one as \emph{external}. If $\deg^{+}(v)<d^2$, we create a full bipartite graph between the internal and external vertices with edge multiplicites that make vertices have degree almost $d$. Otherwise, when $\deg^{+}(v) \ge d^2$, we use any explicit bipartite expander construction of degree $d$ between the two sets \cite{expander1,expander2}.

Now we describe the transformation of the edge set. For every edge ${u_1,u_2} \in E$, we adjust its endpoints as follows. Let $u$ be one of them. If $u$ was directly copied from $G$ to $G'$, then we do nothing. Otherwise, we replace it with one of the vertices in the external set created for $u$. For every vertex $u$ in the original graph that is replaced by a set of vertices, we assign the original edges involving $u$ to the external vertices created for $u$ such that no external vertex is assigned more than $d$ such edges.

Kaufman et al.~\cite{KaufmanKR04} prove the following.

\begin{lemma}[{\cite[Theorem 5]{KaufmanKR04}}]\label{lemma:replace-by-expanders}
The graph $G'$ created as above has the following properties:
\begin{enumerate}[(A)]
 \item\label{item:KaufmanKR-V'-and-E'} $|V'| = \Theta(|E|)$ and $|E'| = \Theta(|E|)$.
 \item\label{item:KaufmanKR-bipartite} If $G$ is bipartite, so is $G'$.
 \item\label{item:KaufmanKR-eps-far} If $G$ is $\eps$-far from bipartite for some $\eps \in (0,1)$, then $G'$ is $\Omega(\eps)$-far from bipartite.\footnote{Kaufman et al.\ in fact switch between two models and slightly different definitions of distance to bipartiteness in this statement, but up to a constant factor, they are equivalent in our setting.}
 \item\label{item:KaufmanKR-max-deg} $\dmax^{+}(G') = O(\davg^{+}(G))$.
\end{enumerate}
\end{lemma}

\subsection{Detecting Odd-Length Cycles in Low-Degree Graphs}

\begin{lemma}\label{lemma:detecting-odd-cycles}
Let $G$ be an undirected graph on $n$ vertices with $m$ edges such that $\dmax^{+}(G) = O(m/(\eps n))$. There is an $l = \poly(\eps^{-1}\log n)$ such that two random walks of length $l$ from a random vertex detect an odd-length cycle with probability at least $\Omega(n^{-1}\poly(\eps/\log n))$ if $G$ is $\eps$-far from bipartite, where  $\eps \in (0,1)$.
\end{lemma}
\begin{proof}
For $G$ such that $\dmax^{+}(G) = O(m/(\eps n))$, Kaufman et al.~\cite[Theorem~1]{KaufmanKR04} (invoked with $\eps^2$ in place of $\eps$) show that the following algorithm can be used for one-sided testing of bipartiteness:
\begin{enumerate}
 \item Sample $k = \Theta(1/\eps^2)$ vertices $v$ and run $t = \sqrt{n}\cdot\poly(\eps^{-1}\log n)$ independent random walks of length $l = \poly(\eps^{-1} \log n)$ from each of them.
 \item If for some $v$, two of the random walks reach the same vertex, one using an even number of steps and the other using an odd number of steps then reject. Otherwise accept.
\end{enumerate}

Let $p$ be the probability of two random walks of length $l$ from a random vertex detecting an odd-length cycle. We use $p$ to upper bound the probability of success of the tester of Kaufman et al.~\cite{KaufmanKR04} for a graph $\eps$-far from bipartite. By the union bound, it cannot succeed with probability greater than $k \cdot \binom{t}{2} \cdot p \le k\cdot t^2 \cdot p$. Since it has to succeed with probability at least $2/3$, we have 
$k\cdot t^2 \cdot p \ge \frac{2}{3}$, and therefore, $p = \Omega(1/(k\cdot t^2)) = \Omega(n^{-1} \cdot \poly(\eps/\log n))$.
\end{proof}

\subsection{Full Bipartiteness Tester}

\begin{theorem}
Let $\alpha \in (0,1)$ be a fixed constant.
There is an MPC algorithm for testing bipartiteness with a proximity parameter $\eps \in (0,1)$ in a graph $G$ with $n$ vertices and $m$ edges that with probability at least $1 - 1/\poly(n)$ has the following properties:
\begin{itemize}
 \item The algorithm uses $O(n^\alpha)$ space per machine, where $\alpha$ is an arbitrary fixed constant.
 \item The total space is $O(m + n \cdot \poly(\eps^{-1} \log n))$.
 \item The number of rounds is $O(\log (\eps^{-1} \log n))$.
\end{itemize}
\end{theorem}

\begin{proof}
We combine the knowledge developed in this section, and show that \cref{alg:bipartiteness} has the desired properties.

Let $G$ be an input graph on $n$ vertices, and let $G'$ be the graph obtained by performing \cref{line:bipartiteness-sparsify}. By \cref{lem:bipartite_sampling}, $G'$ has $O(n / \eps)$ edges. Let $G'' = (V'', E'')$ be obtained by performing \cref{alg:bipartiteness:reduce_degree}. Then, by \cref{lemma:replace-by-expanders} we have $|V''|, |E''| = \Theta(n / \eps)$ (by Property~\ref{item:KaufmanKR-V'-and-E'}) and $\dmax^{+}(G'') = O(1 / \eps)$(by Property~\ref{item:KaufmanKR-max-deg} and $|E(G')| = O(n / \eps)$). Moreover, by \cref{lem:bipartite_sampling} and Property~\ref{item:KaufmanKR-eps-far} of \cref{lemma:replace-by-expanders}, if $G$ is $\eps$-far from bipartite, then $G''$ is $\Omega(\eps)$-far from bipartite. Also, from Property~\ref{item:KaufmanKR-bipartite} of \cref{lemma:replace-by-expanders} and since $G'$ is a subgraph of $G$, we have that if $G$ is bipartite, then $G''$  is bipartite as well.

Next we apply \cref{lemma:detecting-odd-cycles} on $G''$. Let $n'' = |V(G'')| = \Theta(n/\eps)$. \cref{lemma:detecting-odd-cycles} states that to test whether $G''$ is $\eps$-far from bipartite (and consequently whether $G$ is $\Theta(\eps)$-far from bipartite) it suffices to perform the following: choose the multiset $S$ of $\Theta(n'' \cdot \poly(\eps^{-1} \log{n''}))$ random vertices of $G''$ (chosen with repetition); for each random vertex take two random walks of length $l = \poly(\eps^{-1}\log n'')$; if the endpoints of any pair of random walks collide, then $\reject$, and otherwise $\accept$. Moreover, using that $n'' = \Theta(n/\eps)$, this test succeed with probability at least $1 - (1 - \Omega(n^{-1}\poly(\eps/\log n)))^{\Theta(n \cdot \poly(\log{n} / \eps))} \ge 1 - \poly(1 / n)$ for appropriately chosen constants.

Now we show how to use our algorithms from \cref{sec:random-walks} to generate the required random walks from $G''$. Since vertices in $S$ are chosen independently, by Chernoff bound we have that any vertex $v$ appears $O(\poly(\eps^{-1} \log{n''}))$ times in $S$ with probability $1 - 1/\poly(n)$. This implies that from each vertex we need to generate $\Theta(\poly(\eps^{-1} \log{n''}))$ pairs of random walks. For that, we use \cref{alg:random-walks} with $C = \poly(\eps^{-1} \log{n''})$ to obtain the desired random walks in $O(\log{l}) = O(\log{(\eps^{-1} \log{n})})$ MPC rounds and the total space of $O(n \cdot \poly(\eps^{-1} \log{n}))$ (see \cref{lemma:ref-space-1}), where we used that $n'' = O(n/\eps)$. This completes the analysis.
\end{proof}

\subsection{Additional Application: Finding Cycles in Graphs Far from Cycle-Freeness}

We also note that our algorithm for bipartiteness testing can be used to find cycles in graphs that are far from being cycle free. Czumaj, Goldreich, Ron, Seshadhri, Shapira, and Sohler~\cite{CzumajGRSSS14} observe that the problem of finding such a cycle can be reduced to the problem of one sided bipartiteness testing by replacing each edge of the graph independently with probability $1/2$ with a path of length 2. If the initial graph is far from cycle-freeness, one can show that the modified graph is far from bipartiteness. Our bipartiteness testing algorithm has one-sided error and can be used to find a pair of two short random walks from the same vertex that reveal an odd-length cycle in the modified graph. This cycle can then be mapped to a cycle in the initial graph, by contracting some sub-paths of length 2 back to the corresponding original edge.


\newcommand{\defineT}{\tfrac{20 \cdot \log^3{n}}{\eps^6}}

\section{Testing Expansion}\label{sec:testing-expansion}
In this section we show how to test vertex-expansion of graphs. Our approach (see \cref{alg:expansion}) is inspired by the work of Czumaj and Sohler~\cite{czumaj2010testing} and the work of Kale and Seshadhri~\cite{kale2011expansion}. In \cite{czumaj2010testing,kale2011expansion}, the algorithms simulate many, e.g, $\Theta(\sqrt{n})$, random walks from a small number of randomly chosen vertices.
If we applied our algorithms for sampling random walks directly, we could bound the total space usage by $O(m \sqrt{n})$, which is prohibitive. So, instead, we design an approach in which we sample fewer random walks from each vertex (but much more random walks in total).

For $X, Y \subseteq V$, let $N(X, Y)$ denote the vertex-neighborhood of $X$ within $Y$. That is, $N(X, Y) \eqdef \{v \in Y : \exists u \in X \text{ such that $\{u, v\} \in E$} \}$. 
\begin{definition}
Let $G$ be an undirected graph and $\alpha > 0$. We say that $G$ is an \emph{$\alpha$-vertex-expander} if for every subset $U \subset V$ of size at most $|V| / 2$ we have $|N(U, V \setminus U)| \ge \alpha |U|$.
\end{definition}

\begin{definition}
Let $G$ be a graph of maximum outdegree $d$ and $\eps > 0$.
We say that $G$ is \emph{$\eps$-far from an $\alphastar$-vertex-expander} if one has to change (add/remove) more than $\eps d n$ edges of $G$ to obtain a $\alphastar$-vertex-expander.
\end{definition}

\cref{alg:expansion} gets $\alpha$ as its input, and returns $\accept$ if $G$ is an $\alpha$-vertex-expander, or returns $\reject$ if $G$ is far from being such an expander. The idea of the algorithm is as follows. From each vertex, for $O(\poly \log{n})$ many times we run a pair of random walks. The length of these random walks is set in such a way that if $G$ is an $\alpha$-vertex-expander, then the endpoint of any of these walks is almost uniformly and randomly distributed over $V$. Hence, the endpoints of a pair of random walks from the same vertex are the same with probability very close to $1/n$; if they are the same, we say that these two random walks resulted in a collision. If the number of collisions over all the vertices is significantly larger than $1$, then we conclude that $G$ is not an $\alphastar$-vertex-expander, for some $\alphastar < \alpha$ that we set later. Otherwise the algorithm accepts $G$.
\begin{algorithm}[H]
\begin{algorithmic}[1]
	\State Let $G'$ be the graph obtained from $G$ by adding $2d - \deg^{+}(v)$ self-loops to each vertex $v$. \label{line:transform-G}
	\State $T \gets \defineT$ \label{line:definition-T}
	\State $\ell \gets \tfrac{32 d^2 \ln{(n / \eps)}}{\alpha^2}$ \label{line:definition-ell}
	\For{$i \gets 1 \ldots T$}
		\State Using \cref{alg:random-walks}, generate two random walks of length $\ell$ for each vertex of $G'$.
		\State Let $X_v^i = 1$ if the two random walks originating at $v$ end at the same vertex, and $X_v^i = 0$ otherwise. \label{line:expansion-Xv}
	\EndFor
	\If{$\sum_{i = 1}^T \sum_{v \in V} X_v^i > T + \tfrac{10 \log^2{n}}{\eps^3}$ \label{line:sum-of-collisions}}
		\State \Return \reject
	\Else
		\State \Return \accept
	\EndIf
\end{algorithmic}
\caption{\ExpansionTester$(G, \alpha, \eps)$:
An algorithm that tests whether a given graph $G$ of maximum outdegree $d$ is an $\alpha$-vertex-expander or $G$ is $\eps$-far from any $\alphastar$-expander, for $\alphastar = \tfrac{c \alpha}{d^2 \ln{(n / \eps)}}$, where $c$ is a large enough constant. \label{alg:expansion}}
\end{algorithm}
In the rest of this section, we prove the following.
\begin{theorem}
	Let $G = (V, E)$ be a graph of maximum outdegree $d$. If $G$ is an $\alpha$-vertex-expander then \cref{alg:expansion} outputs $\accept$ with high probability. If $G$ is $\eps$ far from any $\alphastar$-vertex-expander, where $\alphastar = \tfrac{c \cdot \alpha^2}{d^2 \ln{(n / \eps)}}$, then the algorithm outputs $\reject$ also with high probability. Let $\ell \eqdef \tfrac{32 d^2 \ln{(n / \eps)}}{\alpha^2}$. \cref{alg:expansion} can be implemented in $O(\log{\ell})$ MPC rounds, using sublinear space per machine and $O(m \ell \log{n})$ total space.
\end{theorem}
The round and space complexity follows from \cref{lemma:ref-space-2}. The \accept and the \reject cases are proved separately in \cref{lemma:expansion-accept} and \cref{lemma:expansion-reject}, respectively.
Our analysis and the prior work we recall are tailored to regular graphs whose random walks converge to a uniform distribution. Therefore, the algorithm first transform the given graph $G$ into $G'$ by adding self-loops (see \cref{line:transform-G}). Observe that adding self-loops does not affect vertex-cuts, and hence $G'$ and $G$ have the same vertex-expansion. We will use $P_v^l$ to denote the distribution of the endpoints of a random walk of length $l$ originating at $v$.

\subsection{Correctness of Acceptance}
We use the notion of TV distance, that we recall next.
\begin{definition}
Let $p_1, \ldots, p_n$ and $q_1, \ldots, q_n$ be two probability distributions.
	Then, the total variation distance (TVD) between these distributions is equal to $1/2\sum_{i=1}^n |p_i - q_i|$.
\end{definition}
In our proof that \cref{alg:expansion} outputs $\accept$ correctly with probability at least $2/3$, we use the following results.
\begin{lemma}[\cite{goldreich2011testing}]\label{lemma:ee-and-var-of-collisions}
	Let $X_v$ be a random variable that equals $1$ if $2$ random walks of length $l$ starting from vertex $v$ collide, and $X_v = 0$ otherwise. Let $P^l_v$ denote the distribution of the endpoints of these random walks. Then,
	\[
		\ee{X_v} = \ltwo{P^l_v}^2. 
	\]
\end{lemma}
\begin{lemma}[Proposition 2.8 of \cite{goldreich2011testing} and discussion thereafter]\label{lemma:TVD-upper-bound-between-P_v^l-and-U}
	Let $P_v^l$ be the distribution of the endpoints a random walk of length $l$ starting from vertex $v$. Let $G'$ be the graph as defined on \cref{line:transform-G} of \cref{alg:expansion}. If $G'$ is an $\alpha$-vertex-expander, then for $\ell$ as defined on \cref{line:definition-ell} of \cref{alg:expansion} the TVD between $P_v^l$ and the uniform distribution on $n$ vertices is upper-bounded by $\eps/ n$.
\end{lemma}

We will use the following result to upper-bound the $l_2$ norm of the vector $P_v^l$ defined in \cref{lemma:TVD-upper-bound-between-P_v^l-and-U}.
\begin{lemma}\label{lemma:ltwo-of-P_v^l}
Let $Y \in \bbR^n$ be a probability distribution vector. If the TVD between $Y$ and the uniform distribution is at most $\eps / n$, i.e., $\tfrac{1}{2}\sum_{i = 1}^n |Y_i - 1/n| \le \eps / n$, then
	\[
		\ltwo{Y}^2 \le \frac{1 + 4 \eps^2 / n}{n}.
	\]
\end{lemma}
\begin{proof}
Let $y_i = 1/n + \alpha_i$.
We have that $1/2\sum_{i=1}^n |\alpha_i| \leq \eps / n$ and  $\sum_{i=1}^n \alpha_i = 0$.
	\begin{align*}
		\ltwo{Y}^2 & = \sum_{i=1}^n (1/n + \alpha_i)^2 = 1/n + \sum_{i=1}^n 2\cdot1/n\cdot \alpha_i + \sum_{i=1}^n \alpha_i^2 = 1/n + \sum_{i=1}^n \alpha_i^2\\
		& \leq 1/n + \left(\sum_{i=1}^n |\alpha_i|\right)^2  \leq 1/n + 4 \eps^2 / n^2 = (1+4 \eps^2/n)/n
	\end{align*}
\end{proof}

We are now ready to provide the main proof of this section.
\begin{lemma}\label{lemma:expansion-accept}
	If $G$ is an $\alpha$-vertex-expander, then $\ExpansionTester(G, \alpha, \eps)$ returns \accept with probability at least $n^{-2}$.
\end{lemma}
\begin{proof}
	As defined on \cref{line:expansion-Xv} of $\ExpansionTester$, let $X_v^i$ be $0/1$ random variable that equals $1$ iff the two random walks originating at $v$ end at the same vertex. Define $X \eqdef \sum_{i = 1}^T \sum_{v \in V} X_v^i$ that corresponds to the summation of \cref{line:sum-of-collisions} of \cref{alg:expansion}. From \cref{lemma:ee-and-var-of-collisions,lemma:ltwo-of-P_v^l,lemma:TVD-upper-bound-between-P_v^l-and-U} we have
	\begin{equation}\label{eq:upper-bound-on-X}
		\ee{X} = T \sum_{v \in V} \ltwo{P_v^\ell}^2 \le T \cdot (1 + 4 \eps^2 / n) \le T + 1,
	\end{equation}
	for any $4 \eps^2 T \le n$. Also, as $P_v^\ell$ is the probability distribution $n$-dimensional vector, we have $\ltwo{P_v^\ell}^2 \ge 1/n$ and hence
	\begin{equation}\label{eq:lower-bound-on-eeX}
		\ee{X} = T \sum_{v \in V} \ltwo{P_v^\ell}^2 \ge T = \defineT,
	\end{equation}
	where we used the definition of $T$ on \cref{line:definition-T} of \cref{alg:expansion}. Now we can write
	\begin{eqnarray}
		\prob{X \ge T + \frac{10 \log^2{n}}{\eps^3}} & \stackrel{\text{from \eqref{eq:upper-bound-on-X}}}{\le} & \prob{X \ge \ee{X} - 1 + \frac{10 \log^2{n}}{\eps^3}} \nonumber \\
		& \le & \prob{X \ge \rb{1 + \frac{9 \log^2{n}}{\eps^3 \ee{X}}} \ee{X}} \label{eq:probability-on-X-deviating}
	\end{eqnarray}
	From \eqref{eq:lower-bound-on-eeX} we have that $\frac{9 \log^2{n}}{\eps^3 \ee{X}} \le 1$. Observe that across $v$ and $i$ the random variables are $X_v^i$ independent. By applying Chernoff bound (\cref{lemma:chernoff}~\eqref{item:delta-at-most-1-ge}) on \eqref{eq:probability-on-X-deviating} we derive
	\begin{eqnarray*}
		\prob{X \ge T + \frac{10 \log^2{n}}{\eps^3}} & \le & \exp{\rb{-\frac{81 \, \log^4{n}}{3 \eps^6 \ee{X}}}}.
	\end{eqnarray*}
	From $\ee{X} \le T + 1$ (see \cref{eq:upper-bound-on-X}) and for $T = \defineT$, as defined on \cref{line:definition-T} of \cref{alg:expansion}, the last chain of inequalities is upper-bounded by $n^{-1}$. Therefore, $\ExpansionTester$ outputs $\accept$ with high probability, as desired.
\end{proof}

\subsection{Correctness of Rejection}
We use the following result to prove that our algorithm reports $\reject$ properly.
\begin{lemma}[Lemma 4.3 and Lemma 4.7 of~\cite{czumaj2010testing}]\label{lemma:set-U-many-collisions}
	Let $G' = (V, E)$ be a $2d$-regular graph such that each vertex has at least $d$ self-loops. Let $\ell$ be as defined on \cref{line:definition-ell} of \cref{alg:expansion}, and let $\alphastar = \tfrac{c \alpha^2}{d^2 \ln{(n / \eps)}}$, where $c$ is a large enough constant. If $G'$ is $\eps$-far from every $\alphastar$-expander, for $\alphastar \le 1 / 10$, then there exists a subset of vertices $U$ such that:
	\begin{itemize}
		\item $|U| \le \tfrac{\eps}{24}|V|$; and,
		\item $\|P_v^\ell\|_2^2 \ge \frac{1 + 9\eps}{n}$ for each $v \in U$.
	\end{itemize}
\end{lemma}
We are now ready to finalize our analysis.
\begin{lemma}\label{lemma:expansion-reject}
	Let $\eps \in (0, 1/5)$ be a parameter. If $G$ is $\eps$ far from every $\alphastar$-vertex-expander, then $\ExpansionTester(G, \alpha, \eps)$ returns \reject with probability at least $n^{-2}$.
\end{lemma}
\begin{proof}
	For any vertex $v \in V$, as $P_v^\ell$ is a probability distribution, it holds $\ltwo{P_v^\ell}^2 \ge 1/n$. As defined on \cref{line:expansion-Xv} of $\ExpansionTester$, let $X_v^i$ be $0/1$ random variable that equals $1$ iff the two random walks originating at $v$ end at the same vertex. Define $X \eqdef \sum_{i = 1}^T \sum_{v \in V} X_v^i$. Then, if $G$ is $\eps$-far from every $\alphastar$-vertex-expander from \cref{lemma:set-U-many-collisions} we have
	\begin{eqnarray*}
		\ee{X} & = & T \rb{\sum_{v \in V \setminus U} \ltwo{P_v^\ell}^2 + \sum_{v \in U} \ltwo{P_v^\ell}^2} \\
		& \ge & T \rb{\rb{1 - \frac{\eps}{24}}\frac{1}{n} + \frac{\eps}{24} \frac{1 + 9\eps}{n}} \\
		& = & T \rb{\rb{1 - \frac{\eps}{24}}\frac{1}{n} + \frac{\eps}{24} \frac{1 + 9\eps}{n}} \\
		& \ge & T\rb{1 + \frac{\eps^2}{3}}.
	\end{eqnarray*}
	By Chernoff bound and the last inequality it holds
	\begin{equation}\label{eq:lower-bound-on-X}
		\prob{X \le \rb{1 - \sqrt{\frac{6 \log{n}}{T\rb{1 + \frac{\eps^2}{3}}}}} T\rb{1 + \frac{\eps^2}{3}}} \le \prob{X \le \rb{1 - \sqrt{\frac{6 \log{n}}{\ee{X}}}} \ee{X}} \le n^{-2}.
	\end{equation}
	From the definition of $T$ (\cref{line:definition-T} of \cref{alg:expansion}), we have
	\begin{eqnarray*}
		\rb{1 - \sqrt{\frac{6 \log{n}}{T\rb{1 + \frac{\eps^2}{3}}}}} T\rb{1 + \frac{\eps^2}{3}} & \ge & T + 6 \frac{\log^3{n}}{\eps^4} - \sqrt{240 \frac{\log^4{n}}{\eps^6}} \\
		& \ge & T + 6 \frac{\log^3{n}}{\eps^4} - 16 \frac{\log^2{n}}{\eps^3} \\
		& \ge & T + \frac{10 \log^2{n}}{\eps^3},
	\end{eqnarray*}
	for $\eps \le 1/5$. This together with \cref{eq:lower-bound-on-X} and \cref{line:sum-of-collisions} of \cref{alg:expansion} concludes the proof.
	
\end{proof}

\section*{Acknowledgments}
We thank Davin Choo and Julian Portmann for valuable discussions.
S.~Mitrovi{\' c} was supported by the Swiss NSF grant P2ELP2\_181772 and MIT-IBM Watson AI Lab.
P.~Sankowski was supported by the ERC CoG grant TUgbOAT no 772346. 
\appendix
\section{Random Walks in Directed Bounded Degree Graphs}\label{sec:rw-directed}

In this section we show how to efficiently sample short random walks from \emph{directed} graphs, provided that the outdegree of each vertex is bounded.

Let $\dist(u, w)$ be the length of the shortest path from $u$ to $w$.
We define a ball of center $v$ and radius, denoted by $B(v, d)$, to be the set of vertices $x$ of $G$, such that $\dist(v, x) \leq d$.
In particular, $|B(v, 1)|$ contains $v$ and all vertices reachable from $v$ by its outedges.

\begin{observation}\label{obs:ballsize}
Let $G$ be a directed graph and let $\Delta$ be the maximum outdegree in $G$.
Then, for each $v \in V(G)$ and any integer $d \geq 0$ we have $|B(v, d)| = O(\Delta^d)$ and $|G[B(v, d)]| = O(\Delta^{d+1})$.
\end{observation}

Let is first describe the high-level idea.
Assume that the goal is to compute a single random walk of length $\log n$.
We can compute $B(v, d)$ for all $v$ and $d = \epsilon \log n$.
In the next step, for each ball $B(v, d)$ we compute $G[B(v, d)]$.
Then, to find a random walk starting from any vertex $v$, we can compute $\epsilon \log n$ steps of that random walk in a single round on a single machine that knows $G[B(v, d)]$.
Hence, only $O(1/\epsilon)$ steps like this are needed to compute a random walk of length $\log n$.
At the same time, if $\Delta$ is a constant, we only need $O(n^{1+\epsilon})$ space to store all graphs $G[B(v, d)]$.
In the remaining part of this section, we describe the details of this approach.

\begin{algorithm}[H]
\begin{algorithmic}[1]
	\ForAll{$v \in V$ in parallel}\label{l:easy-begin}
	\State{$B(v, 1) := \{v\} \cup \{x \mid vx \in E\}$}
\EndFor

	\State $r := \epsilon/2 \log_{\Delta} n$
	\State $r := 2^{\lfloor \log_2 r \rfloor}$ \Comment{Round down to a power of two}

\For{$i \gets 1 \ldots \log_2 r $}
	\ForAll{$v \in V$ in parallel}
	\State{$B(v, 2^i) := \bigcup_{x \in B(v, 2^{i-1})} B(x, 2^{i-1})$}
	\EndFor
\EndFor

\For{$v \in V$}
	\State{$a^0_{v, v} = a_v$}
\EndFor\label{l:easy-end}
\For{$i \gets 1 \ldots l / r$} \label{l:rw-start} \Comment{For the pseudocode, assume that $r$ divides $l$}
	\ForAll{$t$ such that $\exists_t a^{i-1}_{s, t} \neq 0$ in parallel}
		\State{Compute $G[B(t, r)]$}
		\ForAll{$s$ such that $a^{i-1}_{s, t} \neq 0$ in parallel}
		\State{Use $G[B(t, r)]$ to compute $a^{i-1}_{s, t}$ length $r$ random walks from $t$}
		\State{For each random walks computed in the previous step which ends in $t'$ increase $a^i_{s, t'}$}
		\EndFor
	\EndFor\label{l:rw-end}
\EndFor

\end{algorithmic}
\caption{An algorithm for sampling $a_v$ random walks of length $l$ starting from vertex $v$ (for each $v \in V$).}
\label{alg:random-walks-directed-bounded}
\end{algorithm}

\begin{lemma}
\cref{alg:random-walks-directed-bounded} is correct.
\end{lemma}

\begin{proof}
The first step is to show that the algorithm correctly computes the sets $B(v, 	2^i)$ for all $v$ and $i = 1, \ldots, \log_2 r$.
This follows directly from the fact that $B(v, 2^i) := \bigcup_{x \in B(v, 2^{i-1})} B(x, 2^{i-1})$.

The random walk themselves are computed by the loop in lines~\ref{l:rw-start}-\ref{l:rw-end}.
Each iteration of the loop extends all random walks by $r$ edges.
The algorithm uses variables $a^i_{s, t}$ to represent its state, as follows.
After the $i$-th iteration, the algorithm has computed $a^i_{s, t}$ random walks of length $i\cdot r$, which start in $s$ and end in $t$.
For each $v \in V$ and $i$ the algorithm maintains the invariant that $\sum_{t \in V} a^i_{v, t} = a_v$.

To extend the random walks by $r$ steps, we use the basic fact that a length-$r$ random walk from $v$ is fully contained in $G[B(v, r)]$.
Hence, having $G[B(v, r)]$ is enough to compute $r$ steps of a random walk.
It follows easily that the variables $a^i_{s, t}$ are updated correctly.
Finally, we note that the pseudocode assumes that $r$ divides $l$, but this assumption can be easily dropped by computing only $l \bmod r$ steps of random walk in the last iteration of the main loop.
\end{proof}

\begin{lemma}
Assume that $\sum_{v \in V} a_v = O(n^{1+\epsilon})$. For any $\epsilon > 0$, \cref{alg:random-walks-directed-bounded} can be implemented in MPC model to run in $O(\log\log n + l \log \Delta / (\epsilon \log n))$ rounds, using $O(m + n^{1+\epsilon})$ total space and $O(n^\epsilon)$ space per machine, where $\Delta$ is the maximum degree in the graph and $\log \Delta / \log n = o(1)$.
\end{lemma}

\begin{proof}
Lines~\ref{l:easy-begin}--\ref{l:easy-end} can be implemented in the MPC model in a straightforward way, so in the proof we focus on the remaining part of the algorithm.
Let us first bound the space needed to store $G[B(v, r)]$ for all $v$.
Note that $r \leq \epsilon \log n / (2 \log \Delta)$.
By \cref{obs:ballsize} we have
	\[ |G[B(v, r)]| = O(\Delta^{r+1}) = O(n^{\epsilon / 2 + \log \Delta / \log n}) = O(n^{\epsilon}). \]
Hence, storing $G[B(v, r)]$ for all $v \in V$ requires $O(n^{1+\epsilon})$ total space.

In the $i$-th iteration of the algorithm there are $c_t := \sum_{s \in V} a^{i-1}_{s, t}$ random walks that end in $t$ and need to be extended by using $G[B(t, r)]$.
For each such walk that ends in $t'$ we increase $a_{s,t'}$.
In order to make sure that each machine ends up increasing at most $O(n^\epsilon)$ counters $a_{s,t'}$, we use the following batching strategy.
Recall that for each vertex $t$, we need to compute $c_t$ random walks starting in $t$.
In order to do that we use $\lceil c_t / n^\epsilon \rceil$ machines, each of which computes $G[B(v, r)]$ independently.
The total space used by all the machines is then
	\[
		\sum_{t \in V} \lceil c_t / n^\epsilon \rceil O(n^\epsilon) \leq \sum_{t \in V} (c_t / n^\epsilon + 1) O(n^\epsilon) = O(n^{1+\epsilon}) + \sum_{t \in V} c_t = O(n^{1+\epsilon}) + \sum_{v \in V}a_v = O(n^{1+\epsilon})
	\]
\end{proof}

By combining the lemmas from this section we obtain the following.

\begin{theorem}
Let $l > 0$, $G$ be a directed graph with maximum outdegree bounded by $\Delta$ and let $\{ a_v \}_{v\in V}$ be a sequence such that $\sum_{v\in V} a_v = O(n^{1+\epsilon})$.
There exists an MPC algorithm that for each $v \in V$ computes $a_v$ endpoints of random walks of length $l$ starting from $a_v$.
The algorithm uses $O(\log\log n + l \log \Delta / (\epsilon \log n))$ rounds, $O(m + n^{1+\epsilon})$ total space and $O(n^\epsilon)$ space per machine.
\end{theorem}

Note that with $l = O(\log n)$ and $\Delta = \poly \log n$ we get $O(\log \log n)$ rounds.
Moreover, if $\Delta = O(1)$, we can set $\epsilon = 1/\log \log n$ and use only $O(m + n^{1+o(1)})$ total space.

\section{Handling Dangling Vertices}
\label{app:handling-dangling-nodes}
From theoretical point of view, when dangling vertices are present the transition matrix
for the PageRank walks is not stochastic. In particular, the largest eigenvector is no longer
equal to 1. Hence, many ways to handle dangling vertices have been proposed, see e.g.,~\cite{berkhin2005}.
For example, we can delete them, we can lump them to one vertex and add self-loop, we can add a self-loop to them, each dangling
vertex can be linked to an artificial vertex with a self-loop (sink) or we can connect each dangling
vertex to every other vertex. The first solution is often mentioned even in the original PageRank paper, however,
it is infeasible in our case, as it requires finding strongly connected components first.
This last solution seems to be the most
accepted and the most widely used one, as it can be interpreted as restarting the random walk
from a random state if we reach a dangling vertex. In this paper, we assume that
one of the above solutions has been already applied to our graph and, therefore, every vertex has at least one outgoing edge.
Still, in this section, we give a novel relation between two most widely applied methods to handle dangling vertices.
We have made a through literature study and to the best of our knowledge these relations have not been
observed before. Usually, one heuristically argues that these approaches give similar results.
We give formal proofs that it is indeed the case. We show that adding self-loops and restarting the walks 
are equivalent up to a simple transformation. 

In the case when dangling vertices are present the transition matrix of our graph $T=T(G)$ can be decomposed
into the following blocks
\[
T=\left[\begin{matrix}
T_1 & 0 \\
T_2 & 0
\end{matrix}
\right],
\]
where the bottom $k$ rows correspond to dangling vertices without any out edges. Adding self-loops 
gives us the following transition matrix 
\[
T^s=\left[\begin{matrix}
T_1 & 0 \\
T_2 & I
\end{matrix}\right].
\]
Note that the way we defined transition matrices $T$ and $T^s$ means that the entry at row $u$ and column $v$ corresponds to \emph{incoming} arc to $u$ from $v$. Hence, a stationary distribution of $T$ is a right eigenvector of $T$.

The PageRank matrix of $T^s$ is given by
\begin{equation}
\label{eq:pagerank-s}
(1-\eps)T^s + \frac{\eps}{n} J =\left[\begin{matrix}
(1-\eps) T_1  + \frac{\eps}{n}J_{n-k,n-k} & \frac{\eps}{n} J_{n-k,k} \\
(1-\eps) T_2 +  \frac{\eps}{n}J_{k,n-k}  & (1-\eps)I +  \frac{\eps}{n}J_{k,k} 
\end{matrix}\right],
\end{equation}
where $J_{i,j}$ refers to an all $1$ matrix of size $i\times j$.
We denote by $\pi^s=\left[\begin{matrix}\pi^s_1\\ \pi^s_2\end{matrix}\right]$ the stationary distribution of the above PageRank matrix. 
When we consider restarting the walks we obtain the following matrix
\[
T^r=\left[\begin{matrix}
T_1 & \frac{1}{n}J_{n-k,k} \\
T_2 & \frac{1}{n}J_{k,k}
\end{matrix}\right].
\]
The PageRank matrix of $T^r$ is given by
\begin{equation}
\label{eq:pagerank-r}
(1-\eps)T^r + \frac{\eps}{n} J =\left[\begin{matrix}
(1-\eps) T_1  + \frac{\eps}{n}J_{n-k,n-k} & \frac{1}{n} J_{n-k,k} \\
(1-\eps) T_2 +  \frac{\eps}{n}J_{k,n-k}  &   \frac{1}{n}J_{k,k}
\end{matrix}\right],
\end{equation}
with $\pi^r=\left[\begin{matrix}\pi^r_1\\ \pi^r_2\end{matrix}\right]$ being the stationary distribution. In the rest of this section, given a vector $x \in \bbR^n$, we use $|x|$ to denote that $l_1$ norm of $x$.

\begin{theorem}
\label{thm:dangling-nodes}
Let $\pi^r$ and $\pi^s$ be the stationary distribution vectors as defined above. Then, it holds
\[
\pi^r_1 = \frac{1}{\eps - \eps|\pi_1^s| +|\pi_1^s|} \pi^s_1
\] 
and 
\[
\pi^r_2=  \frac{1}{1/\eps - |\pi_2^s|/\eps + |\pi_2^s|} \pi^s_2.
\]
\end{theorem}
\begin{proof}
Let us first consider the upper blocks of \eqref{eq:pagerank-s} when used in the stationary equation. We obtain
\[
\rb{(1-\eps) T_1 + \frac{\eps}{n}J} \pi^s_1 + \frac{\eps}{n} |\pi_2^s| \vec{1} = \pi_1^s
\]
\[
\rb{I-(1-\eps)T_1 -\frac{\eps}{n} J} \pi_1^s = \frac{\eps}{n} |\pi_2^s| \vec{1}
\]
\begin{equation}
\label{eq:pi1s}
\pi_1^s=  \frac{\eps}{n} |\pi_2^s| \rb{I-(1-\eps)T_1 -\frac{\eps}{n} J}^{-1}\vec{1} = \frac{\eps}{n} |\pi_2^s| g,  
\end{equation}
where $g=  \rb{I-(1-\eps)T_1 -\frac{\eps}{n} J}^{-1}\vec{1}$. 

Now, consider the upper blocks of \eqref{eq:pagerank-r}. Also from the stationary equation we derive
\[
\rb{(1-\eps) T_1 + \frac{\eps}{n}J} \pi^r_1 + \frac{1}{n} |\pi_2^r| \vec{1} = \pi_1^r.
\]
\[
\rb{I-(1-\eps)T_1 -\frac{\eps}{n} J} \pi_1^r = \frac{1}{n} |\pi_2^r| \vec{1}.
\]
\begin{equation}
\label{eq:pi1r}
\pi_1^r = \frac{1}{n} |\pi_2^r|\rb{I-(1-\eps)T_1 -\frac{\eps}{n} J}^{-1} \vec{1} = \frac{1}{n} |\pi_2^r| g.
\end{equation}
From \eqref{eq:pi1s} with \eqref{eq:pi1r} we conclude that both $\pi_1^r$ and $\pi_1^s$ are parallel to $g$, and hence $\pi^r_1 = x \pi^s_1$, for some $x\in \bbR$. 

Next we consider lower blocks of \eqref{eq:pagerank-s} and \eqref{eq:pagerank-r}. From the lower blocks of \eqref{eq:pagerank-s} we establish
\[
\rb{(1-\eps) T_2 + \frac{\eps}{n}J} \pi^s_1 + (1-\eps)\pi_2^s+  \frac{\eps}{n} |\pi_2^s| \vec{1} = \pi_1^s
\]
\[
\rb{(1-\eps) T_2 + \frac{\eps}{n}J} \pi^s_1 +  \frac{\eps}{n} |\pi_2^s| \vec{1} = \eps \pi_2^s.
\]
By plugging the last equality into \eqref{eq:pi1s} we get
\[
\rb{(1-\eps) T_2 + \frac{\eps}{n}J} \frac{\eps}{n} |\pi_2^s| g +  \frac{\eps}{n} |\pi_2^s| \vec{1} = \eps \pi_2^s
\]
\[
\rb{(1-\eps) T_2 + \frac{\eps}{n}J} \frac{1}{n} |\pi_2^s| g +  \frac{1}{n} |\pi_2^s| \vec{1} = \pi_2^s
\]
\begin{equation}
\label{eq:pi2s}
|\pi_2^s|\rb{\rb{(1-\eps) T_2 + \frac{\eps}{n}J} \frac{1}{n}  g +  \frac{1}{n}  \vec{1}} = \pi_2^s.
\end{equation}
The lower blocks of \eqref{eq:pagerank-r} give 
\[
\rb{(1-\eps) T_2 + \frac{\eps}{n}J} \pi^r_1 + \frac{1}{n} |\pi_2^r| \vec{1} = \pi_2^r.
\]
Plugging the lest equality into \eqref{eq:pi1r} leads to
\[
\rb{(1-\eps) T_2 + \frac{\eps}{n}J} \frac{1}{n} |\pi_2^r| g + \frac{1}{n} |\pi_2^r| \vec{1} = \pi_2^r
\]
\begin{equation}
\label{eq:pi2r}
|\pi_2^r|\rb{\rb{(1-\eps) T_2 + \frac{\eps}{n}J} \frac{1}{n} g + \frac{1}{n}  \vec{1}} = \pi_2^r.
\end{equation}
Again by \eqref{eq:pi2s} and \eqref{eq:pi2r} we see that $\pi_2^r =y \pi_2^s$, for some $y\in \bbR$. 
We have that $|\pi_1^s| = 1-|\pi_2^s|$ and $|\pi_1^r| = 1-|\pi_2^r|$, so from \eqref{eq:pi1s} and \eqref{eq:pi1r}
we obtain 
\begin{eqnarray*}
1-|\pi_2^s| = \frac{\eps}{n} |\pi_2^s| |g|  &\textrm{ and }& 1-|\pi_2^r| = \frac{1}{n} |\pi_2^r| |g|,
\end{eqnarray*}
which implies
\begin{eqnarray*}
1-|\pi_2^s| = \frac{\eps}{n} |\pi_2^s| |g|  &\textrm{ and }& 1-y |\pi_2^s| = \frac{1}{n} y |\pi_2^s| |g|.
\end{eqnarray*}
By solving these two equations for $y$ we obtain 
\[
y = \frac{1}{1/\eps - |\pi_2^s|/\eps + |\pi_2^s|}.
\]
Using the same approach to $x$ we get 
\[
x= \frac{1}{\eps - \eps|\pi_1^s| +|\pi_1^s|}, 
\]
what finishes the proof. 
\end{proof}
\cref{thm:dangling-nodes} has a few consequences. First of all, one can easily see that scores of vertices in $\pi_1^r$ are all higher than
scores of vertices in $\pi_1^s$. More importantly, we can easily obtain $\pi^r$ in $O(1)$ MPC rounds from $\pi^s$. We
note that $\pi^s$ are much easier to compute as obtaining them requires only minor graph modification, i.e., adding self-loops to dangling vertices. 
Whereas, the scores $\pi^r$ are the most widely accepted ones. 
%
%
%
%
%
%
%


\bibliographystyle{alpha}
\bibliography{bibliography}

\newcommand{\etalchar}[1]{$^{#1}$}
\begin{thebibliography}{PBMW99}

\bibitem[ABB{\etalchar{+}}19]{assadi2019coresets}
Sepehr Assadi, MohammadHossein Bateni, Aaron Bernstein, Vahab Mirrokni, and
  Cliff Stein.
\newblock Coresets meet {EDCS}: algorithms for matching and vertex cover on
  massive graphs.
\newblock In {\em Proceedings of the Thirtieth Annual ACM-SIAM Symposium on
  Discrete Algorithms}, pages 1616--1635. SIAM, 2019.

\bibitem[ACK19]{assadi2019sublinear}
Sepehr Assadi, Yu~Chen, and Sanjeev Khanna.
\newblock Sublinear algorithms for ({$\Delta+1$}) vertex coloring.
\newblock In {\em Proceedings of the Thirtieth Annual ACM-SIAM Symposium on
  Discrete Algorithms}, pages 767--786. Society for Industrial and Applied
  Mathematics, 2019.

\bibitem[ACL06]{AndersenCL06}
Reid Andersen, Fan R.~K. Chung, and Kevin~J. Lang.
\newblock Local graph partitioning using {PageRank} vectors.
\newblock In {\em 47th Annual {IEEE} Symposium on Foundations of Computer
  Science {(FOCS} 2006), 21-24 October 2006, Berkeley, California, USA,
  Proceedings}, pages 475--486, 2006.

\bibitem[ALNO07]{avrachenkov2007monte}
Konstantin Avrachenkov, Nelly Litvak, Danil Nemirovsky, and Natalia Osipova.
\newblock {Monte {Carlo} methods in {PageRank} computation: When one iteration
  is sufficient}.
\newblock {\em SIAM Journal on Numerical Analysis}, 45(2):890--904, 2007.

\bibitem[ANOY14]{AndoniNOY14}
Alexandr Andoni, Aleksandar Nikolov, Krzysztof Onak, and Grigory Yaroslavtsev.
\newblock Parallel algorithms for geometric graph problems.
\newblock In {\em Proceedings of the 46th ACM Symposium on Theory of Computing,
  {STOC} 2014, New York, NY, USA, May 31--June 3, 2014}, pages 574--583, 2014.

\bibitem[ASS{\etalchar{+}}18]{log-diameter}
Alexandr Andoni, Zhao Song, Clifford Stein, Zhengyu Wang, and Peilin Zhong.
\newblock Parallel graph connectivity in log diameter rounds.
\newblock In {\em 2018 IEEE 59th Annual Symposium on Foundations of Computer
  Science (FOCS)}, pages 674--685. IEEE, 2018.

\bibitem[ASW19]{ASW}
Sepehr Assadi, Xiaorui Sun, and Omri Weinstein.
\newblock Massively parallel algorithms for finding well-connected components
  in sparse graphs.
\newblock In {\em Proceedings of the 2019 {ACM} Symposium on Principles of
  Distributed Computing, {PODC} 2019, Toronto, ON, Canada, July 29 - August 2,
  2019.}, pages 461--470, 2019.

\bibitem[BBCT12]{borgs2012sublinear}
Christian Borgs, Michael Brautbar, Jennifer Chayes, and Shang-Hua Teng.
\newblock A sublinear time algorithm for {PageRank} computations.
\newblock In {\em International Workshop on Algorithms and Models for the
  Web-Graph}, pages 41--53. Springer, 2012.

\bibitem[BCX11]{BahmaniCX11}
Bahman Bahmani, Kaushik Chakrabarti, and Dong Xin.
\newblock Fast personalized {PageRank} on {MapReduce}.
\newblock In {\em Proceedings of the {ACM} {SIGMOD} International Conference on
  Management of Data, {SIGMOD} 2011, Athens, Greece, June 12-16, 2011}, pages
  973--984, 2011.

\bibitem[BDE{\etalchar{+}}19]{BDELM2019}
Soheil Behnezhad, Laxman Dhulipala, Hossein Esfandiari, Jakub Łącki, and
  Vahab Mirrokni.
\newblock Near-optimal massively parallel graph connectivity.
\newblock {\em FOCS}, 2019.

\bibitem[Ber05a]{berkhin2005survey}
Pavel Berkhin.
\newblock A survey on {PageRank} computing.
\newblock {\em Internet Mathematics}, 2(1):73--120, 2005.

\bibitem[Ber05b]{berkhin2005}
Pavel Berkhin.
\newblock A survey on pagerank computing.
\newblock {\em Internet Math.}, 2(1):73--120, 2005.

\bibitem[BFU18]{brandt2018matching}
Sebastian Brandt, Manuela Fischer, and Jara Uitto.
\newblock Matching and {MIS} for uniformly sparse graphs in the low-memory
  {MPC} model.
\newblock {\em arXiv preprint arXiv:1807.05374}, 2018.

\bibitem[BHH19]{behnezhad2019exponentially}
Soheil Behnezhad, MohammadTaghi Hajiaghayi, and David~G Harris.
\newblock Exponentially faster massively parallel maximal matching.
\newblock {\em FOCS}, 2019.

\bibitem[BKS13]{BeameKS13}
Paul Beame, Paraschos Koutris, and Dan Suciu.
\newblock Communication steps for parallel query processing.
\newblock In {\em Proceedings of the 32nd {ACM} {SIGMOD-SIGACT-SIGART}
  Symposium on Principles of Database Systems, {PODS} 2013, New York, NY,
  {USA}, June 22--27, 2013}, pages 273--284, 2013.

\bibitem[BP98]{brin1998anatomy}
Sergey Brin and Lawrence Page.
\newblock The anatomy of a large-scale hypertextual web search engine.
\newblock {\em Computer networks and ISDN systems}, 30(1-7):107--117, 1998.

\bibitem[BP11]{bressan2011local}
Marco Bressan and Luca Pretto.
\newblock Local computation of {PageRank}: the ranking side.
\newblock In {\em Proceedings of the 20th ACM international conference on
  Information and knowledge management}, pages 631--640. ACM, 2011.

\bibitem[BPP18]{BPP18}
Marco Bressan, Enoch Peserico, and Luca Pretto.
\newblock Sublinear algorithms for local graph centrality estimation.
\newblock In {\em 59th {IEEE} Annual Symposium on Foundations of Computer
  Science, {FOCS} 2018, Paris, France, October 7-9, 2018}, pages 709--718,
  2018.

\bibitem[Bre02]{Breyer02markovianpage}
LA~Breyer.
\newblock Markovian page ranking distributions: some theory and simulations.
\newblock 2002.

\bibitem[CFSV16]{CFSV16}
Keren Censor{-}Hillel, Eldar Fischer, Gregory Schwartzman, and Yadu Vasudev.
\newblock Fast distributed algorithms for testing graph properties.
\newblock In {\em Distributed Computing - 30th International Symposium, {DISC}
  2016, Paris, France, September 27-29, 2016. Proceedings}, pages 43--56, 2016.

\bibitem[CGR{\etalchar{+}}14]{CzumajGRSSS14}
Artur Czumaj, Oded Goldreich, Dana Ron, C.~Seshadhri, Asaf Shapira, and
  Christian Sohler.
\newblock Finding cycles and trees in sublinear time.
\newblock {\em Random Struct. Algorithms}, 45(2):139--184, 2014.

\bibitem[CGS04]{chen2004local}
Yen-Yu Chen, Qingqing Gan, and Torsten Suel.
\newblock Local methods for estimating {PageRank} values.
\newblock In {\em Proceedings of the thirteenth ACM international conference on
  Information and knowledge management}, pages 381--389. ACM, 2004.

\bibitem[CKK{\etalchar{+}}18]{8555132}
A.~{Chiplunkar}, M.~{Kapralov}, S.~{Khanna}, A.~{Mousavifar}, and Y.~{Peres}.
\newblock Testing graph clusterability: Algorithms and lower bounds.
\newblock In {\em 2018 IEEE 59th Annual Symposium on Foundations of Computer
  Science (FOCS)}, pages 497--508, Oct 2018.

\bibitem[C{\L}M{\etalchar{+}}18]{czumaj2018round}
Artur Czumaj, Jakub {\L}{\k{a}}cki, Aleksander M{\k{a}}dry, Slobodan
  Mitrovi{\'c}, Krzysztof Onak, and Piotr Sankowski.
\newblock Round compression for parallel matching algorithms.
\newblock In {\em Proceedings of the 50th Annual ACM SIGACT Symposium on Theory
  of Computing}, pages 471--484. ACM, 2018.

\bibitem[CMOS19]{CzumajMOS}
Artur Czumaj, Morteza Monemizadeh, Krzysztof Onak, and Christian Sohler.
\newblock Planar graphs: Random walks and bipartiteness testing.
\newblock {\em Random Structures \& Algorithms}, 2019.

\bibitem[Col88]{cole1988parallel}
Richard Cole.
\newblock Parallel merge sort.
\newblock {\em SIAM Journal on Computing}, 17(4):770--785, 1988.

\bibitem[CPS15]{Czumaj:2015}
Artur Czumaj, Pan Peng, and Christian Sohler.
\newblock Testing cluster structure of graphs.
\newblock In {\em Proceedings of the Forty-seventh Annual ACM Symposium on
  Theory of Computing}, STOC '15, pages 723--732, New York, NY, USA, 2015. ACM.

\bibitem[CS10]{czumaj2010testing}
Artur Czumaj and Christian Sohler.
\newblock Testing expansion in bounded-degree graphs.
\newblock {\em Combinatorics, Probability and Computing}, 19(5-6):693--709,
  2010.

\bibitem[DCGR05]{delcorso2005}
Gianna~M. Del~Corso, Antonio Gullí, and Francesco Romani.
\newblock Fast pagerank computation via a sparse linear system.
\newblock {\em Internet Math.}, 2(3):251--273, 2005.

\bibitem[DGP11]{SarmaGP11}
Atish {Das Sarma}, Sreenivas Gollapudi, and Rina Panigrahy.
\newblock Estimating {PageRank} on graph streams.
\newblock {\em J. {ACM}}, 58(3):13:1--13:19, 2011.

\bibitem[DMPU15]{2015:FDP}
Atish {Das Sarma}, Anisur~Rahaman Molla, Gopal Pandurangan, and Eli Upfal.
\newblock Fast distributed {PageRank} computation.
\newblock {\em Theor. Comput. Sci.}, 561(PB):113--121, January 2015.

\bibitem[DNPT13]{SarmaNPT13}
Atish {Das Sarma}, Danupon Nanongkai, Gopal Pandurangan, and Prasad Tetali.
\newblock Distributed random walks.
\newblock {\em J. {ACM}}, 60(1):2:1--2:31, 2013.

\bibitem[DSB09]{duhan2009page}
Neelam Duhan, AK~Sharma, and Komal~Kumar Bhatia.
\newblock Page ranking algorithms: a survey.
\newblock In {\em 2009 IEEE International Advance Computing Conference}, pages
  1530--1537. IEEE, 2009.

\bibitem[GGK{\etalchar{+}}18]{ghaffari2018improved}
Mohsen Ghaffari, Themis Gouleakis, Christian Konrad, Slobodan Mitrovi{\' c},
  and Ronitt Rubinfeld.
\newblock Improved massively parallel computation algorithms for {MIS},
  matching, and vertex cover.
\newblock {\em Proceedings of the 37th ACM Principles of Distributed Computing
  (PODC 2018)}, 2018.

\bibitem[GKK13]{GoelKK13}
Ashish Goel, Michael Kapralov, and Sanjeev Khanna.
\newblock Perfect matchings in {$O(n \log n)$} time in regular bipartite
  graphs.
\newblock {\em {SIAM} J. Comput.}, 42(3):1392--1404, 2013.

\bibitem[GKMS18]{gamlath2018weighted}
Buddhima Gamlath, Sagar Kale, Slobodan Mitrovi{\'c}, and Ola Svensson.
\newblock Weighted matchings via unweighted augmentations.
\newblock {\em arXiv preprint arXiv:1811.02760}, 2018.

\bibitem[GKU19]{ghaffari2019conditional}
Mohsen Ghaffari, Fabian Kuhn, and Jara Uitto.
\newblock Conditional hardness results for massively parallel computation from
  distributed lower bounds.
\newblock {\em FOCS}, 2019.

\bibitem[GLM19]{ghaffari2019improved}
Mohsen Ghaffari, Silvio Lattanzi, and Slobodan Mitrovi{\'c}.
\newblock Improved parallel algorithms for density-based network clustering.
\newblock In {\em International Conference on Machine Learning}, pages
  2201--2210, 2019.

\bibitem[GR99]{GoldreichR99}
Oded Goldreich and Dana Ron.
\newblock A sublinear bipartiteness tester for bounded degree graphs.
\newblock {\em Combinatorica}, 19(3):335--373, 1999.

\bibitem[GR11]{goldreich2011testing}
Oded Goldreich and Dana Ron.
\newblock On testing expansion in bounded-degree graphs.
\newblock In {\em Studies in Complexity and Cryptography. Miscellanea on the
  Interplay between Randomness and Computation}, pages 68--75. Springer, 2011.

\bibitem[GSZ11]{goodrich2011sorting}
Michael~T. Goodrich, Nodari Sitchinava, and Qin Zhang.
\newblock Sorting, searching, and simulation in the {MapReduce} framework.
\newblock In {\em International Symposium on Algorithms and Computation}, pages
  374--383. Springer, 2011.

\bibitem[GU19]{ghaffari2019sparsifying}
Mohsen Ghaffari and Jara Uitto.
\newblock Sparsifying distributed algorithms with ramifications in massively
  parallel computation and centralized local computation.
\newblock In {\em Proceedings of the Thirtieth Annual ACM-SIAM Symposium on
  Discrete Algorithms}, pages 1636--1653. SIAM, 2019.

\bibitem[HLL18]{harvey2018greedy}
Nicholas~JA Harvey, Christopher Liaw, and Paul Liu.
\newblock Greedy and local ratio algorithms in the {MapReduce} model.
\newblock In {\em Proceedings of the 30th on Symposium on Parallelism in
  Algorithms and Architectures}, pages 43--52. ACM, 2018.

\bibitem[HZ96a]{DBLP:journals/jcss/HalperinZ96}
Shay Halperin and Uri Zwick.
\newblock An optimal randomised logarithmic time connectivity algorithm for the
  {EREW} {PRAM}.
\newblock {\em J. Comput. Syst. Sci.}, 53(3):395--416, 1996.

\bibitem[HZ96b]{Halperin:1996}
Shay Halperin and Uri Zwick.
\newblock Optimal randomized {EREW} {PRAM} algorithms for finding spanning
  forests and for other basic graph connectivity problems.
\newblock In {\em Proceedings of the Seventh Annual ACM-SIAM Symposium on
  Discrete Algorithms}, SODA '96, pages 438--447, Philadelphia, PA, USA, 1996.
  Society for Industrial and Applied Mathematics.

\bibitem[Jin19]{Jin19}
Ce~Jin.
\newblock Simulating random walks on graphs in the streaming model.
\newblock In {\em 10th Innovations in Theoretical Computer Science Conference,
  {ITCS} 2019, January 10-12, 2019, San Diego, California, {USA}}, pages
  46:1--46:15, 2019.

\bibitem[JS96]{Jerrum:1996:MCM}
Mark Jerrum and Alistair Sinclair.
\newblock The {Markov} chain {Monte} {Carlo} method: an approach to approximate
  counting and integration.
\newblock {\em Approximation algorithms for NP-hard problems}, pages 482--520,
  1996.

\bibitem[KKR04]{KaufmanKR04}
Tali Kaufman, Michael Krivelevich, and Dana Ron.
\newblock Tight bounds for testing bipartiteness in general graphs.
\newblock {\em {SIAM} J. Comput.}, 33(6):1441--1483, 2004.

\bibitem[KM09]{KelnerM09}
Jonathan~A. Kelner and Aleksander Mądry.
\newblock Faster generation of random spanning trees.
\newblock In {\em 50th Annual {IEEE} Symposium on Foundations of Computer
  Science, {FOCS} 2009, October 25-27, 2009, Atlanta, Georgia, {USA}}, pages
  13--21, 2009.

\bibitem[KNP99]{DBLP:journals/siamcomp/KargerNP99}
David~R. Karger, Noam Nisan, and Michal Parnas.
\newblock Fast connected components algorithms for the {EREW} {PRAM}.
\newblock {\em {SIAM} J. Comput.}, 28(3):1021--1034, 1999.

\bibitem[KS11]{kale2011expansion}
Satyen Kale and Comandur Seshadhri.
\newblock An expansion tester for bounded degree graphs.
\newblock {\em SIAM Journal on Computing}, 40(3):709--720, 2011.

\bibitem[KSV10]{KarloffSV10}
Howard~J. Karloff, Siddharth Suri, and Sergei Vassilvitskii.
\newblock A model of computation for {MapReduce}.
\newblock In {\em Proceedings of the 21st Annual {ACM-SIAM} Symposium on
  Discrete Algorithms, {SODA} 2010, Austin, Texas, USA, January 17--19, 2010},
  pages 938--948, 2010.

\bibitem[LF80]{ladner1980parallel}
Richard~E Ladner and Michael~J Fischer.
\newblock Parallel prefix computation.
\newblock {\em Journal of the ACM (JACM)}, 27(4):831--838, 1980.

\bibitem[LM04]{langville2004deeper}
Amy~N Langville and Carl~D Meyer.
\newblock Deeper inside {PageRank}.
\newblock {\em Internet Mathematics}, 1(3):335--380, 2004.

\bibitem[LMSV11]{lattanzi2011filtering}
Silvio Lattanzi, Benjamin Moseley, Siddharth Suri, and Sergei Vassilvitskii.
\newblock Filtering: a method for solving graph problems in {MapReduce}.
\newblock In {\em Proceedings of the twenty-third annual ACM symposium on
  Parallelism in algorithms and architectures}, pages 85--94. ACM, 2011.

\bibitem[LPS86]{expander2}
Alexander Lubotzky, Ralph Phillips, and Peter Sarnak.
\newblock Explicit expanders and the {Ramanujan} conjectures.
\newblock In {\em Proceedings of the 18th Annual {ACM} Symposium on Theory of
  Computing, May 28-30, 1986, Berkeley, California, {USA}}, pages 240--246,
  1986.

\bibitem[Mar73]{expander1}
Gregory~A. Margulis.
\newblock Explicit constructions of expanders.
\newblock {\em Problemy Peredachi Informatsii}, 9(4):71--80, 1973.

\bibitem[NS10]{NachmiasS10}
Asaf Nachmias and Asaf Shapira.
\newblock Testing the expansion of a graph.
\newblock {\em Inf. Comput.}, 208(4):309--314, 2010.

\bibitem[Ona18]{onak2018round}
Krzysztof Onak.
\newblock Round compression for parallel graph algorithms in strongly sublinear
  space.
\newblock {\em arXiv preprint arXiv:1807.08745}, 2018.

\bibitem[PBMW99]{pagerank}
Lawrence Page, Sergey Brin, Rajeev Motwani, and Terry Winograd.
\newblock The {PageRank} citation ranking: Bringing order to the web.
\newblock Technical Report 1999-66, Stanford InfoLab, November 1999.

\bibitem[Rei85]{Reif:1985}
John~H. Reif.
\newblock An optimal parallel algorithm for integer sorting.
\newblock In {\em Proceedings of the 26th Annual Symposium on Foundations of
  Computer Science}, SFCS '85, pages 496--504, Washington, DC, USA, 1985. IEEE
  Computer Society.

\bibitem[RVW18]{roughgarden2018shuffles}
Tim Roughgarden, Sergei Vassilvitskii, and Joshua~R Wang.
\newblock Shuffles and circuits (on lower bounds for modern parallel
  computation).
\newblock {\em Journal of the ACM (JACM)}, 65(6):41, 2018.

\end{thebibliography}

\end{document}